\documentclass[11pt,a4paper]{article}
\usepackage[table]{xcolor}
\usepackage{a4wide}
\usepackage{filecontents}

\usepackage{amsmath,amsthm,amssymb,enumerate}
\usepackage{listings}
\usepackage{longtable}
\usepackage{thm-restate}

\usepackage{tikz}
\usetikzlibrary{calc}
\usetikzlibrary{shapes.symbols}
\usetikzlibrary{decorations.pathreplacing}
\definecolor{lightgray}{gray}{0.9}

\usepackage[bookmarks=true,hypertexnames=false,pagebackref]{hyperref}
\hypersetup{colorlinks=true, citecolor=blue, linkcolor=blue, urlcolor=blue} 
\newtheorem{lemma}{Lemma}

\newtheorem{theorem}[lemma]{Theorem}
\newtheorem{cor}[lemma]{Corollary}
\newtheorem{obs}[lemma]{Observation}
\theoremstyle{definition}\newtheorem{defn}[lemma]{Definition}
\newcommand{\floor}[1]{\lfloor #1 \rfloor}

\newcommand{\calF}{\mathcal{F}}
\newcommand{\calG}{\mathcal{G}}
\newcommand{\calH}{\mathcal{H}}
\newcommand{\calI}{\mathcal{I}}

\newcommand{\calT}{\mathcal{T}}

\newcommand{\bx}{\mathbf{x}}

\newcommand{\sawt}[2]{T_\mathsf{SAW}(#1,#2)}  

\newcommand{\ourfamily}{\mathcal{D}}

\newcommand{\wout}{w_{-}}
\newcommand{\win}{w_{+}}

\newcommand{\posints}{\mathbb{Z}_{>0}}
\newcommand{\posreals}{\mathbb{R}_{>0}}

\newcommand{\aad}{\textnormal{aad}}

\newcommand{\biprepot}{bi-potential.csv}
\newcommand{\genprepot}{potential.csv}
\newcommand{\validatecode}{validator.nb}

\newcommand{\probIS}{\#\textsc{IS}}
\newcommand{\probBIS}{\#\textsc{BIS}}

\newcommand{\IS}{Z}
\newcommand{\Z}{\mathbb{Z}}
\newcommand{\R}{\mathbb{R}}
\newcommand{\eps}{\varepsilon}
\renewcommand{\epsilon}{\varepsilon}
\newcommand{\Gout}{\calG_\mathsf{out}}
\newcommand{\Gin}{\calG_\mathsf{in}}
\newcommand{\rGout}{G_\mathsf{out}}
\newcommand{\rGin}{G_\mathsf{in}}

\newcommand{\Wout}{W'_\mathsf{out}}
\newcommand{\Win}{W'_\mathsf{in}}
\newcommand{\newWout}{W_\mathsf{out}}
\newcommand{\newWin}{W_\mathsf{in}}
\newcommand{\Cout}{C_\mathsf{out}}
\newcommand{\Cin}{C_\mathsf{in}}

\newcommand{\tGin}{G_\mathsf{in}}
\newcommand{\tGout}{G_\mathsf{out}}

\newcommand{\saad}{\textnormal{aad}^*}
\newcommand{\poly}{\textnormal{poly}}

\newcommand{\hcalgo}{\textnormal{\texttt{ApproxZ}}}
\newcommand{\hcunialgo}{\textnormal{\texttt{ApproxZUni}}}
\newcommand{\basecount}{\textnormal{\texttt{BaseCount}}}
\newcommand{\brute}{\textnormal{\texttt{ExactCount}}}

\newcommand{\iscount}{\textnormal{\texttt{Count}}}
\newcommand{\TR}{\textnormal{\texttt{TR}}}
\newcommand{\Prune}{\textnormal{\texttt{Prune}}}
\newcommand{\Red}{\textnormal{\texttt{Reduce}}}

\def\Stefankovic{\v{S}tefankovi\v{c}}

\newenvironment{algorr}[2]{%
	\smallskip
	\noindent\textbf{Algorithm #1:} {\itshape#2}
	\begin{enumerate}[(i)]%
		\smallskip
		\setlength{\itemsep}{0pt}\leftmargin=0pt}{%
	\end{enumerate}%
}

\title{Faster Exponential-time Algorithms for Approximately Counting Independent Sets\thanks{The research leading to these results has received funding from the European Research Council under the European Union's Seventh Framework Programme (FP7/2007--2013) ERC grant agreement no.\ 334828. The paper reflects only the authors' views and not the views of the ERC or the European Commission. The European Union is not liable for any use that may be made of the information contained therein.}}
\author{
Leslie Ann Goldberg\thanks{Department of Computer Science, University of Oxford, UK.}  
\and
  John Lapinskas\thanks{Department of Computer Science, University of Bristol, UK.}
\and
  David Richerby\thanks{School of Computer Science and Electronic Engineering, University of Essex, UK.}
 }

\date{26 August 2021}

\begin{document}
\maketitle
\begin{abstract}
Counting the independent sets of a graph is a classical \#P-complete problem, even in the bipartite case.
We give an exponential-time approximation scheme for this problem which is faster than the best known
algorithm for the exact problem. 
The running time of our algorithm on general graphs with error tolerance $\eps$ is 
at most $O(2^{0.2680n})$  times a polynomial in~$1/\epsilon$.
On bipartite graphs, the exponential term in the running time is improved to
 $O(2^{0.2372n})$.
Our methods combine techniques from exact exponential algorithms with techniques from approximate counting. 
Along the way we generalise (to the multivariate case) the FPTAS of Sinclair, Srivastava, \Stefankovic{} and Yin
for approximating the hard-core partition function on graphs with bounded connective constant. Also, we obtain an
FPTAS for counting independent sets on graphs with no vertices with degree at least 6 whose neighbours' degrees sum to $27$ or more.
By a result of Sly, there is no FPTAS that applies to all graphs with maximum degree~$6$ unless $\mbox{P}=\mbox{NP}$.
\end{abstract}

\section{Introduction} \label{sec:intro}

The problem of counting the independent sets of a bipartite graph 
was originally shown to be \#P-complete by Provan and Ball~\cite{provan}.
Clearly, this implies that counting the independent sets of an arbitrary graph is also \#P-complete.
The problem of enumerating independent sets goes back even further~\cite{bron}.

Given that there is unlikely to be a polynomial-time algorithm for counting independent sets,
much work \cite{DJW-old, DJW, Dubois, FK, GL, JT, WahlPaper, Zhang} 
has focussed on developing exponential algorithms that are as fast as possible.
The fastest known algorithm for counting independent sets is due to Gaspers and Lee~\cite{GL}, with running time  $O(2^{0.3022n})$,
where $n$ denotes the number of vertices of the graph~$G$. The algorithm of Gaspers and Lee builds on  a long history of  improvements, as detailed in Table~\ref{tab:bis-history}.

\begin{table}[th]
\label{tab:bis-history}
\begin{center}
\begin{tabular}{|l|l|}
\hline
Algorithm & Running time \\
			\hline
			Brute force enumeration & $O(2^n)\,\poly(n)$ \\
			Dubois~\cite{Dubois} and Zhang~\cite{Zhang} & $O(2^{0.6943n})$\\
			Dahll\"of, Jonsson and Wahlstr\"om~\cite{DJW-old} & $O(2^{0.4057n})$ \\
			Dahll\"of, Jonsson and Wahlstr\"om~\cite{DJW} & $O(2^{0.3290n})$ \\
			F\"urer and Kasiviswanathan~\cite{FK} & $O(2^{0.3174n})$ \\
			Wahlstr\"om~\cite{WahlPaper} & $O(2^{0.3077n})$ \\
			Junosza-Szaniawski and Tuczy\'nski~\cite{JT} & $O(2^{0.3068n})$ \\
			Gaspers and Lee~\cite{GL} & $O(2^{0.3022n})$ \\
			\hline
			\rule{0pt}{2.5ex}
			This paper (general graphs)& $O(2^{0.2680n})\, \poly(1/\epsilon)$ approximation \\ 
			This paper (bipartite graphs)& $O(2^{0.2372n})\, \poly(1/\epsilon)$ approximation \\
\hline
\end{tabular}
\end{center}
\caption{A history of exponential-time algorithms for counting independent sets.}
\end{table}

Given the \#P-completeness of \emph{exactly} counting independent sets, one might wonder whether there are faster algorithms for approximation.
Let $Z(G)$ 
\label{def:Z} denote the number of independent sets of a graph~$G$
and let \probIS{} be the problem of computing $Z(G)$, given a graph~$G$.
An \emph{$\epsilon$-approximation} of~$Z(G)$ is a number 
$N$ satisfying $(1-\epsilon) Z(G) \leq N \leq (1+\epsilon) Z(G)$. An \emph{approximation scheme} for \probIS{} is an 
algorithm that takes as input a graph~$G$ and an  error tolerance $\epsilon \in (0,1)$
and outputs an $\epsilon$-approximation of~$Z(G)$.
It is a \emph{fully polynomial-time} approximation scheme (FPTAS) if its running time
is  at most a polynomial in $|V(G)|$ and $1/\epsilon$.
Weitz~\cite{Weitz} gave an FPTAS for approximating $Z(G)$ on graphs of degree at most~$5$.
However, Sly~\cite{Sly} showed that, unless $\mbox{P}=\mbox{NP}$, there is no FPTAS  on graphs of degree at most~$6$.

Although there is  unlikely to be an FPTAS for  \probIS,
our main result is that there is an approximation scheme 
 whose running time is  faster than the best known algorithms
for exactly counting independent sets. In particular (see Theorem~\ref{thm:main}),
we give an approximation scheme for \probIS{}  whose running time, given an $n$-vertex graph~$G$
and an  error tolerance~$\epsilon$,
is
at most 
 $O(2^{0.2680n})$ times a polynomial in~$1/\epsilon$.
  
The problem of approximately counting the independent sets of a \emph{bipartite} graph,  denoted \probBIS, is a canonical problem  in  approximate counting.
It is known \cite{relative} to be complete in a large complexity class with respect to approximation-preserving reductions.
Many important approximate counting problems have been shown to be \probBIS-hard or \probBIS-equivalent \cite{ferropotts, hcol, csps, bdpotts}  so resolving  its complexity   is a major open problem in 
approximate counting.  

Our second result (also given in Theorem~\ref{thm:main}) is 
  that there  is 
an approximation scheme for \probBIS{} whose running time is faster than the best known  approximation scheme
for general graphs. 
  In particular,
we give an approximation scheme   whose running time is
at most $O(2^{0.2372n})$ 
 times a polynomial in~$1/\epsilon$.

\begin{restatable}{theorem}{thmmain}\label{thm:main}	
There is an approximation scheme for \probIS{} which, given an $n$-vertex input graph with error tolerance $\eps$, runs in time $O(2^{0.268n})\,\poly(1/\epsilon)$. 
There is an approximation scheme for \probBIS{} with running time $O(2^{0.2372n})\,\poly(1/\epsilon)$.
\end{restatable}

Our improvement for approximate counting over the best exact-counting algorithm has two sources.
The first is a technical analysis for bipartite graphs; we defer discussion of this to the   proof sketch in Section~\ref{sec:sketch}. The second  is a substantially faster handling of sparse instances for both bipartite and non-bipartite graphs. It is not at all surprising that this is possible, since there is an FPTAS for instances with maximum degree at most~5~\cite{Weitz} but it is \#P-complete to solve 3-regular instances exactly~\cite{XZZ} (even if they are bipartite). However,  it is necessary to do substantially more than to use~\cite{Weitz} as a black box.

Recall that the result of~\cite{Weitz} is tight in the sense that it is NP-hard to approximate $\IS(G)$ when $G$ is a 6-regular graph~\cite{Sly}. 
Nevertheless, the result of~\cite{Weitz} has been improved in several directions \cite{JKP, LL-oneside, SSSY}.
The most relevant of these in  our context  is the result of \cite{SSSY} which replaces
the maximum degree bound with a weaker bound on a natural measure of average degree called the \emph{connective constant}, which we now define.
Given a graph~$G$, a vertex~$v$, and a positive integer~$\ell$,
let $N_G(v,\ell)$ \label{def:NGvell} denote the number of 
simple paths of length~$\ell$ starting from~$v$ in~$G$. 

\begin{restatable}{defn}{defcc}\label{def:cc-2}
 (\cite{SSY}).
Let $\calF$ be a family of graphs and let $\kappa$ be a positive real number.  The connective constant
 of~$\calF$ is at most~$\kappa$ if there are real numbers $a$
 and~$c$ such that, for every $G\in\calF$, every vertex
  $v \in V(G)$, and every integer $\ell \ge a\log |V(G)|$, we have
  $\sum_{i=1}^\ell N_G(v,i) \le c\kappa^\ell$.  
\end{restatable}
 
As an example, let $\calF$ be a family of graphs such that
the maximum degree of any graph in~$\calF$ is~$d$.
It is easy to see that
 the connective constant of~$\calF$ is at most~$d-1$. However, for some such families~$\calF$ it is even smaller.

 Sinclair, Srivastava, \Stefankovic{} and Yin~\cite{SSSY}  give an FPTAS for $\IS(G)$ which works on any family of graphs  with connective constant at most 4.141.  These are the ``sparse instances'' that we referred to earlier.
 Our approach is to use this FPTAS for sparse instances  as a base case in our exponential-time algorithm
for counting independent sets.
We are left with two key difficulties. The first issue is that most previous 
fast exponential-time algorithms for counting independent sets exploit the ability to add ``weights'' to vertices; amongst other things, this allows for quick removal of degree-1 vertices so it substantially improves the running time. However, the result of~\cite{SSSY} requires all vertices to have the same weight, precluding the use of such techniques. 
We therefore give a more general, multivariate version of the theorem of~\cite{SSSY} (which may also be useful in other contexts).

In order to describe the generalisation, we need some notation.
Let $\calI(G)$ 
denote the set of independent sets of a graph~$G$.
We define a \emph{weighted graph} $\calG = (G,\win,\wout,W)$ 
to be a tuple consisting of a graph~$G$, two  functions~$\win$ and~$\wout$ assigning positive integer weights to the elements of~$V(G)$,  and a positive integer~$W$. 
We then define the partition function 
$$ \IS(\calG)   = W \sum_{I \in \calI(G)} 
  \prod_{v \in I}\win(v)  \prod_{v \in V(G)\setminus I}\wout(v).$$
 Modulo normalisation (which will be useful for our recursive instances), this is the multivariate hard-core partition function of~$G$.
 
 We need two more definitions in order to state our result.
 Given a positive real number~$\lambda$, we say that 
 a weighted graph 
 $\calG$~is \emph{$\lambda$-balanced} if, for every vertex $v$ of~$G$,
 we have 
 $\win(v)\leq\lambda\wout(v)$. ``$\lambda$-balance'' is just a way of bounding how much more weight accrues for being in an independent set, 
 as opposed to out of it.
 The result of \cite[Theorem 1.1]{SSSY} applies to the (usual) univariate hard-core model where every vertex has $\win(v) = \lambda$ and $\wout(v)=1$.
 They give an FPTAS which works for any family~$\calF$ with connective constant at most~$\kappa$
 as long as $\lambda <   \lambda_c(\kappa) $
 where $\lambda_c(\kappa) = \kappa^\kappa/{(\kappa-1)}^{\kappa+1} $ is the critical activity for the hard-core model on an infinite $\kappa$-ary tree.
 This leads to our last definition here.
 If  $\lambda < \lambda_c(\kappa)$ we say that $\kappa$ is \emph{subcritical with respect to~$\lambda$}.  
 \label{def:subcrit}

\begin{restatable}{theorem}{thmhcalgo}\label{thm:hc-algo}
Let $\lambda$ be a positive real number and let
 $\calF$ be a family of graphs whose connective constant is 
 at most a quantity which is subcritical with respect to~$\lambda$.  
 There is an FPTAS   for~$Z(\calG)$ on $\lambda$-balanced weighted graphs $\calG = (G,\win,\wout,W)$ such that $G\in\calF$.  
\end{restatable}
	
Thus, Theorem~\ref{thm:hc-algo} gives an FPTAS up to the critical point for the multivariate hard-core model, allowing our algorithm to use vertex weight in order to remove degree-1 vertices (amongst other things).
The next issue is how to use this FPTAS to get a faster running time. 
For this, we \label{def:2deg}
 define the \emph{2-degree} of a vertex as the sum of its neighbours' degrees, and we use Theorem~\ref{thm:hc-algo} to prove the following result,
 which may also be of independent interest.

\begin{restatable}{theorem}{lastintro}\label{thm:lastintro}
There is an FPTAS for $\IS(G)$ on graphs $G$ 
in which every vertex of degree at least~$6$ has 2-degree at most~$26$.
 \end{restatable} 

Theorem~\ref{thm:lastintro} implies that, given
an instance, either we can solve  it in polynomial time, 
or    we can find a vertex of degree at least~$6$ and 2-degree at least~$27$.
Most previous work in the area  is already built around finding vertices with high degree and high 2-degree (often called ``size''), so we can then use standard techniques to establish an improved bound on the running time, enabling us to prove Theorem~\ref{thm:main}.

\subsection{Proof Sketch}\label{sec:sketch}

\begin{figure}
\begin{center}
\begin{tikzpicture}[scale=0.75,node distance = 1.5cm]
    \tikzstyle{vtx}=[fill=black, draw=black, circle, inner sep=2pt]
    \tikzstyle{gry}=[fill=lightgray, draw=lightgray, circle, inner sep=2pt]

    \node at (-0.5,1) {$G$:};
    \node[vtx] (b) at (0,0) {};
    \node[vtx] (c) at (2,0) {};
    \node[vtx] (d) at (2,2) {};
    \node[vtx] (a) at (0,2) {};
    \node[vtx] (v) at (1,1) [label=270:{$v$}] {};
    \node[vtx] (e) at (3,1) {};
    \node[vtx] (f) at (4.4,1) {};

    \draw (a)--(b)--(c)--(d)--(a)--(v)--(c)--(e)--(d)--(v)--(b);
    \draw (e)--(f);

    \begin{scope}[shift={(7,1.75)}]
        \node at (-0.75,1) {$\tGin$:};
        \node[gry] (b2) at (0,0) {};
        \node[gry] (c2) at (2,0) {};
        \node[gry] (d2) at (2,2) {};
        \node[gry] (a2) at (0,2) {};
        \node[gry] (v2) at (1,1) [label={[gray]270:{$v$}}] {};
        \node[vtx] (e2) at (3,1) {};
        \node[vtx] (f2) at (4.4,1) {};

        \draw[lightgray] (a2)--(b2)--(c2)--(d2)--(a2)--(v2)--(c2)--(e2)--(d2)--(v2)--(b2);
        \draw (e2)--(f2);
    \end{scope}

    \begin{scope}[shift={(7,-1.75)}]
        \node at (-0.75,1) {$\tGout$:};
        \node[vtx] (b3) at (0,0) {};
        \node[vtx] (c3) at (2,0) {};
        \node[vtx] (d3) at (2,2) {};
        \node[gry] (v3) at (1,1) [label={[gray]270:{$v$}}] {};
        \node[vtx] (a3) at (0,2) {};
        \node[vtx] (e3) at (3,1) {};
        \node[vtx] (f3) at (4.4,1) {};

        \draw (d3)--(a3)--(b3)--(c3)--(e3)--(f3);
        \draw (e3)--(d3)--(c3);
        \draw[lightgray] (a3)--(v3)--(b3);
        \draw[lightgray] (c3)--(v3)--(d3);
    \end{scope}

    \draw[->,very thick,gray] (3.25,2.25) .. controls (4.4,2.65) .. (5.4,2.75);
    \draw[->,very thick,gray] (3.25,-0.25) .. controls (4.4,-0.65) .. (5.4,-0.75);
    
\end{tikzpicture}
\vskip 1cm
\begin{tikzpicture}[scale=0.75,node distance = 1.5cm]
    \tikzstyle{vtx}=[fill=black, draw=black, circle, inner sep=2pt]
    \tikzstyle{gry}=[fill=lightgray, draw=lightgray, circle, inner sep=2pt]
    \tikzstyle{vset}=[gray,rounded corners=5pt,dashed];

    \begin{scope}[shift={(0,0)}]
        \node [label=180:{$G$:}] at (-0.5,1) {};
        \node[vtx] (a) at (0,1) {};
        \node[vtx] (b) at (1,0) {};
        \node[vtx] (c) at (1,2) {};
        \node[vtx] (z) at (2,1) [label=270:{$z$}] {};

        \node[vtx] (d) at (3,0) {};
        \node[vtx] (e) at (3,2) {};

        \draw[vset,blue] (-0.5,-0.5) rectangle (1.5,2.5);
        \node at (0,2) {$S$};
        
        \draw (a)--(b)--(z)--(c)--(a);
        \draw (d)--(z)--(e)--(d);

        \draw (d)--($(d)+(50:0.9)$);
        \draw (d)--($(d)+(-20:1)$);

        \draw (e)--($(e)+(30:0.8)$);
        \draw (e)--($(e)+(-10:1)$);
        \draw (e)--($(e)+(-45:0.7)$);
    \end{scope}

    \begin{scope}[shift={(8.75,0)}]
        \node [label=180:{$G'$:}] at (-0.5,1) {};
        \node[gry] (a) at (0,1) {};
        \node[gry] (b) at (1,0) {};
        \node[gry] (c) at (1,2) {};
        \node[vtx] (z) at (2,1) [label=270:{$z$}] {};

        \node[vtx] (d) at (3,0) {};
        \node[vtx] (e) at (3,2) {};

        \draw[vset] (-0.5,-0.5) rectangle (1.5,2.5);
        \node at (0,2) {$S$};
        
        \draw[lightgray] (a)--(b)--(z)--(c)--(a);
        \draw (d)--(z)--(e)--(d);

        \draw (d)--($(d)+(50:0.9)$);
        \draw (d)--($(d)+(-20:1)$);

        \draw (e)--($(e)+(30:0.8)$);
        \draw (e)--($(e)+(-10:1)$);
        \draw (e)--($(e)+(-45:0.7)$);
    \end{scope}

    \draw[->,very thick,gray] (4.5,1) -- (6.75,1);
\end{tikzpicture}
\end{center}
\caption{Top: \emph{branching} on a vertex $v\in V(G)$.  Any independent set of~$G$ either contains~$v$ and has the form $\{v\} \cup S$, where $S$~is an independent set of $\tGin = G-v-\Gamma(v)$ or does not contain~$v$ and is an independent set of $\tGout = G-v$. Bottom: \emph{pruning} a set~$S\subseteq V(G)$ with $|\Gamma_G(S)|=1$. We form the graph~$G'$ by deleting~$S$ and updating the weights of the vertex~$z$ to compensate.}
\label{fig:branch}
\end{figure}
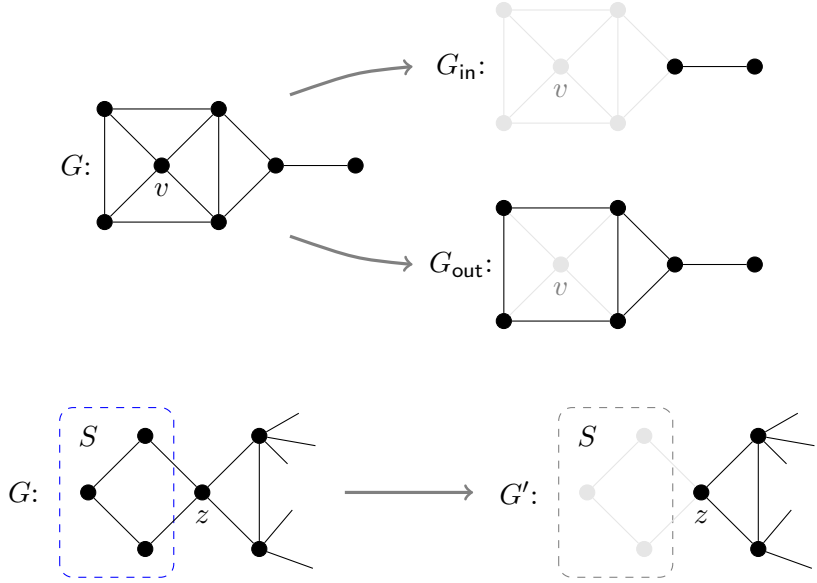

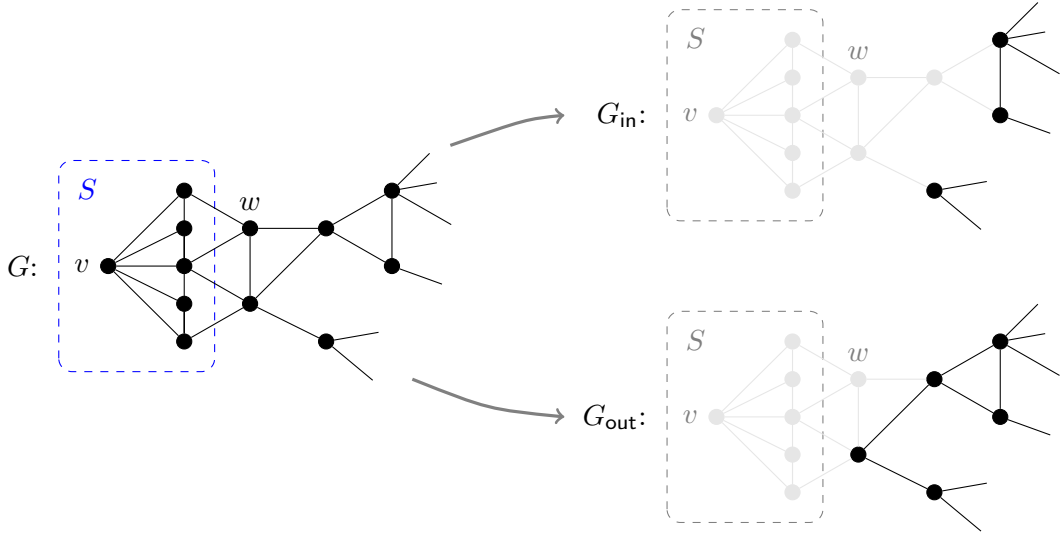
\begin{figure}
\begin{center}
\begin{tikzpicture}[scale=1,node distance = 1.5cm]
    \tikzstyle{vtx}=[fill=black, draw=black, circle, inner sep=2pt]
    \tikzstyle{gry}=[fill=lightgray, draw=lightgray, circle, inner sep=2pt]
    \tikzstyle{vset}=[gray,rounded corners=5pt,dashed];

    \begin{scope}[shift={(0,0)}]
        \node at (-0.65,1) [label=180:{$G$:}] {};
        \node[vtx] (v) at (0,1) [label=180:{$v$}] {};
      
        \node[vtx] (a) at (1,0) {};
        \node[vtx] (b) at (1,0.5) {};
        \node[vtx] (c) at (1,1) {};
        \node[vtx] (d) at (1,1.5) {};
        \node[vtx] (e) at (1,2) {};

        \node[vtx] (f) at ($(a)+(30:1)$) {};
        \node[vtx] (g) at ($(c)+(30:1)$)[label=90:{$w$}] {};

        \node[vtx] (i) at ($(g)+(1,0)$)  {};
        \node[vtx] (h) at ($(i)+(0,-1.5)$) {};
        \node[vtx] (j) at ($(i)+(30:1)$) {};
        \node[vtx] (k) at ($(j)+(0,-1)$) {};

        \draw[vset,blue] (-0.65,-0.4) rectangle (1.4,2.4);
        \node[blue] at (-0.275,2.025) {$S$};
        
        \draw (v)--(a)--(b)-- (v)--(c)--(d)-- (v)--(e)--(a)
            --(f)--(c)--(g)--(e);
        \draw (h)--(f)--(g)--(i)--(j)--(k)--(i)--(f);

        \draw (h)--($(h)+(10:0.7)$);
        \draw (h)--($(h)+(-40:0.8)$);

        \draw (j)--($(j)+(45:0.7)$);
        \draw (j)--($(j)+(10:0.6)$);
        \draw (j)--($(j)+(-30:0.9)$);

        \draw (k)--($(k)+(-20:0.7)$);
    \end{scope}

    \begin{scope}[shift={(8,2)}]
        \node at (-0.65,1) [label=180:{$\tGin$:}] {};
        \node[gry] (v) at (0,1) [label={[gray]180:{$v$}}] {};
      
        \node[gry] (a) at (1,0) {};
        \node[gry] (b) at (1,0.5) {};
        \node[gry] (c) at (1,1) {};
        \node[gry] (d) at (1,1.5) {};
        \node[gry] (e) at (1,2) {};

        \node[gry] (f) at ($(a)+(30:1)$) {};
        \node[gry] (g) at ($(c)+(30:1)$)[label={[gray]90:{$w$}}] {};

        \node[gry] (i) at ($(g)+(1,0)$)  {};
        \node[vtx] (h) at ($(i)+(0,-1.5)$) {};
        \node[vtx] (j) at ($(i)+(30:1)$) {};
        \node[vtx] (k) at ($(j)+(0,-1)$) {};

        \draw[vset] (-0.65,-0.4) rectangle (1.4,2.4);
        \node[gray] at (-0.275,2.025) {$S$};
        
        \draw[lightgray] (v)--(a)--(b)-- (v)--(c)--(d)-- (v)--(e)--(a)
            --(f)--(c)--(g)--(e);
        \draw[lightgray] (h)--(f)--(g)--(i)--(j);
        \draw[lightgray] (k)--(i)--(f);
        \draw (j)--(k);

        \draw (h)--($(h)+(10:0.7)$);
        \draw (h)--($(h)+(-40:0.8)$);

        \draw (j)--($(j)+(45:0.7)$);
        \draw (j)--($(j)+(10:0.6)$);
        \draw (j)--($(j)+(-30:0.9)$);

        \draw (k)--($(k)+(-20:0.7)$);
    \end{scope}

    \begin{scope}[shift={(8,-2)}]
        \node at (-0.65,1) [label=180:{$\tGout$:}] {};
        \node[gry] (v) at (0,1) [label={[gray]180:{$v$}}] {};
      
        \node[gry] (a) at (1,0) {};
        \node[gry] (b) at (1,0.5) {};
        \node[gry] (c) at (1,1) {};
        \node[gry] (d) at (1,1.5) {};
        \node[gry] (e) at (1,2) {};

        \node[vtx] (f) at ($(a)+(30:1)$) {};
        \node[gry] (g) at ($(c)+(30:1)$)[label={[gray]90:{$w$}}] {};

        \node[vtx] (i) at ($(g)+(1,0)$)  {};
        \node[vtx] (h) at ($(i)+(0,-1.5)$) {};
        \node[vtx] (j) at ($(i)+(30:1)$) {};
        \node[vtx] (k) at ($(j)+(0,-1)$) {};

        \draw[vset] (-0.65,-0.4) rectangle (1.4,2.4);
        \node[gray] at (-0.275,2.025) {$S$};
        
        \draw[lightgray] (v)--(a)--(b)-- (v)--(c)--(d)-- (v)--(e)--(a)
            --(f)--(c)--(g)--(e);
        \draw[lightgray] (f)--(g)--(i);
        \draw (h)--(f);
        \draw (i)--(j)--(k)--(i)--(f);

        \draw (h)--($(h)+(10:0.7)$);
        \draw (h)--($(h)+(-40:0.8)$);

        \draw (j)--($(j)+(45:0.7)$);
        \draw (j)--($(j)+(10:0.6)$);
        \draw (j)--($(j)+(-30:0.9)$);

        \draw (k)--($(k)+(-20:0.7)$);
    \end{scope}

    \draw[->,very thick,gray] (4.5,2.6) .. controls (5.5,2.95) .. (6,3);
    \draw[->,very thick,gray] (4,-0.5) .. controls (5,-.9) .. (6,-1);

\end{tikzpicture}
\end{center}
\caption{A vertex~$v$ with $|\Gamma^2_G(v)|=2$.  Rather than branching on~$v$, we branch on~$w$ and prune the set~$S$.}
\label{fig:smart-branch}
\end{figure}

Given a graph~$G=(V,E)$, a vertex $v\in V$, and a subset~$X$ of~$V$, 
\label{def:GamGv}
let $\Gamma_G(v)$ denote the set of neighbours of~$v$ in~$G$, 
and  let 
\label{def:dGv}
$d_G(v)= |\Gamma_G(v)|$ and 
$\Gamma_G(X) = \bigcup_{v \in X}\Gamma_G(v) \setminus X$. 
\label{def:GamGX}
Let $\calG =(G,\win,\wout,W)$ be a weighted graph, and write $G=(V,E)$. The two main building blocks of our algorithm are shown in Figure~\ref{fig:branch}. First, we can \emph{branch} 
\label{def:branch}
on a vertex~$v$, applying the relation
\[
\IS(\calG) = \wout(v)\,\IS(\calG-v) + \win(v)\,\prod_{x \in \Gamma_G(v)}\wout(x)\,\IS(\calG-v-\Gamma_G(v))
\]
to transform $\calG$ into two smaller instances, then solve the two smaller instances recursively. Second, if a set $S \subseteq V$ satisfies $|\Gamma_G(S)| \le 1$, then we can apply \emph{pruning}
\label{def:pruning}
(sometimes known as \emph{multiplier reduction}) to remove $S$ from $\calG$ in $O(2^{|S|})$ time; if $|\Gamma_G(S)| = 1$, say $\Gamma_G(S) = \{z\}$, then the weight on~$z$ will change in the process (see Definition~\ref{defn:prune}). In particular, this allows us to easily remove vertices of degree~$1$; for this reason, we will typically assume that input graphs have minimum degree~$2$. 

A slightly simplified version of the algorithm is then as follows. For all $v \in V$, we write $\Gamma_G^2(v)$ 
\label{gamG2v}
for the set of vertices at distance exactly two from~$v$. If we can determine $\IS(\calG)$ in polynomial time, we do so. Otherwise, if $G$ contains a vertex $v$ of degree at least~$11$, we branch on~$v$. Otherwise, we take $v$ to be a vertex of maximum $2$-degree subject to $d_G(v) \ge 2|E|/|V|$. If 
(after some suitable transformations to~$G$)
$|\Gamma_G^2(v)| \ge 3$, then we branch on~$v$. 
If $|\Gamma_G^2(v)| = 2$, then the choice of the branching-vertex is more complicated. In the simplest case, we branch on a vertex in $\Gamma_G^2(v)$ and apply pruning as shown in Figure~\ref{fig:smart-branch}. Finally, if $|\Gamma_G^2(v)| \le 1$, we apply pruning to remove $\{v\} \cup \Gamma_G(v)$ from $\calG$.

To analyse the running time of this algorithm, we use the potential-based method described in e.g.~\cite[Chapter~6]{FKbook}. This method requires us to define a non-negative potential function~$f$ on weighted graphs~$\calG$ with the property that every time our algorithm branches, yielding two instances 
$\Gout$ and~$\Gin$, we have $f(\Gout),f(\Gin) < f(\calG)$. (In particular, if $f(\calG) = 0$, then we must be able to approximate $\IS(\calG)$ in polynomial time without further branching.) The running time of the algorithm is then bounded in terms of the ensemble of possible values for the pair $(f(\calG)-f(\Gout),f(\calG)-f(\Gin))$. We take $f$ to be of roughly the following form (see Definition~\ref{def:pre-pot}). 
Let $s$ be a positive integer.
Taking $f_1,\dots,f_s\colon \R^2\to \R$ to be linear functions (\emph{slices}), and $-1=k_0 < k_1 < \dots < k_s = \infty$ to be real \emph{boundary points}, we set
\[
f(\calG) = f_i(|E|,|V|)\mbox{ whenever }2|E|/|V| \in (k_{i-1},k_i].
\]
We require that $f$ is continuous, and that for all $i \in [s]$ and all $m,n\in\R$, $f_i(m,n) = \min\{f_j(m,n) \colon j \in [s]\}$. This latter requirement will allow us to reduce the analysis of possible values of $(f(\calG)-f(\Gout),f(\calG)-f(\Gin))$ to an analysis of $(f_i(\calG)-f_i(\Gout),f_i(\calG)-f_i(\Gin))$ for each $i \in [s]$, subject to $2|E|/|V| \in (k_{i-1},k_i]$; since each slice $f_i$ is linear, this is much more tractable. We can then turn this analysis of possible branches into a near-optimal choice of $f_1,\dots,f_s$ and $k_1,\dots,k_s$ relatively easily with the help of a computer.

The ideas described above are classical, and were originally introduced in~\cite{DJW}. The novel parts of our argument are contained in the following three steps.

\medskip\noindent\textbf{Step 1:} If $G$ has average degree in $(k_i,k_{i+1}]$, prove that $G$ contains a vertex~$v$ with degree at least $2|E|/|V|$ and $2$-degree increasing with~$i$.

\medskip\noindent\textbf{Step 2:} Prove Theorem~\ref{thm:lastintro}, i.e.\ that if $G$ does not contain a vertex~$w$ with degree 
 at least~$6$
  and $2$-degree at least~$27$ then we can approximate $\IS(G)$ in polynomial time.

\medskip\noindent\textbf{Step 3:} Analyse the possible branches of the algorithm given that $G$ contains both the vertex~$v$ guaranteed by Step~1 and the vertex~$w$ guaranteed by Step~2.

\medskip\noindent We deal with Step 1 in Section~\ref{sec:aad}, and our approach largely follows that of~\cite{DJW}. Using a discharging argument, they showed that any graph with high average degree $d$ must contain a vertex with degree at least~$d$ and ``associated average degree'' at least~$d$ (see Lemma~\ref{lem:aad}). They then used a computer search to bound the minimum $2$-degree of any vertex with associated average degree at least~$d$. However, because our algorithm runs significantly faster, our analysis needs to work with graphs of maximum degree~$10$ rather than graphs of maximum degree~$5$, and this makes a na\"ive computer search prohibitively difficult; we require some preliminary analysis to narrow the search space, culminating in Lemma~\ref{lem:suitable}.

For Step~2, we first give an FPTAS for families $\mathcal{F}$ of suitably-weighted graphs with small connective constant (Theorem~\ref{thm:hc-algo}, proved in Section~\ref{sec:gadgets}). The work of Sinclair et al.~\cite{SSSY} applies only to graphs
with uniform weights on vertices, and pruning requires the use of differing vertex weights, so we cannot simply quote their result. Instead, we simulate vertex weights in the unweighted model using gadgets, and then apply~\cite{SSSY} as a black box to the resulting graph. The gadgets themselves are standard --- the main difficulty is in proving that the family formed by adding gadgets to graphs in~$\mathcal{F}$ has the same connective constant as $\mathcal{F}$ itself. 
With Theorem~\ref{thm:hc-algo} in hand, we then show in Section~\ref{sec:local} that the family~$\mathcal{D}$ of graphs with minimum degree at least two and with no vertices of degree at least~$6$ and $2$-degree at least~$27$ has connective constant small enough for Theorem~\ref{thm:hc-algo} to apply, thereby proving 
Theorem~\ref{thm:2-degree-FPTAS}.
This theorem contains almost everything needed to prove
Theorem~\ref{thm:lastintro}, though the proof of the latter is deferred until Section~\ref{sec:digress} where suitable pruning
has been introduced to remove isolated vertices and vertices of degree~$1$.
The proof of Theorem~\ref{thm:2-degree-FPTAS}
uses a potential function to bound the number of leaves of the graph's self-avoiding walk trees, in a way that could easily be repurposed to prove other bounds on the connective constant based on local structure; see Lemma~\ref{lem:kappa-dec}.

For Step~3, our analysis is significantly more complicated than that of~\cite{DJW}. The main difficulty is a technical one. If $\calG$ has average degree~$2$, then since it also has minimum degree~$2$, it must be $2$-regular and so we can obtain $\IS(\calG)$ in polynomial time. As such, we would like to set $f(\calG)=0$, which means taking the first slice of~$f$ to be $f_1(m,n) = c(m-n)$ for an appropriate constant $c>0$. But this has an undesirable consequence: In maintaining the invariant that $G$~has minimum degree~$2$, we will frequently remove isolated vertices and tree components from~$G$. Doing so may actually increase the value of~$f(\calG)$, which we are not allowed to do. Thus in analysing our branching operations, we must be very careful to account for the effect of removing all the isolated vertices and tree components that we create.

In the non-bipartite case, as in the analyses of~\cite{DJW} and~\cite{FK}, this turns out not to matter too much. Let $v$ be a vertex with some $2$-degree $D$, and suppose $|\Gamma_G^2(v)| \ge 3$ so that we branch on~$v$. Recall that we are concerned with the possible values of $(f_i(\calG)-f_i(\Gout), f_i(\calG) - f_i(\Gin))$ for a choice of~$i$ depending on the average degree of~$G$; we denote this pair by $\partial = (\partial_\textrm{out}, \partial_\textrm{in})$. (Roughly speaking, the larger $\partial_\textrm{out}$ and $\partial_\textrm{in}$ are, the more progress we make by branching on~$v$. Thus, the larger these values are in the worst case, the stronger the bound we secure on the algorithm's running time.) In general, we will have $f_i(\calG) = \sigma_i |V| + \rho_i |E|$ for some $\rho_i \ge 0$ and $\sigma_i \ge -\rho_i$, so $\partial$~is determined by the number of vertices and edges lost in passing to $\Gout$ and~$\Gin$. Since $\rho_i \ge 0$, losing more edges always increases~$\partial$ and thereby makes the branch better, but losing more vertices may make the branch worse if $\sigma_i < 0$. In passing to~$\Gin$, we always lose at least $D - |E(G[\Gamma_G(v)])|$ edges, and it turns out that even when $\sigma_j < 0$, the worst cases arise when $G[\Gamma_G(v)]$ contains many edges (see Lemma~\ref{lem:branch-bounds}(ii)). In the bipartite setting, however, $G[\Gamma_G(v)]$ never contains any edges, and when $\sigma_i < 0$ the worst cases turn out to arise when we lose as many vertices as possible. In particular, any vertex in $\Gamma^2_G(v)$ whose neighbourhood is contained in $\Gamma_G(v)$ will become isolated on passing to~$\Gin$, and will need to be removed to keep the minimum degree at least two; the worst cases are those in which $\Gamma^2_G(v)$ contains many such vertices (see Lemma~\ref{lem:branch-bounds}(ii)).

A careful analysis of all possible branches, and their effects on a slice~$f_i$ of our potential function~$f$, is the meat of Section~\ref{sec:slices}. We also set out the algorithm in detail, with some minor modifications to improve the worst-case branches when $|\Gamma_G^2(v)| = 2$ and to make the branching analysis easier, and bring the results of previous sections together to bound its running time in terms of~$f$ (see Corollaries~\ref{cor:fullalgo-runtime} and~\ref{cor:fullalgo-runtime-bi}). Finally, in Section~\ref{sec:pre-pot} we set out our choice of~$f$ and explain the code included to verify that it is valid, and in Section~\ref{sec:conclusion} we prove Theorem~\ref{thm:main}.

\section{Definitions and notation}\label{sec:notation}
	
Appendix~\ref{sec:index} contains an index of the notation used in this paper.	
		For all positive integers $n$, we write $[n]$ for the set $\{1,\dots,n\}$. All our logarithms are base $e$ unless specified otherwise.
	
Let $G$ be a graph. We write $n(G)$ for the number of vertices in $G$ and $m(G)$ for the number of edges in $G$. We will always assume $V(G) = [n(G)]$. For a vertex $v \in V(G)$ and a set $X \subseteq V(G)$, we write $\Gamma_G(v,X) = \{w\in X \mid \{v,w\} \in E(G)\}$ and $d_G(v,X) = |\Gamma_G(v,X)|$. If $X$ is not specified, it defaults to $V(G)$. In addition, we write $d^2_G(v) = \sum_{\{v,w\} \in E(G)} d_G(w)$; 
this is the same as the \emph{2-degree} of~$v$, as defined in the introduction.	 We write $\Gamma^2_G(v)$ for the set of vertices at distance exactly $2$ from~$v$; note that, unless $v$~has no neighbours, $d^2_G(v) > |\Gamma^2_G(v)|$. For a set $S \subseteq V(G)$, we write $\Gamma_G(S) = \bigcup_{v \in S}\Gamma_G(v) \setminus S$. We will omit the subscripts from all of these definitions when $G$ is clear from context. We also write $\delta(G)$ for the minimum degree of $G$ and $\Delta(G)$ for the maximum degree of $G$.  The length of a path is the number of edges it contains.

We define a {weighted graph}  \label{def:calG}
in a slightly more general way than what we described in the introduction. The generality helps to simplify the notation
for our recursive instances. A
 \emph{weighted graph} $\calG = (G,\win,\wout,W)$ 
is a tuple consisting of a graph~$G$, two positive integer functions $\win$ and $\wout$ whose domains contain $V(G)$, and a positive integer~$W$. For any $\lambda\in\R_{>0}$, we say that $\calG$~is \emph{$\lambda$-balanced} 
\label{def:lambal}
if $\win(v)\leq\lambda\wout(v)$ for all $v\in V(G)$. Sometimes, we will take $\win$ and $\wout$ to be the constant function with value 1 on domain $V(G)$ --- we denote this by $\mathbf{1}$. For all $X \subseteq V(G)$, we define 
\label{def:wplusmin}
$\win(X) = \prod_{x \in X}\win(x)$ and $\wout(X) = \prod_{x \in X}\wout(X)$. We adopt the convention that empty products are equal to one; in particular, $\win(\emptyset) = \wout(\emptyset) = 1$. We write $\calI(G)$  \label{def:calIG}
for the set of independent sets in $G$, and~define
\label{def:partition}
	\[
		\IS(\calG) = \IS(G,\win,\wout,W) = W\sum_{I \in \calI(G)} \win(I)\wout(V(G) \setminus I).
	\]
Thus $\IS(\calG)$ is the (multivariate) hardcore partition function of $G$ under weights $\win$ and $\wout$. By convention, if $G$ has no vertices then it has one independent set --- the empty set --- so $\IS(\calG) = W$. Note that if $\win(v)=\wout(v) = 1$ for all $v \in V(G)$, then $\IS(\calG) = W|\calI(G)|$. 
	
For all vertex sets $X$, we write $G-X = G[V(G)\setminus X]$, 
\label{def:GminX}
and for all $v \in V(G)$ we write $G-v = G-\{v\}$. We extend this notation to weighted graphs in the natural way, so that $\calG-X = (G-X,\win,\wout,W)$ and $\calG-v = (G-v,\win,\wout,W)$. Likewise, for all $X \subseteq V(G)$ we take 
\label{def:calGX}
$\calG[X] = (G[X],\win,\wout,W)$.

For $x,\,y\in\mathbb{R}_{\geq0}$ and $\epsilon\in(0,1)$, $x$~is an
\emph{$\epsilon$-approximation} to~$y$ if $(1-\epsilon)y\leq x\leq
(1+\epsilon)y$.

\section{Bounds on local structure from high average degree}\label{sec:aad}
	
In this section, we bound the ``worst-case'' branch, given a lower
bound on the average degree of our input graph. For this our main tool
will be as follows. Following Dahll\"of, Jonsson and
Wahlstr\"om~\cite{DJW} and subsequent authors~\cite{FK,WahlPhD}, given
a graph~$G$ with average degree~$k$, for each $v \in V(G)$ we define
\begin{align*}
    \alpha(v) &= d_G(v) + |\{w\in \Gamma_G(v)\mid d_G(w)<k\}| \\
    \beta(v)  &= 1 \ + \!\!\sum_{\substack{w\in \Gamma_G(v)\\d_G(w)<k}} \frac{1}{d_G(w)}\,.
\end{align*}
The \emph{associated average degree} of a vertex~$v$
is then given by $\aad(v) = \alpha(v)/\beta(v)$.

\subsection{Translating average degree to high associated average degree}

The following is essentially Lemma~6 of Dahll\"of \emph{et
  al.}~\cite{DJW}, but translated from \#2-SAT to \probIS{} under the usual correspondence.
  We give a full proof for completeness and so that the reader
doesn't have to translate notation.

\begin{lemma}[\cite{DJW}]
\label{lem:aad}
    Let $G=(V,E)$ be a graph with average degree~$k$. There is some
    $v\in V$ with $d_G(v)\geq k$ and $\aad(v)\geq k$.
\end{lemma}
\begin{proof}
    In this proof, we will write $d(v) = d_G(v)$ for all $v \in V(G)$. Let $X = \{v\in V\mid d(v)\geq k\}$.  There is at least one vertex
    with at-least-average degree, so $X\neq\emptyset$.  Let
    $A = \sum_{v\in X} \alpha(v)$ and $B = \sum_{v\in X} \beta(v)$.

    For a vertex~$v$, let $d_{\geq k}(v) = |\{w\in \Gamma(v)\mid d(w)\geq
    k\}|$ and write
    \[
        f(v) = \begin{cases}
                   d_{\geq k}(v) & \text{if }d(v)<k\\
                   d(v)        & \text{if }d(v)\geq k\,.
               \end{cases}
    \]
    We can view each $v\in V$ with $d(v)\geq k$ as contributing $d(v)$
    to~$A$, and each vertex with $d(v)<k$ as contributing $1$ to~$A$ for each neighbour with degree~$\geq k$.  Therefore,
    $A = \sum_{v\in V} f(v)$.  By similar reasoning,
    $B=\sum_{v\in V} f(v)/d(v)$.

    For $i\geq 1$, let $n_i$ be the number of degree-$i$ vertices
    in~$G$.  For $i\geq k$, let $n'_i = n_i$ and, for $1\leq i<k$, let
    $n'_i = |\{v\in V\mid d(v)<k\text{ and } d_{\geq k}(v)=i\}|$.
    There may be vertices with $d(v)<k$ and $d_{\geq k}(v)=0$, so
    $\sum_{1\leq i<k}n'_i \leq \sum_{1\leq i<k} n_i$.  Writing
    $S = \sum_{1\leq i<k}(n_i-n'_i)$, we have
    \[
        A = \sum_{i\geq 1} in'_i
              = \sum_{i\geq 1} in_i\ -\!\! \sum_{1\leq i<k}i(n_i-n'_i)
              = 2|E| \ -\!\! \sum_{1\leq i<k}i(n_i-n'_i)
              \geq 2|E| - kS
              = k(|V| - S)\,.
    \]
    For $i\geq k$, there are $n'_i$ vertices of degree~$i$, each of
    which contributes exactly~$1$ to~$B$.  For $1\leq i<k$, there are $n'_i$
    vertices~$v$ with $d_{\geq k}(v)=i$ and $d(v)<k$, each of which
    contributes at most~$1$ to~$B$.  Therefore,
    \[
        B \leq \sum_{i\geq 1} n'_i
              = \sum_{i\geq 1} n_i \ - \!\!\sum_{1\leq i<k}(n_i-n'_i)
              = |V| - S\,.
    \]

    Therefore, $A\geq kB$.
    If $\alpha(v)<k\,\beta(v)$ for all $v\in X$, then $A<kB$,
    contradicting what we have just derived.  Therefore, $X$~contains
    a vertex~$v$ with $\alpha(v)/\beta(v)\geq k$, and $d(v)\geq k$ by the
    definition of~$X$.
\end{proof}

\subsection{Using high associated average degree to lower-bound 2-degree}

We now lower-bound $d^2_G(v)$ in terms of $\aad(v)$. We will
find it useful to work with $\alpha(v)$ and $\beta(v)$ as numerical
functions, as follows.
	
\begin{defn}\label{def:rom-aad}
    Let $d\in\posints$ and $k\in\posreals$ with 
    $d\geq k$. 
    For all vectors $\bx \in \posints^d$, we define
    \begin{align*}
        a_k(\bx) &= d + |\{i \mid x_i < k\}|\\
        b_k(\bx) &= 1 \ + \!\!\sum_{i\colon x_i < k} \frac{1}{x_i}\\
		\saad_k(\bx) &= a_k(\bx)/b_k(\bx)\,.
    \end{align*}
\end{defn}

Associated average degree and $\saad$ correspond naturally.  Let $\bx
= (x_1, \dots, x_d)\in\posints^d$.  If $G$ is
a graph with average degree~$k$, and $v \in V(G)$ has degree~$d$ and
neighbours with degrees $x_1, \dots, x_d$, we have
$\alpha(v)=a_k(\bx)$, $\beta(v)=b_k(\bx)$ and $\aad(v)=\saad_k(\bx)$.

We use associated average degree to guarantee vertices of high
$2$-degree.

\begin{defn}
\label{def:D2}  
    For any  
 real numbers $k\geq 2$, let    
   $K =  \floor{k}+1    $ be the smallest integer that is larger than~$k$. Let
        $D_2(k)$ be the minimum integer~$z$ such that
    there is an integer $d\geq K$ and a tuple $\bx = (x_1, \dots,
    x_d)\in\Z_{\geq 2}^d$ with $\sum_{i=1}^d x_i = z$ and $\saad_{K}(\bx)>k$.
\end{defn}

\begin{lemma}\label{lem:D2} 
	For any    real number~$k\geq 2$,
	let  
   $K =  \floor{k}+1    $.  Every graph~$G$ with minimum degree at least~$2$ and average degree in $(k, K]$
	contains a vertex with degree at least $K$ and $2$-degree at least $D_2(k)$.
\end{lemma}
\begin{proof}
Let $A>k$ be the average degree of $G$. By
Lemma~\ref{lem:aad}, $G$~contains a vertex~$v$ 
such that $d_G(v) \geq A>k$ 
and $\aad(v) \geq A >k$.
Since  $d_G(v)$ is an integer, it is at least~$K$ (which is part of what is required in the lemma statement).

Let  $\mathbf{x} = (x_1, \dots, x_{d_G(v)})$  be the degrees of $v$'s neighbours in $G$.
By the definition of $\saad$, 
 $\aad(v) = \saad_{A}(\mathbf{x})$.
 
 Since 
 the elements of $\mathbf{x}$ are integers,
 $K$ is the smallest integer that is larger than~$k$, and
 $A \in (k,K]$,
 $x_i < A$ iff $x_i < K$.
 Thus, 
    $  \saad_{A}(\mathbf{x}) = \saad_{K}(\mathbf{x})$.

To finish, we wish to show that 	  $d^2_G(v) = \sum_i x_i$
is at least~$ D_2(k)$.
From the definition of~$D_2(k)$, it suffices to show that 
$\saad_{K}(\bx)>k$, which we have shown.
\end{proof}

\subsection{Efficiently computing the function $D_2(k)$ that we used to bound 2-degree}

In the remainder of this section, we work towards computing the
function $D_2(k)$.  Suppose $k\geq2$ 
and
let $K=\floor{k}+1$. Note that
\[
    \saad_K(\underbrace{K, \dots, K}_{K\text{ times}})
        = K > k\,,
\]
so $D_2(k)\leq K^2$.  In principle, one
could find $D_2(k)$ by exhaustive search through all tuples of
positive integers with total at most $K^2$, but this would take an unpleasantly long time for the values we wish to calculate, so we will restrict the search space. This leads us to the following definition. 

\begin{defn}\label{def:suitable}
	For any real number~$k\geq 2$, let $K = \floor{k}+1$. We say that a positive integer~$z$ is \emph{suitable} for~$k$ if either $z \ge K^2$ or there exist integers $d$ and $s$ such that the following hold:
	\begin{enumerate}[(S1)]
		\item $K \le d \le z/2$;
		\item $0 \le s \le d-1$;
		\item $Ks + 2(d-s) \le z \le Ks + (K-1)(d-s)$; and
		\item Writing $q = \floor{(z-Ks)/(d-s)}$, $d_1 = (z-Ks)\bmod (d-s)$, and $d_0 = d - s - d_1$,
		\[
			\frac{d+d_0+d_1}{1 + d_0/q + d_1/(q+1)} > k\,.
		\]
	\end{enumerate}
\end{defn}

Note that $d$  and $q$  are integers that are at least~$2$. Also,   $s$, $d_0$ and $d_1$ are non-negative integers.
The point of Definition~\ref{def:suitable} is Lemma~\ref{lem:suitable}, which  says that $D_2(k)$ is the minimum value of $z$ such that $k$ is suitable for $z$. We will use the following weaker result in the proof of Lemma~\ref{lem:suitable}.

\begin{lemma}\label{lem:suitable-prelim}
For any real number $k \ge 2$, and any positive integer $z$ which is suitable for $k$, we have $D_2(k) \le z$.
\end{lemma}

\begin{proof}
By the definition of $D_2(k)$, it suffices to exhibit a vector $\bx \in \Z_{\geq 2}^d$ (for some $d \ge K$) with $\sum_i x_i = z$ and $\saad_K(\bx) > k$. 
Let $K = \floor{k}+1$.  
	
\medskip\noindent \textbf{Case 1: $\boldsymbol{z \ge K^2}$.} In this case, we set $d=K$ and take $\bx$ to be the vector with $x_1,\dots,x_{K-1} = K$ and $x_K = (z-K(K-1))$. Since $z \ge K^2$, we have $x_K \ge K$, and so $\saad_K(\bx) = 
d=K > k$. Since $K > k \ge 2$, we have $\bx \in Z_{\geq 2}^K$. Finally, we have $\sum_i x_i = z$ as required, so the result follows.
	
\medskip\noindent \textbf{Case 2: $\boldsymbol{z \le K^2-1}$.} In this case, since $z$ is suitable for $k$, we can take $d$, $s$, $q$, $d_0$ and $d_1$ as in Definition~\ref{def:suitable}. We take $\bx$ to be the vector with $x_1,\dots,x_{d_0} = q$, $x_{d_0+1},\dots,x_{d_0+d_1} = q+1$, and $x_{d_0+d_1+1},\dots,x_{d_0+d_1+s} = K$. Then we have
	\begin{align*}
		\sum_{i=1}^{d_0+d_1+s} x_i &= qd_0 + (q+1)d_1 + Ks = q(d_0+d_1) + d_1 + Ks = q(d-s) + d_1 + Ks\\
		&=\Big\lfloor \frac{z-Ks}{d-s} \Big\rfloor(d-s) + (z-Ks) \bmod (d-s) + Ks = (z-Ks) + Ks = z\,,
	\end{align*}
	as required. 
	
	By (S3) we have $z-Ks \le (K-1)(d-s)$, so $q \le K-1$. If $q \le K-2$, then all entries in $x_1,\dots,x_{d_0+d_1}$ are at most $K-1$. If instead $q=K-1$, then we must have $z-Ks = (K-1)(d-s)$ and hence $d_1 = 0$; thus once again all entries in $x_1, \dots, x_{d_0+d_1}$ are at most $K-1$. All entries in $x_{d_0+d_1+1},\dots,x_d$ are equal to $K$, so in both cases we have 
	\[
		\saad_K(\bx) = \frac{d+d_0+d_1}{1+d_0/q+d_1/(q+1)}\,.
	\]
	By (S4), this is strictly greater than $k$ as required. 
	
	By (S3) we have $z-Ks \ge 2(d-s)$, so $q \ge 2$ and $\bx \in \Z_{\ge 2}^{d_0+d_1+s} = \Z_{\ge 2}^d$ as required. Finally, (S1) implies that $d \ge K$ as required, so the result follows.
\end{proof}

\begin{lemma}\label{lem:suitable}
	For any real number $k \ge 2$, $D_2(k)$ is the minimum integer~$z$ which is suitable for~$k$.
\end{lemma}
\begin{proof}
	By Lemma~\ref{lem:suitable-prelim}, the minimum integer $z$ which is suitable for $k$ is at least $D_2(k)$, so it suffices to show that $D_2(k)$ itself is suitable for $k$. Let $z = D_2(k)$; then by Definition~\ref{def:D2}, there is an integer $d\geq K$ and a tuple $\bx = (x_1, \dots,
	x_d)\in\Z_{\geq 2}^d$ 
	with $\sum_{i=1}^d x_i = z$ and $\saad_{K}(\bx)>k$. Fix such a $d$ and $\bx$.
	
	\medskip\noindent
	{\bf Claim 1.\quad} We may choose $\textbf{x}$ to ensure that for all $x_i, x_j \le K-1$, we have $|x_i-x_j|\leq 1$.
	
	\medskip\noindent
	\textit{Proof:} Suppose that there are elements $2\leq x_i < x_j \le K-1$ with $x_j \ge x_i + 2$.
	Consider the tuple $\bx'$ formed from~$\bx$   by replacing $x_i$ and $x_j$ with 
	$x'_i = x_i +1$ and $x'_j = x_j-1$.  
	It is clear that $x'_j = x_j - 1 > x_i \geq2$,
	so $\bx' \in \Z_{\geq 2}^d$, and that $\sum_{i=1}^d x_i' = \sum_{i=1}^d x_i = z$.
	
	To finish proving the claim, we will show that
	$\saad_K(\bx') \geq \saad_K(\bx)$.
	This means that we can consider $\bx'$ in place of $\bx$, repeating as necessary until Claim~1 is satisfied. We have $x_i,x_i',x_j' < x_j \le K-1$, so $a_K(\bx') = a_K(\bx)$. Let $f(y)$ be the function given by $f(y)=\tfrac{1}{y}-\tfrac{1}{y+1}$. Then
	\[
	b_K(\bx')
	= b_K(\bx) - \frac{1}{x_i} -  \frac{1}{x_j}
	+ \frac{1}{x_i+1} + \frac{1}{x_j-1}
	= b_K(\bx) - f(x_i) + f(x_j-1)
	\leq b_K(\bx)\,,
	\]
	since $0<x_i\leq x_j-1$ and $f(y)$~is decreasing for
	$y>0$. Hence $\saad_K(\bx') = a_K(\bx')/b_K(\bx') \ge a_K(\bx)/b_K(\bx) = \saad_K(\bx)$, as required.
	
	\medskip\noindent
	{\bf Claim 2.\quad}  For all $i\in[d]$, $x_i \leq K$.  
	
	\medskip\noindent
	\textit{Proof:} Suppose for contradiction that for some $i\in [d]$, $x_i>K$. Consider the tuple $\bx' \in \Z_{\ge 2}^d$ formed from~$\bx$ by replacing $x_i$ with $x'_i=K$. By the definition of $\saad$, $\saad_K(\bx') = \saad_K(\bx)>k$. But the sum of the elements in $\bx'$ is smaller than~$z$, contradicting $z=D_2(k)$.
	
	\medskip\noindent
	{\bf Claim 3.\quad} $z$ is suitable for $k$.
	
	\medskip\noindent
	\textit{Proof:} Let $q = \min\{x_i \colon i \in [d]\}$. By Claims 1 and 2, we can choose $\bx$ so that its coordinates lie in the set $\{q,q+1,K\}$. Let $s$ be the number of coordinates of $\bx$ equal to $K$, let $d_0$ be the number of coordinates of $\bx$ equal to $q$, and let $d_1 = d - d_0 - s$. (Thus $d_1$ is the number of coordinates of $\bx$ equal to $q+1$ unless $q+1=K$, in which case $d_1=0$.) 
	
	If $s=d$, then we have $z = Kd \ge K^2$, and so $z$ is suitable for $k$; suppose instead $0 \le s \le d-1$ (so (S2) is satisfied). Since $\sum_{i=0}^d x_i = z$, we have 
	\[
		q=\Big\lfloor \frac{z-Ks}{d-s} \Big\rfloor \mbox{ and } d_1 = (z-Ks)\bmod (d-s),
	\]
	so $q$, $d_0$ and $d_1$ are exactly as in Definition~\ref{def:suitable}. Moreover, since either $q+1\le K-1$ or $d_1 = 0$, we have
	\[
	\frac{d+d_0+d_1}{1+d_0/q+d_1/(q+1)} = \frac{a_K(\bx)}{b_K(\bx)} = \saad_K(\bx) > k,
	\]
	so (S4) is satisfied. Since $x \in \Z_{\ge 2}^d$, we have $z = \sum_{i=1}^d x_i \ge Ks + 2(d-s)$, and by Claim 1, we have $z = \sum_{i=1}^d x_i \le Ks + (K-1)(d-s)$; hence (S3) is satisfied. Finally, we have $d \ge K$ and (S3) implies that $z \ge 2d$, so (S1) is satisfied. Hence $z$ is suitable for $k$, and the result follows.
 
\end{proof}
\begin{cor}\label{cor:suitable}
	For any real number $k \ge 2$, writing $K = \floor{k}+1$, we have $2K \le D_2(k) \le K^2$.
\end{cor}
\begin{proof}
	$D_2(k) \le K^2$ follows immediately from Lemma~\ref{lem:suitable} and the fact that $K^2$ is always suitable for $K$, and $D_2(k) \ge 2K$ follows immediately from Lemma~\ref{lem:suitable} and (S1).
\end{proof}

Given Lemma~\ref{lem:suitable} and Corollary~\ref{cor:suitable}, we can compute $D_2(k)$ 
for any real number~$k\geq 2$  
by considering increasing integers~$2K \le z \le K^2$ 
and checking whether $z$ is suitable for~$k$. 
Mathematica code to do this is included in Appendix~\ref{app:mathematica-valid} as the functions \texttt{isSuitable[]} and \texttt{Dtwo[]}.

\section{\texorpdfstring{Approximating $\boldsymbol{\IS(\calG)}$ when the connective constant is small}{Approximating Z(G) when the connective constant is small}}\label{sec:gadgets}

In this section, we prove Theorem~\ref{thm:hc-algo}.
That is, for a positive real number~$\lambda$ and a family~$\calF$ of graphs
 whose connective constant is subcritical with respect to~$\lambda$, 
 we give an  FPTAS   for~$Z(\calG)$ on $\lambda$-balanced weighted graphs $\calG = (G,\win,\wout,W)$ such that $G\in\calF$.  
 This FPTAS provides the base case for our main algorithm  (see Theorem~\ref{thm:main}).  
 
\subsection{Definitions and Notation} 
 
 We start by recalling some ideas from the introduction, and giving more detailed notation.

\begin{defn}[\cite{Godsil}]\label{def:saw}
Let $G$ be a graph and let $v \in V(G)$. The \emph{self-avoiding walk tree}     of~$G$ with root~$v$ is the following tree $\sawt{v}{G}$. 
The vertices are the self-avoiding walks (simple paths) from~$v$ in~$G$.
The root is the trivial path~$v$.
Given any vertex~$P$, the children of~$P$ are the simple paths that extend~$P$ by one
edge.    $N_G(v,\ell)$  denotes the number of 
length-$\ell$
self-avoiding walks in~$G$ starting from~$v$
so it is  the number of
depth-$\ell$ vertices in $\sawt{v}{G}$, where the root is considered to have depth~0.
\end{defn}

Recall the following definitions from the introduction.

\defcc*

\begin{defn}\label{def:subcritical}
Let $\kappa$ and $\lambda$ be positive real numbers.
$\kappa$ is said to be \emph{subcritical w.r.t.~$\lambda$} 
if $\lambda <  \lambda_c(\kappa):= \kappa^\kappa/{(\kappa-1)}^{\kappa+1}$.  
\end{defn}

Note that $\lambda_c(\kappa)$ is decreasing in~$\kappa$, so
if $\kappa$ is subcritical w.r.t.\@ some~$\lambda$ and $\kappa'<\kappa$, then $\kappa'$~is also subcritical w.r.t.~$\lambda$.

\subsection{Reduction to approximating the univariate hardcore partition function}

In the \emph{univariate hardcore model}, we are given a graph~$G$ and a real parameter $\lambda>0$ and we wish to compute the partition function
\label{def:hardcore}
\[
    Z_\lambda(G) = \sum_{I\in \calI(G)} \lambda^{|I|}\,.
\]
Thus, $\lambda=1$ corresponds to counting independent sets and, for any $\lambda>0$, $Z_\lambda(G) = Z(\calG)$ where $\calG=(G, \win, \mathbf{1}, 1)$ and $\win$~is the constant function which assigns weight~$\lambda$ to every vertex.

Note that, as $\lambda$~decreases, the condition that a family $\calF$~has subcritical connective constant w.r.t.~$\lambda$ becomes weaker, but the requirement that $\win(v)\leq \lambda\wout(v)$ becomes more restrictive.

We will prove Theorem~\ref{thm:hc-algo} by reducing the problem 
of approximating~$Z(\calG)$
on $\lambda$-balanced weighted graphs $\calG = (G,\win,\wout,W)$ such that
$G$ comes from a family whose   connective constant    is subcritical with respect to~$\lambda$
to the problem of approximating the partition function of
the (univariate) hardcore model in the subcritical case, using    the following theorem of Sinclair \emph{et al}.

\begin{theorem}[\cite{SSSY}]\label{thm:cc-2}
    Fix a family~$\calF$ of graphs and $\lambda\in\R_{>0}$ such that $\calF$~has subcritical connective constant w.r.t.~$\lambda$.  There is an FPTAS $\hcunialgo_{\lambda,\calF}(G,\epsilon)$ for $Z_\lambda(G)$ for input graphs~$G\in\calF$.
\end{theorem}

\subsection{Defining the gadgets for the reduction}

We first define the gadgets we will use in our reduction. The construction is relatively simple --- given a weighted input graph $\calG = (G,\win,\wout,W)$, we form an instance of the hardcore model by attaching appropriately chosen depth-$2$ trees to each vertex of~$G$. However, to apply Theorem~\ref{thm:cc-2}, we must bound the connective constant of the family of all possible hardcore instances that we may construct (see Lemma~\ref{lem:conn-const-2}). For this reason, we use the following notation.

\begin{defn}\label{def:gadg}
	For any graph $G=(V,E)$, a \emph{weight map for $G$} is a map $\phi$ from $V$ to finite (possibly empty) multisets of positive integers. The \emph{realisation of $\phi$} is the graph $G'$ formed from $G$ as follows. Let $\{T_{v,t} \mid v \in V,\,t \in \phi(v)\}$ be a multiset of vertex-disjoint stars, where $T_{v,t}$ has $t$~leaves. Take the disjoint union of $G$ and the $T_{v,t}$'s, then join the centre of each star $T_{v,t}$ to $v$ by an edge.
\end{defn}

An example of the construction is given in Figure~\ref{fig:gadget}. We will apply Theorem~\ref{thm:cc-2} to a family of graphs defined as in Definition~\ref{def:Fplus}, using Lemma~\ref{lem:conn-const-2} to argue that the connective constant is subcritical. 

\begin{figure}
\begin{center}
\begin{tikzpicture}[scale=1,node distance = 1.5cm]
    \tikzstyle{vtx}=[fill=black, draw=black, circle, inner sep=2pt]
    \tikzstyle{svtx}=[fill=black, draw=black, circle, inner sep=1.5pt]

    \node[vtx] (v1) at (0,0)    [label=135:{$v_1$}] {};
    \node[vtx] (v2) at (1.5,0)  [label=45:{$v_2$}] {};
    \node[vtx] (v3) at (60:1.5) [label=90:{$v_3$}]  {};
    \draw (v1) -- (v2) -- (v3) -- (v1);

    \node[svtx] (v11) at ($(v1)+(180:1)$) {};
    \node[svtx] (v12) at ($(v1)+(210:1)$) {};
    \node[svtx] (v13) at ($(v1)+(240:1)$) {};
    \node[svtx] (v111) at ($(v11)+(165:0.67)$) {};
    \node[svtx] (v112) at ($(v11)+(195:0.67)$) {};
    \draw (v11)--(v1)--(v12);
    \draw (v13)--(v1);
    \draw (v111)--(v11)--(v112);

    \node[svtx] (v21)  at ($(v2)+(-30:1)$) {};
    \node[svtx] (v211) at ($(v21)+(0:0.67)$) {};
    \node[svtx] (v212) at ($(v21)+(-30:0.67)$) {};
    \node[svtx] (v213) at ($(v21)+(-60:0.67)$) {};
    \draw (v2)--(v21)--(v211);
    \draw (v212)--(v21)--(v213);
\end{tikzpicture}
\end{center}
\caption{An example of the realisation of a weight map $\phi$ for a triangle with vertices $v_1$, $v_2$ and $v_3$. Here $\phi(v_1) = \{0,0,2\}$, $\phi(v_2) = \{3\}$, and $\phi(v_3) = \emptyset$.} \label{fig:gadget}
\end{figure}
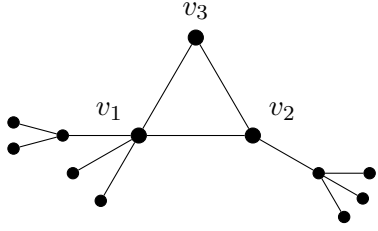

\begin{defn}\label{def:Fplus}
	Let $\calF$ be a family of graphs and let $y\in\Z_{>0}$. For any graph $G=(V,E)$, let $\Phi_y(G)$ be the set of all possible weight maps~$\phi$ for~$G$ such that, for all $v \in V$, $\sum_{i \in \phi(v)} (i+1) \le y|V|^2$. We define $\calF^+_y$ to be the set of all realisations of weight maps in $\bigcup_{G \in \calF}\Phi_y(G)$.
\end{defn}

The choice of $y|V|^2$ in the above definition is not fundamental; the following lemma remains true, with minor modifications to the proof, for any choice of polynomial in~$|V|$.

\begin{lemma}\label{lem:conn-const-2}
	Let $\kappa\in\R_{\geq 1}$, $\lambda\in\R_{>0}$ and $y\in\Z_{\geq1}$.  Let $\calF$ be a family of graphs with connective constant at most~$\kappa$, subcritical w.r.t.~$\lambda$. Then $\calF^+_y$ also has subcritical connective constant w.r.t.~$\lambda$.
\end{lemma}
\begin{proof}
    Fix $\gamma>1$ such that $\gamma\kappa$ is subcritical w.r.t.~$\lambda$.  We show that $\calF^+_y$ has connective constant at most $\gamma\kappa$.

	By the definition of $\kappa$ (Definition~\ref{def:cc-2}), there exist real numbers $a$ and~$c$ such that for all $G \in \calF$, all $v \in V(G)$, and all $\ell \ge a\log|V(G)|$, we have $\sum_{i=1}^\ell N_G(v,i) \le c\kappa^\ell$. Let $a' = \max\{a,8/\log\gamma\}$ and $c' = 40y^4(1+c)$; we will show that for all $G' \in \calF^+_y$, all $v \in V(G')$, and all $\ell \ge a'\log|V(G')|$, we have $\sum_{i=1}^\ell N_{G'}(v,i) \le c'(\gamma\kappa)^\ell$.  For ease of notation, let $n=|V(G)|$ and $n'=|V'(G)|$.
	
	Let $G' \in \calF^+_y$, let $v \in V(G')$. Let $G \in \calF$ and $\phi \in \Phi_y(G)$ be such that $G'$~is a realisation of~$\phi$. Observe that, for all positive integers~$i$, any self-avoiding walk in $G'$ of length~$i$ must stay almost entirely within~$G$, only entering the attached trees~$T_{w,t}$ for at most two steps at the start (if $v \in V(G') \setminus V(G)$) and at most two steps at the end. Since $\phi \in \Phi_y(G)$, for all $w \in V(G)$ we have
	\begin{equation}\label{eqn:n-prime}
		\Big|\bigcup_{t\in\phi(w)} V(T_{w,t})\Big| = \sum_{t \in \phi(w)}(t+1) \le yn^2,
	\end{equation}
	so there are at most $yn^2$ choices for each step outside $V(G)$. Thus, for all $i\ge 4$,
	\begin{align*}
		N_{G'}(v,i) &\le N_G(v,i) + 2yn^2 N_G(v,i-1) + 3y^2n^4 N_G(v,i-2)\\
		&\qquad\qquad + 2y^3n^6 N_G(v,i-3) + y^4n^8 N_G(v,i-4)\\
		&\le 3y^4n^8\sum_{j=i-4}^i N_G(v,j).
	\end{align*}
    Now, suppose that $0\leq i\leq 3$.  Note that, by~\eqref{eqn:n-prime}, $n'\leq yn^2+n$.  For all $v\in V(G')$,
    \[
        N_{G'}(v,i)\leq (n')^i \leq (yn^2+n)^i \leq 8y^3n^6 < 8y^4n^8
            \leq 8y^4n^8 \sum_{j=0}^i N_G(v,j)\,.
    \]
	Therefore, for all $0 \le i \le n$,
	\[
		N_{G'}(v,i) \le 8y^4n^8 \!\!\!\!\!\!\sum_{j=\max\,\{0,i-4\}}^i \!\!\!\!\!\!N_G(v,j).
	\]
	For any~$\ell$, Each term $N_G(v,k)$ appears at most five times in $\sum_{i=1}^\ell\sum_{j=\max\,\{0,i-4\}}^i N_G(v,j)$, so
	\[
		\sum_{i=1}^\ell N_{G'}(v,i) \le 40y^4n^8 \sum_{i=0}^\ell N_G(v,i).
	\]
	For $\ell \ge a'\log n' \ge a\log n$ and $\kappa\ge 1$, it follows that
	\[
		\sum_{i=1}^\ell N_{G'}(v,i) \le 40y^4n^8 (1+ c\kappa^\ell) \le c'n^8\kappa^\ell.
	\]
	By the definition of $a'$, we have $n^8 \le \gamma^{a'\log n} \le \gamma^\ell$, so $\sum_{i=1}^\ell N_{G'}(v,i) \le c'(\gamma\kappa)^\ell$. Thus, the connective constant of $\calF^+_y$ is at most $\gamma\kappa$, which is subcritical.
\end{proof}

\subsection{The approximation algorithm}

We now state the algorithm of Theorem~\ref{thm:hc-algo}.  Throughout, we write $\Lambda_t=1 + \lambda(1+\lambda)^{-t}$ 
\label{defLamt}
for any $t\in\Z_{\geq 0}$.  Note that, for all $\lambda,\,t>0$, we have $1 < \Lambda_{t+1} < \Lambda_{t}$.  Recall that $\calG=(G,\win,\wout,W)$~is $\lambda$-balanced if $\win(v)\leq\lambda\wout(v)$ for all $v\in V(G)$.

\label{def:alg1} 
\begin{algorr}{$\hcalgo_{\lambda,\calF}(\calG,\eps)$}{
      The parameters are $\calF$, a family of graphs, and $\lambda\in\R_{>0}$ such that $\calF$~has subcritical connective constant w.r.t.~$\lambda$; $\calF$ and~$\lambda$ are not part of the input.  The inputs are $\calG = (G,\win,\wout,W)$, a non-empty $\lambda$-balanced $n$-vertex weighted graph with $G \in \calF$, and a rational number $0 < \eps < 1$.  The output is an $\epsilon$-approximation to $Z(\calG)$.}
	\item If $\eps \le 3n\lambda(1+\lambda)^{-n}$, then calculate $\IS(\calG)$ exactly by brute force and return it.\label{HC-1}
	\item Form a weighted graph $\calG' = (G',\win,\wout,W')$ from $\calG$ as follows: for each $v \in V(G)$ with $\win(v) \le \frac{\eps}{3n}\wout(v)$, delete $v$ and multiply $W$ by $\wout(v)$. Let $W'$ be the final value of~$W$.\label{HC-2}
	\item For all $v \in V(G')$, calculate $a_{0,v}, \dots, a_{n,v}$ as follows. Initially, set $x_v\leftarrow\lambda\wout(v)/\win(v)$. Then, for $t = 0$ to $n$, let $a_{t,v} = \lfloor\log_{\Lambda_t} x_v\rfloor$, and update $x_v \leftarrow x_v/\Lambda_t^{a_{t,v}}$.\label{HC-3}
	\item Define a weight map $\phi$ for $G'$ by mapping each vertex $v$ to the multiset consisting of exactly $a_{t,v}$ copies of~$t$ for all $t \in \{0, \dots, n\}$. Construct the realisation $G''$ of $\phi$. \label{HC-4}
	\item Return $\hcunialgo_{\lambda,\calF^+_{\lfloor 5(1+\lambda)\rfloor}}(G'',\eps/3) \cdot W'\lambda^{-n}\prod_{v \in V(G')}\big(\win(v)/\prod_{i=0}^n (1+\lambda)^{i\,a_{i,v}}\big)$.\label{HC-5}
\end{algorr}

\subsection{Building up to the proof of correctness}

It is immediate that $\hcalgo_{\lambda,\calF}(\calG,\eps)$ outputs $\IS(\calG)$ when $\eps \le 3n\lambda(1+\lambda)^{-n}$. When $\eps > 3n\lambda(1+\lambda)^{-n}$, we prove correctness in three steps, as detailed in the proof of Theorem~\ref{thm:hc-algo}. We show that $\IS(\calG)$ is roughly $\IS(\calG')$ (see Lemma~\ref{lem:hcalgo-1}) and that $Z_\lambda(G'')$ is roughly $\IS(\calG')/C$ where
\[
	C = W'\lambda^{-n}\!\!\!\prod_{v \in V(G')}\!\!\!\Big(\win(v)\,\big/\,\prod_{i=0}^n (1+\lambda)^{i\,a_{i,v}}\Big)
\]
is easily computed (Lemma~\ref{lem:hc-algo-2}).
Finally, we will show that $G''\in \calF^+_{\lfloor 5(1+\lambda)\rfloor}$ and that $\calF^+_{\lfloor 5(1+\lambda)\rfloor}$ has subcritical connective constant w.r.t.~$\lambda$, so that $\hcunialgo_{\lambda,\calF^+_{\lfloor 5(1+\lambda)\rfloor}}(G'',\eps/3)$ is roughly $Z_\lambda(G'')$ in step~(\ref{HC-5}) (Lemma~\ref{lem:hc-ai}). We conclude that the output is roughly $\IS(\calG)$, as required (see the proof of Theorem~\ref{thm:hc-algo}). This last part of the argument is the reason we bound $\eps$ below in step~(\ref{HC-1}) and pass from $\calG$ to~$\calG'$ in step~(\ref{HC-2}).

\begin{lemma}\label{lem:hcalgo-1}
	Let $\lambda>0$ and let $\calF$ be a family of graphs with subcritical connective constant w.r.t.~$\lambda$. Let $\calG = (G,\win,\wout,W)$ be a non-empty $\lambda$-balanced $n$-vertex weighted graph with $G \in \calF$. Let $3n\lambda(1+\lambda)^{-n} < \eps < 1$ be rational. Then after step~(\ref{HC-2}) of $\hcalgo_{\lambda,\calF}(\calG,\eps)$, we have $\IS(\calG) \ge \IS(\calG') \ge (1 - \eps/3)\IS(\calG)$.
\end{lemma}
\begin{proof}
	Write $\calG' = (G',\win,\wout,W')$ as in $\hcalgo_{\lambda,\calF}(\calG,\eps)$, and write $S = V(G) \setminus V(G')$. Observe that $W' = W\wout(S)$, so
	\begin{equation}\label{eqn:hcalgo-1-init}
		\IS(\calG') = W'\!\!\sum_{I \in \calI(G')}\!\! \win(I)\,\wout(V(G') \setminus I) = W\!\!\!\sum_{\substack{I \in \calI(G)\\I\subseteq V(G')}}\!\!\! \win(I)\,\wout(V(G) \setminus I)\,.
	\end{equation}
	It follows immediately that
	\begin{equation*}
		\IS(\calG') \le W\!\!\sum_{I \in \calI(G)}\!\!\win(I)\,\wout(V(G)\setminus I) = \IS(\calG)\,,
	\end{equation*}
	as required. Moreover, since for all $v \in S$ we have $\win(v) \le \tfrac{\eps}{3n}\wout(v)$,
	\begin{align*}
		\wout(S) 
		&= \prod_{v \in S}\wout(v) 
		\ge \prod_{v \in S} \frac{\wout(v)+\win(v)}{1+\eps/3n}
		= \frac{1}{(1+\eps/3n)^{|S|}}\sum_{J \subseteq S}\win(J)\,\wout(S\setminus J)\\
		&\ge (1-\eps/3n)^{|S|}\sum_{J\subseteq S}\win(J)\,\wout(S\setminus J)
		\ge (1-\eps/3)\sum_{J\subseteq S}\win(J)\,\wout(S\setminus J)\,.
	\end{align*}
	It follows from~\eqref{eqn:hcalgo-1-init} that
	\begin{align*}
		\IS(\calG') &\ge (1-\eps/3)\,W\!\!\!\sum_{\substack{I \in \calI(G)\\I\subseteq V(G')}}\!\!\!\win(I)\,\wout(V(G')\setminus I)\sum_{J \subseteq S}\win(J)\,\wout(S\setminus J)\\
		&\ge (1-\eps/3)\,W\!\!\sum_{I \in \calI(G)}\!\!\win(I)\,\wout(V(G)\setminus I)
		= (1-\eps/3)\,\IS(\calG)\,.\qedhere
	\end{align*}
\end{proof}

We collect the relevant properties of the $a_{t,v}$'s and $x_v$'s in the following lemma.

\begin{lemma}\label{lem:hc-ai}
	Let $\lambda>0$ and let $\calF$ be a family of graphs with subcritical connective constant w.r.t.~$\lambda$. For every 
non-empty $\lambda$-balanced $n$-vertex weighted graph $\calG = (G,\win,\wout,W)$ with $G \in \calF$, and every rational~$\epsilon$ satisfying $3n\lambda(1+\lambda)^{-n} < \eps < 1$, the following are true after step~(\ref{HC-4}) of $\hcalgo_{\lambda,\calF}(\calG,\eps)$:
    \begin{enumerate}[(i)]
    \item for all $v \in V(G')$,\label{eqn:hc-ai-state}
	$\displaystyle
		\Big(1-\frac{\eps}{3n}\Big)\frac{\lambda\wout(v)}{\win(v)} \le \prod_{t=0}^n\Lambda_t^{a_{t,v}} \le \frac{\lambda\wout(v)}{\win(v)}\,,$
    \item $G'' \in \calF^+_{\lfloor 5(1+\lambda)\rfloor}$.
    \end{enumerate}
\end{lemma}
\begin{proof}
	Towards~\eqref{eqn:hc-ai-state}, fix $v \in V(G')$. For all $0 \le t \le n+1$, let
	\[
		x_{t,v} = \frac{\lambda\wout(v)}{\win(v)}\prod_{i=0}^{t-1}\Lambda_i^{-a_{i,v}}.
	\]
	Thus, $x_{0,v}$ is the initial value of $x_v$ in step~(\ref{HC-3}) of the algorithm, and for all $i \in [n+1]$, $x_{i,v}$ is the value of $x_v$ at the end of the iteration of the for loop in which $t=i-1$. In particular, at the end of step~\eqref{HC-3}, $x_v = x_{n+1,v}$.
	
	By definition, for all $0 \le t \le n$ we have $a_{t,v} = \floor{\log_{\Lambda_t} x_{t,v}}$ and $x_{t+1,v} = x_{t,v}\Lambda_t^{-a_{t,v}}$, so $1 \le x_{t+1,v} < \Lambda_t$. By taking $t=n$, this implies that
	\[
		1 \le \frac{\lambda\wout(v)}{\win(v)}\prod_{i=0}^n\Lambda_i^{-a_{i,v}} \le \Lambda_n\,,\qquad\text{so}\qquad
        \prod_{i=0}^n\Lambda_i^{a_{i,v}} \leq \frac{\lambda\wout(v)}{\win(v)}\leq \Lambda_n \prod_{i=0}^n\Lambda_i^{a_{i,v}}\,,
	\]
    and we have the upper bound of~\eqref{eqn:hc-ai-state}.
    Recall that $\Lambda_t = 1+\lambda(1+\lambda)^{-t}$.  For any $x>0$, $1/(1+x)\geq 1-x$, so $1/\Lambda_n \geq 1-\lambda(1+\lambda)^{-n} \geq 1-\epsilon/3n$, establishing the lower bound of~\eqref{eqn:hc-ai-state}.

	It remains to show that $\calG''\in\calF^+_{\lfloor 5(1+\lambda)\rfloor}$.  Recall that, for all $1 \le t \le n+1$, we have $1 \le x_{t,v} < \Lambda_{t-1}$. For any $z>0$, $\log (1+z)\leq z$ and, for any $z\in(0,1)$, $\log (1+z)\geq z/2$.  For all $\lambda>0$ and all $t\geq 1$, $0<\lambda(1+\lambda)^{-t}<1$.  Therefore, for all $t\geq 1$,
    \[
        a_{t,v} \le \lfloor (\log\Lambda_{t-1})/\log\Lambda_t\rfloor\leq \lfloor 2(1+\lambda)\rfloor\,.
    \]

	By the construction of $G'$ in step~(\ref{HC-2}), we have $\win(v)/\wout(v) > \eps/3n > \lambda(1+\lambda)^{-n}$. Thus $a_{0,v} = \floor{\log_{1+\lambda} \lambda\wout(v)/\win(v)} \le n$. It follows that
	\begin{align*}
		\sum_{t \in \phi(v)}(t+1) = \sum_{t=0}^n a_{t,v}(t+1) &\le n + \lfloor 2(1+\lambda)\rfloor \sum_{i=1}^{n}(i+1) \\
            &= n + \tfrac12 \lfloor 2(1+\lambda)\rfloor n(n+3)\leq \lfloor 5(1+\lambda)\rfloor n^2\,,
	\end{align*}
    so $G'' \in \calF^+_{\lfloor 5(1+\lambda)\rfloor}$, as required.
\end{proof}

\begin{lemma}\label{lem:hc-algo-2}
	Let $\lambda>0$ and let $\calF$ be a family of graphs with subcritical connective constant w.r.t.~$\lambda$. Let $\calG = (G,\win,\wout,W)$ be a non-empty $\lambda$-balanced $n$-vertex weighted graph with $G \in \calF$. Let $3n\lambda(1+\lambda)^{-n} < \eps < 1$ be rational. Then in step~(\ref{HC-4}) of $\hcalgo_{\lambda,\calF}(\calG,\eps)$,
	\begin{equation}\label{eqn:hc-algo-2-state}
		(1-\eps/3)\IS(\calG') \le Z_\lambda(G'')\, W'\lambda^{-n}\!\!\!\prod_{v \in V(G')}\!\!\!\Big(\win(v)\,\big/\prod_{i=0}^n (1+\lambda)^{i\,a_{i,v}} \Big) \le \IS(\calG')\,.
	\end{equation}
\end{lemma}
\begin{proof}
	Let $\calG' = (G',\win,\wout,W')$ be as in step~(\ref{HC-2}).  For notational convenience, let $M = Z_\lambda(G'')\,W'\lambda^{-n} \prod_{v\in V(G')} \big(\win(v)\,\big/\prod_{i=0}^n (1+\lambda)^{i\,a_{i,v}}\big)$ be the middle term of~\eqref{eqn:hc-algo-2-state}.

    Let $T$ be a star with $t$~leaves. The total weight of independent sets in~$T$ is $Z_\lambda(T) = \lambda + (1+\lambda)^t$, and the total weight of independent sets in~$T$ that do not include $T$'s centre is $(1+\lambda)^t$.  Therefore,
    \[
        Z_\lambda(G'') = \sum_{I\in\calI(G')}
            \lambda^{|I|}
            \Big(\prod_{v\in I} \prod_{i=0}^n (1+\lambda)^{i\,a_{v,i}}\Big)
            \Big(\prod_{v\in V(G')\setminus I} \prod_{i=0}^n (\lambda+(1+\lambda)^i)^{a_{v,i}}
            \Big).
    \]
    This gives
   \begin{align}
        M &= W'\lambda^{-n} \sum_{I\in\calI(G')}
                \lambda^{|I|}
                \bigg(\prod_{v\in I} \win(v)\bigg)
                \bigg(\prod_{v\in V(G')\setminus I} \win(v)
                     \prod_{i=0}^n \bigg(
                         \frac{\lambda+(1+\lambda)^i}{(1+\lambda)^i}
                     \bigg)^{a_{i,v}}
                \bigg) \notag\\
          &= W'\lambda^{-n} \sum_{I\in\calI(G')}
                \lambda^{|I|}
                \Big(\prod_{v\in I} \win(v)\Big)
                \Big(\prod_{v\in V(G')\setminus I} \win(v)
                     \prod_{i=0}^n \Lambda_i^{a_{i,v}}
                \Big)\,.\label{eqn:thingy}
   \end{align} 

   By Lemma~\ref{lem:hc-ai}, it follows that
   \[
       M \leq W'\lambda^{-n}\sum_{I\in \calI(G')}
                  \lambda^{|I|}
                  \Big(\prod_{v\in I} \win(v)\Big)
                  \Big(\prod_{v\in V(G')\setminus I} \lambda\wout(v)\Big)= Z(G')\,,
   \]
   since the $\lambda$'s cancel.  For the lower bound, \eqref{eqn:thingy} similarly implies that $M\geq (1-\epsilon/3n)^n Z(\calG')$, which is at least $(1-\epsilon/3)\,Z(\calG')$.
\end{proof}

\subsection{The proof of the theorem}

We are now ready to prove Theorem~\ref{thm:hc-algo}, which we restate for convenience.

\thmhcalgo*

\begin{proof}
We will show that   $\hcalgo_{\lambda,\calF}(\calG,\epsilon)$
is the desired FPTAS.

Let $\lambda$, $\calF$ and $\calG$ be as in the theorem statement.
Let $\epsilon$ be an error tolerance in $(0,1)$ and let  $n=|V(G)|$. If $\eps \le 3n\lambda(1+\lambda)^{-n}$, then $\hcalgo_{\lambda,\calF}(\calG,\eps)$ runs in time $\poly(2^n) = \poly(1/\epsilon)$. Suppose instead that $\eps > 3n\lambda(1+\lambda)^{-n}$. Then steps (\ref{HC-1})--(\ref{HC-4}) all take time $\poly(n,1/\epsilon)$. Moreover, by Lemma~\ref{lem:hc-ai} $G'' \in \calF^+_{\lfloor 5(1+\lambda)\rfloor}$, and $\calF^+_{\lfloor 5(1+\lambda)\rfloor}$~has subcritical connective constant w.r.t.~$\lambda$ by Lemma~\ref{lem:conn-const-2}. Thus by Theorem~\ref{thm:cc-2}, step~(\ref{HC-5}) also takes time $\poly(n,1/\epsilon)$. Thus in all cases, $\hcalgo_{\lambda,\calF}(\calG,\eps)$ runs in time $\poly(n,1/\epsilon)$ as required.
	
	By Lemma~\ref{lem:hcalgo-1},
	$\IS(\calG) \ge \IS(\calG') \ge (1-\eps/3)\IS(\calG)$.
	By Lemma~\ref{lem:hc-algo-2}, it follows that
	\[
		\IS(\calG) \ge Z_\lambda(G'')\, W'\lambda^{-n}\!\!\!\prod_{v \in V(G')}\!\!\!\Big(\win(v)\,\big/\prod_{i=0}^n (1+\lambda)^{i\,a_{i,v}} \Big) \ge (1-\eps/3)^2\IS(\calG)\,.
	\]
	By the correctness of $\hcunialgo_{\lambda,\calF^+_{\lfloor 5(1+\lambda)\rfloor}}$ (Theorem~\ref{thm:cc-2}), it follows that
	\begin{multline*}
		(1+\eps/3)\IS(\calG) \ge \hcalgo_{\lambda,\calF^+_{\lfloor 5(1+\lambda)\rfloor}}(G'',\eps/3)\cdot W'\lambda^{-n}\!\!\!\prod_{v \in V(G')}\!\!\!\Big(\win(v)\,\big/\prod_{i=0}^n (1+\lambda)^{i\,a_{i,v}} \Big) \\\ge (1-\eps/3)^3\IS(\calG)\,.
	\end{multline*}
	Thus the output of $\hcalgo_{\lambda,\calF}(\calG,\eps)$ lies in $(1\pm\eps)\IS(\calG)$, as required.
\end{proof}

\section{The connective constant and local structure}
\label{sec:local}
\label{sec:conn-const}	

In this section, we will give an FPTAS for $Z(\calG)$ in the case where the input weighted graph~$\calG$ has an underlying graph~$G$ which comes from the following family.

\begin{defn}\label{defn:ourfamily}
    Let $\ourfamily$ be the family of graphs~$G$ such that
    $\delta(G)\geq 2$ and, for all $x\in V(G)$ with $d(x)\geq 6$,
    $d^2(x)\leq 26$.
\end{defn}

Suppose we are given a weighted graph $\calG=(G,\win,\wout,W)$ with $\delta(G)\geq 2$ and, for all $v\in V(G)$, $\win(v)\leq\wout(v)$ (i.e., $\calG$~is $1$-balanced).  Then either we can approximate $\IS(\calG)$ relatively quickly using 
the FPTAS from this section, or $\calG$~contains a vertex of degree at least~$6$ and $2$-degree at least $27$ --- which our main algorithm will exploit. Thus $\ourfamily$ will serve as the base case for our algorithm (see Section~\ref{sec:slices}). We will show that $\ourfamily$ has subcritical connective constant w.r.t.\@ $\lambda=1$, so that the existence of an FPTAS follows immediately from Theorem~\ref{thm:hc-algo}; see Theorem~\ref{thm:2-degree-FPTAS}.

Recall from Section~\ref{sec:gadgets} that to show $\ourfamily$ has subcritical connective constant, we must bound the number of leaves in the self-avoiding walk trees $\sawt{G}{v}$ of graphs $G \in \ourfamily$. To this end, we first introduce some notation and ancillary lemmas.

\label{def:treedefs}
Let $T=(V,E,r)$ be a tree $(V,E)$ with root~$r\in V$.  We write $L(T)$ for
the set of $T$'s leaves.  For any $v\in V$, we write $C(v)$ for the
set of $v$'s children (i.e., $v$'s neighbours that are not on the
unique path from $v$ to~$r$; we take $C(r)=\Gamma(r)$), and $D(v)$
for $v$'s depth (i.e., the length of the path from $v$ to~$r$, so
$D(r)=0$).
	
\begin{defn}\label{def:kapdec}
    Let $T=(V,E,r)$ be a rooted tree and let $\kappa\geq 1$.  A
    function $\theta\colon V\to \mathbb{R}$ is \emph{$\kappa$-decreasing
      for~$T$} if, for every $x\in V$,
    $1/\kappa \leq \theta(x)\leq 1$ and, for every $x\in V\setminus\{r\}$,
    \begin{equation}
        \sum_{y\in C(x)}\!\! \theta(y)\leq \kappa\,\theta(x)\,.\label{eq:kappa-dec}
    \end{equation}
\end{defn}
	
The specific requirement that $\theta(x)\geq 1/\kappa$ is for
convenience.  We will use a family of $\kappa$-decreasing functions
for trees in a particular family, and any positive lower bound on the
values taken by all of these functions would suffice.

\begin{lemma}
\label{lem:kappa-squared}
    Let $G$ be a graph and let $\kappa>1$.  Suppose that, for every
    $v\in V(G)$, there is a $\kappa$-decreasing function for
    $\sawt{G}{v}$.  Then $\Delta(G)\leq \kappa^2+1$.
\end{lemma}
\begin{proof}
    Towards a contradiction, suppose that $d_G(u)>\kappa^2+1$ for some
    $u\in V(G)$.  $d_G(u)>1$ so $u$~has a neighbour~$v$.  Let $\theta$
    be a $\kappa$-decreasing function for $\sawt{G}{v}$.  Now consider
    the node~$x$ of $\sawt{G}{v}$ corresponding to the walk~$vu$
    in~$G$.  For each $w\in\Gamma_G(u)\setminus\{v\}$, $x$~has a child
    corresponding to the walk $vuw$, so it has more than $\kappa^2$
    children.  Therefore,
    $\sum_{y\in C(x)}\theta(y) > \kappa^2\,\tfrac{1}{\kappa}\geq
    \kappa\,\theta(x)$, contradicting~\eqref{eq:kappa-dec}.
\end{proof}

The following lemma is a useful way to bound connective constants.

\begin{lemma}
\label{lem:kappa-dec}
    Let $\calF$ be a family of graphs. Let $\kappa>1$ and suppose
    that, for every $G\in\calF$ and every $v\in V(G)$, there is a
    $\kappa$-decreasing function for $\sawt{G}{v}$.  Then the
    connective constant of~$\calF$ is at most~$\kappa$.
\end{lemma}
\begin{proof}
    Let $G\in\calF$, let $v\in V(G)$, let $T=\sawt{G}{v}$, let $r$ be
    $T$'s root and let $\theta$ be $\kappa$-decreasing for~$T$.
    For any $i\geq 0$, let $T_i$ be the subgraph of~$T$ induced
    by vertices at distance at most~$i$ from~$r$.  We claim that, for
    all $i\geq 0$, 
    \[
        S_i := \!\!\!\sum_{w\in L(T_i)}\!\!\!\theta(w)\,\kappa^{-D(w)}
        < 2\kappa\,.
    \]
    We have $S_0 = \theta(r)\leq 1<2\kappa$.  By
    Lemma~\ref{lem:kappa-squared}, $d_G(v)\leq \kappa^2+1$, so
    $|L(T_1)|\leq \kappa^2+1$ and
    $S_1\leq (\kappa^2+1)/\kappa<2\kappa$. Finally, suppose
    $S_i< 2\kappa$ for some $i\geq 1$ and consider any leaf~$w$
    of~$T_i$.  If $w$~is a leaf in~$T_{i+1}$, then it contributes the
    same amount to $S_i$ and~$S_{i+1}$.  If $w$~is not a leaf
    in~$T_{i+1}$, then it does not contribute to~$S_{i+1}$ but all of
    its children do.  Their contribution is
    \[
        \sum_{y\in C(w)}\!\!\! \theta(y)\,\kappa^{-D(y)}
            = \kappa^{-(D(w)+1)} \!\!\sum_{y\in C(w)}\!\!\theta(y)
            \leq \kappa^{-(D(w)+1)}\kappa\,\theta(w)\,,
    \]
    where the inequality is by~\eqref{eq:kappa-dec}.  But
    $\kappa^{-D(w)}\,\theta(w)$ is exactly the contribution of~$w$
    to~$T_i$, so $S_{i+1}\leq S_i<2\kappa$ as claimed.
    
    $N(v,i)$ is the number of vertices
    in~$T$ with depth~$i$ so
    \[
        N(v,i) \leq |L(T_i)|
        \leq \kappa\!\!\! \sum_{w\in L(T_i)}\!\!\!\theta(w)
        \leq \kappa^{i+1} \!\!\!\sum_{w\in L(T_i)}\!\!\!\theta(w)\,\kappa^{-D(w)}
        = \kappa^{i+1}\,S_i
        < 2\kappa^{i+2}\,.
    \]
    So, for all $\ell\geq 1$,
    \[
        \sum_{i=1}^\ell N(v,i) < \sum_{i=1}^\ell 2\kappa^{i+2}
        < \frac{2\kappa^3}{\kappa-1}\,{\kappa^\ell}\,.
    \]
    In particular, this holds for all $\ell\geq\log |V(G)|$, so
    the connective constant of~$\calF$ is at most~$\kappa$.
\end{proof}

We are now in a position to prove the main result of the section.

\begin{theorem}
\label{thm:2-degree-FPTAS}
    There is an FPTAS $\basecount(\calG,\epsilon)$ for $\IS(\calG)$ on $1$-balanced weighted graphs $\calG=(G,\win,\wout,W)$ with $G\in \ourfamily$.
\end{theorem}
\begin{proof}
    We use $4.141$-decreasing functions to show that $\ourfamily$ has
    connective constant at most $4.141$. This is subcritical w.r.t.\@ $\lambda=1$, so 
    the algorithm $\hcalgo_{1,\ourfamily}$ is an FPTAS as shown in the proof of
    Theorem~\ref{thm:hc-algo} and the result follows.

    Let $G\in\ourfamily$.  The restrictions on $\delta(G)$ and maximum
    $2$-degree imply that $\Delta(G)\leq 13$.  Let $v\in V(G)$ and let
    $T=\sawt{G}{v}$.
	
    Let $\psi(d,p)$ be the following function for integers $d,p\geq 2$.
    \begin{alignat*}{4}
        \psi(2,p)     &= 0.245 &&\text{ for all }p\geq 2 &\qquad
            \psi(d,2) &= 1     &&\text{ for all }d\geq 6 \\
        \psi(3,p)     &= 0.456 &&\text{ for all }p\geq 2 &
            \psi(d,3) &= 0.941 &&\text{ for all }d\geq 6 \\
        \psi(4,p)     &= 0.647 &&\text{ for all }p\geq 2 &
            \psi(d,p) &= 0.889 &&\text{ for all }d\geq 6,\,p\geq 4\,.\\
        \psi(5,p)     &= 0.859 &&\text{ for all }p\geq 2
    \end{alignat*}
	
    Define the function $\theta\colon V(T)\to \mathbb{R}$ by setting
    $\theta(x)=1$ if $x$~is the root of~$T$ and, otherwise, setting
    $\theta(x)=\psi(d_T(x),d_T(p(x)))$, where $p(x)$ is $x$'s parent
    in~$T$.

    We claim that $\theta$~is $4.141$-decreasing for~$T$, from which the
    result follows by Lemma~\ref{lem:kappa-dec} and
    Theorem~\ref{thm:hc-algo}. It is immediate that
    $1/4.141 < 0.242 \leq \theta(x) \leq 1$ for all $x\in V(T)$ so
    it remains to show that \eqref{eq:kappa-dec}~holds for every
    $x\in V(T)\setminus\{r\}$.  We verify this with the Mathematica
    program given in Appendix~\ref{app:mathematica:conn-const}.

    Consider a vertex $x\in V(T)\setminus\{r\}$.  Let
    $d_T(x)=d$.  For each $i\in\{2, \dots, 13\}$, let $n[i]$ be the number of
    children of~$x$ with degree~$i$, and let $p=d_T(p(x))$.  Then
    \eqref{eq:kappa-dec}~becomes
    \begin{equation}
        \sum_{i=2}^{13} n[i]\,\psi(i,d) \leq 4.141\,\psi(d,p)\,.
        \label{eq:kappa-check}
    \end{equation}

    The program attempts to find some combination of
    $n[2], \dots, n[13]\in\{0, \dots, 12\}$ and
    $d, p\in\{2, \dots, 13\}$ that can arise from a vertex in~$T$ and
    that does not satisfy~\eqref{eq:kappa-check}, i.e., according to Definition~\ref{defn:ourfamily} it checks every
    combination that has $1+\sum_i n[i]=d$ and that has $d\leq 5$ or
    $p+\sum_i i\,n[i]\leq 26$.  The program outputs ``\texttt{\{\}}'',
    indicating that every instantiation of~\eqref{eq:kappa-check} is
    satisfied.  The program uses exact arithmetic, with no numerical approximations, so this is a rigorous check. Therefore, by Lemma~\ref{lem:kappa-dec}, the connective constant of~$\ourfamily$ is
    at most $4.141$ as claimed.
\end{proof}

	\section{\texorpdfstring{Approximating $\boldsymbol{Z(\calG)}$ for arbitrary graphs}{Approximating Z(G) for arbitrary graphs}}\label{sec:slices}

	\subsection{Preprocessing}\label{sec:preproc}

To analyse our algorithm's running time, we use a continuous piecewise-linear potential function whose linear ``slices'' have the following form.

\begin{defn}\label{defn:good-slice}
A function $f\colon\R^2\to\R$ is a \emph{good slice} if it is of the form $f(m,n) = \rho m+\sigma n$ for some $\rho \ge 0$ that satisfies $\rho \ge -\sigma$ and which is not identically zero (that is, it is not the case that $\sigma=\rho=0$).  
Given any good slice $f$, we also use the name $f$ to denote the function 
from graphs to real numbers given  by $f(G)=f(|E(G)|,|V(G)|)$. Likewise, if $\calG=(G,\win,\wout,W)$ is a weighted graph, we write $f(\calG) = f(G)$. 
	\end{defn}
	
Note that we allow $\sigma$ to be negative as long as $\sigma \ge -\rho$, and indeed our final potential function will contain a slice with $\sigma = -\rho$ (see Section~\ref{sec:pre-pot}). This essentially encodes a well-known trick: we can remove vertices of degree at most 1 from our input graph $\calG=(G,\win,\wout,W)$ in polynomial time (see  Corollary~\ref{cor:prune-fast}). If after doing this, $G$ has average degree 2, then $G$ must be 2-regular, so we can compute $\IS(\calG)$ in polynomial time. Thus we can safely set $f(G) = 0$ when $|E(G)| = |V(G)|$.
	
	Unfortunately, taking $\sigma=-\rho$ introduces technical difficulties. In order for $f$ to be useful for analysing the running time of our algorithm, we require that whenever the algorithm makes recursive calls with instances $\Gin$ and $\Gout$ and original input $\calG$, we have $f(\Gin)< f(\calG)$ and $f(\Gout) < f(\calG)$. If $G$ contains a tree component $T$, then repeatedly removing degree-1 vertices from $G$ will remove all of $T$. We then have
	\[
		f(G-V(T)) = f(G) - \rho |E(T)| - \sigma |V(T)| = f(G) + \rho > f(G).
	\]
	Thus removing degree-1 vertices may increase $G$'s potential, so we must be very careful to keep track of how many tree components we create as we branch. For this reason, our preprocessing doesn't just remove degree-1 vertices. Instead, it prunes $G$ in other ways to decrease the number of tree components 
	that are formed. Our main operation incorporates what is often referred to \emph{multiplier reduction} in the literature~\cite{DJW}, and is given below. At the end of the section we will describe a further, simpler operation of removing tree components.

	\begin{defn}\label{defn:prune} 
Let $\calG = (G,\win,\wout,W)$ be a weighted graph.
Let $S$ be a subset of $V(G)$ satisfying $|\Gamma_G(S)| \le 1$. If $\Gamma_G(S) = \emptyset$, then $\Prune(\calG,S)$ is the weighted graph $(G-S,\win,\wout,\allowbreak \IS(G[S],\win,\wout,W))$. If instead $\Gamma_G(S) = \{v\}$, then let $S' = S \cap \Gamma_G(v)$ and let $S'' = S \setminus \Gamma_G(v)$. We define $\Prune(\calG,S)$ to be the weighted graph $(G-S,\win',\wout',W)$, where
		\begin{align*}
			\win'(x) &= \begin{cases}
				\win(v)\,\wout( S')\,\IS(G[S''],\win,\wout,1) & \mbox{ if }x=v,\\
				\win(x) & \mbox{ otherwise,}
			\end{cases}\\
			\wout'(x) &= \begin{cases}
				\wout(v)\, \IS(G[S],\win,\wout,1) & \mbox{ if }x=v,\\
				\wout(x) & \mbox{ otherwise.}
			\end{cases}
		\end{align*}
	\end{defn}

\begin{lemma}\label{lem:MR}
Let $\calG = (G,\win,\wout,W)$ be a $1$-balanced weighted graph, and let $S \subseteq V(G)$ with $|\Gamma_G(S)| \le 1$. Then $\IS(\Prune(\calG,S)) = \IS(\calG)$ and $\Prune(\calG,S)$ is $1$-balanced.
\end{lemma}
\begin{proof}
If $\Gamma_G(S) = \emptyset$, then the result is immediate, so suppose instead that $\Gamma_G(S) = \{v\}$ for some $v \in V(G)$. Write $\calG' = \Prune(\calG,S) = (G-S,\win',\wout',W')$.
Partitioning sets $I \in \calI(\calG)$ according to whether or not $v \in I$, we have
		\[
			\IS(\calG) = \win(v)\,\wout(\Gamma_G(v))\,\IS(\calG-v-\Gamma_G(v)) + \wout(v)\,\IS(\calG-v).
		\]
Similarly, using the fact that $\Gamma_G(S) = \{v\}$ and Definition~\ref{defn:prune},
		\begin{align*}
			\IS(\calG) &= \win(v)\,\wout(\Gamma_G(v))\,
			 \IS(G[S''],\win,\wout,1)\,\IS(\calG-S-v-\Gamma_G(v))\\
			  &\qquad\qquad\qquad\qquad\qquad\qquad\qquad + \wout(v)\,\IS(\calG-S-v)\,\IS(G[S],\win,\wout,1)\\
			&= \win'(v)\,\wout(\Gamma_G(v) \setminus S)\,\IS(\calG-S-v-\Gamma_G(v)) + \wout'(v)\,\IS(\calG-S-v)
			= \IS(\calG')\,.
		\end{align*}

        It remains to show that $\Prune(\calG,S)$ is $1$-balanced. We have $\win(x)\leq\wout(x)$ for all $x\in V(G)$ and we must show that $\win'(x)\leq\wout'(x)$ for all $x\in V(G')$.  So, let $x\in V(G')$.  If $x\neq v$ then $\win'(x)=\win(x)\leq\wout(x)=\wout'(x)$ and we are done.
    Otherwise, $x=v$. We have
    \begin{align*}
        \wout(S')\,Z(G[S''],\win,\wout,1)
           \ &= \!\!\!\sum_{\substack{I\in\calI(G[S])\\I\cap S'=\emptyset}}\!\!\! \win(I)\wout(S\setminus I)
           \ \leq \!\!\!\sum_{I\in\calI(G[S])}\!\!\! \win(I)\wout(S\setminus I)\\
           &=\ Z(G[S],\win,\wout,1)\,,
    \end{align*}
    and $\win(v)\leq\wout(v)$ by hypothesis.  Therefore, $\win'(v)\leq \wout'(v)$.
\end{proof}
	
In order to compute the weighted graph $\Prune(\calG,S)$ we need to evaluate $\IS(\calG[S])$ and possibly $\IS(\calG[S''])$. 
For this, we will use the following brute-force algorithm.
	
	\begin{algorr}{$\brute(\calG,Y)$} {\label{algo:brute}The input is a   weighted graph $\calG=(G,\win,\wout,W)$ and a 
subset $Y$ of $V(G)$  such that $\calG-Y$ is a forest.}
	\item For all $I \in \calI(G[Y])$, compute $\IS(\calG-Y-\Gamma_G(I))$.
	\item Return
		$\displaystyle \sum_{I\in \calI(G[Y])}\!\!\!\!\win(I)\, \wout(Y\setminus I)\, \IS\big(\calG - Y - \Gamma_G(I)\big)$.
	\end{algorr}
	
	\begin{obs}\label{obs:unstable-easy}
	If  $\calG=(G,\win,\wout,W)$ 
is a weighted graph
and $Y$
is a subset of $V(G)$   such that $\calG-Y$ is a forest	
then	
	$\brute(\calG,Y)$	
		returns $\IS(\calG)$ and has running time $2^{|Y|}\,\poly(|V(G)|)$.
	\end{obs}

\subsection{A digression}	\label{sec:digress}

As noted in Section~\ref{sec:local} 
our main algorithm
will use the  FPTAS from Theorem~\ref{thm:2-degree-FPTAS}
 which approximates $\IS(\calG)$ on $1$-balanced weighted graphs 
 $\calG=(G,\win,\wout,W)$ with $G\in \ourfamily$.
(Recall from Definition~\ref{defn:ourfamily}
that $\ourfamily$ is the family of graphs with minimum degree at least~$2$
such that every vertex with degree at least~$6$ has $2$-degree at most~$26$.)

In the introduction we stated Theorem~\ref{thm:lastintro}, which is 
essentially an unweighted version of  Theorem~\ref{thm:2-degree-FPTAS}.
The only difference is that, in Theorem~\ref{thm:lastintro},
 inputs are allowed
to have vertices with degree less than~$2$.
Given the preprocessing that we have just done, we can use Theorem~\ref{thm:2-degree-FPTAS} to
prove   Theorem~\ref{thm:lastintro}, which we restate for convenience.

\lastintro*
\begin{proof}
Consider a graph~$G$ 
in which every vertex of degree at least~$6$ has 2-degree at most~$26$. 
We will show how to construct a $1$-balanced weighted graph
 $\calG=(G',\win,\wout,W)$ with $G'\in \ourfamily$ such that 
 $\IS(\calG) = \IS(G)$. We can then use the FPTAS from Theorem~\ref{thm:2-degree-FPTAS}.
  
The construction is simple. Let $\calG = (G,\mathbf{1},\mathbf{1},1)$.
Now repeat the following.
If the minimum degree of~$\calG$ is at least~$2$, we are finished.
Otherwise, let $S$ be the set containing any 
vertex of~$G$ whose degree is less than~$2$. 
By Lemma~\ref{lem:MR}, 
$\IS(\Prune(\calG,S)) = \IS(\calG)$ and $\Prune(\calG,S)$ is $1$-balanced.
Since $\calG$ has the property that
any vertex of degree at least~$6$ has 2-degree at most~$26$,
so does 
$\Prune(\calG,S)$. Replace $\calG$ with 
$\Prune(\calG,S)$.  
\end{proof}

\subsection{More preprocessing}	\label{sec:morepreproc}
	
Our preprocessing algorithm will prune its input graph by removing non-empty induced subgraphs of the following form.
	
	\begin{defn}\label{def:near-forest}
		We say that a graph is a \emph{near-forest} if it is the empty graph (with no vertices)  or if it contains a vertex $v$ such that $G - v - \Gamma_G(v)$ is a forest. 
	\end{defn}
	
In this definition (and everywhere else) we consider the empty graph 
  to be a forest. Thus, a forest is a near-forest.
The point of  Definition~\ref{def:near-forest} is the following corollary.

\begin{cor}\label{cor:prune-fast}
	Let $\calG=(G,\win,\wout,W)$ be an $n$-vertex weighted graph with maximum degree at most~$\Delta$. Let $X \subseteq V(G)$ span a near-forest in $G$ with $|\Gamma_G(X)| \le 1$. Then $\Prune(G,X)$ can be computed in time at most $2^{\Delta}\,\poly(n)$.
\end{cor}
\begin{proof}
	If $G[X]$ is a forest, let $S = \emptyset$. Otherwise, there exists $v \in X$ such that $G[X]-v-\Gamma_G(v)$ is a forest; take $S = \{v\} \cup \Gamma_G(v)$. Thus $G[X] - S$ is a forest, so by Observation~\ref{obs:unstable-easy}, for all $Y \subseteq X$ we have $\IS(\calG[Y]) = \brute(\calG[Y],S)$. Again by Observation~\ref{obs:unstable-easy}, it follows that $\IS(\calG[Y])$ can be computed in time $2^\Delta\,\poly(n)$ for all $Y \subseteq X$. In particular, this allows us to compute $\Prune(\calG,X)$ in $\poly(n)$ time using Definition~\ref{defn:prune}.
\end{proof}

The output of our preprocessing algorithm will be of the following form.

	\begin{defn}\label{def:reduced}
		A graph $G$ is \emph{reduced} if for all 
non-empty subsets $X$ of $V(G)$   such that $G[X]$ is a near-forest, $|\Gamma_G(X)| \ge 2$.
	\end{defn}

Note that a reduced graph $G$ is either empty or satisfies $\delta(G) \ge 2$, and it has no components which are near-forests. 	
We now set out our preprocessing algorithm and prove its correctness.

	\begin{algorr}{$\Red(\calG)$}{\label{algo:reduce}
	The input is 
	a   non-empty weighted graph $\calG$ with no tree components. For each weighted graph $\calG_i$ that we construct in this algorithm, we write $\calG_i = (G_i, \win^i, \wout^i, W_i)$.}
		\item\label{countfew:oracle} \label{redstep1}
		Set $i=1$ and $\calG_1 = \calG$.
		\item\label{countfew:oracle23} \label{redstep2}
		If either $G_i$ or some subgraph $G_i-v$ (where $v \in V(G_i)$) contains a component $C$ which is a near-forest, then:
		\begin{enumerate}[(a)]
			\item Let $\calG_{i+1} = \Prune(\calG_i,V(C))$.\label{redstep2a}
			\item Increment $i$ and go to (ii).
		\end{enumerate}
		\item\label{countfew:oracle4} \label{redstep4} Otherwise, return $\calG_i$.
	\end{algorr}

\begin{lemma}\label{lem:poly-reduce}
Let $\calG = (G,\win,\wout,W)$ be a $1$-balanced $n$-vertex weighted graph 
with no tree components and maximum degree at most $\Delta$. Then
$\Red(\calG)$ returns a $1$-balanced reduced weighted graph $\calH$ such that 
\begin{itemize}
\item	$\IS(\calG) = \IS(\calH)$, 
\item   $\calH$ has maximum degree at most $\Delta$, and
\item   for any good slice $f$, $f(\calH) \le f(\calG)$. 
\end{itemize} The running time  is $2^\Delta \,\poly(n)$.	
\end{lemma}
\begin{proof}
Observe that each graph $G_i$ is produced by repeated calls to $\Prune$ (which meet the requirements of Definition~\ref{defn:prune}). Thus each $\calG_i$ has maximum degree at most $\Delta$ and, by Lemma~\ref{lem:MR}, is $1$-balanced and satisfies $\IS(G_i) = \IS(\calG)$.  Suppose that $G_i$ does not satisfy the condition in step~(ii) of Algorithm $\Red$.  We want to show, using Definition~\ref{def:reduced}, that it is reduced.  To see this, note that if $X \subseteq V(G_i)$ spans a non-empty near-forest in~$G_i$ with $|\Gamma_{G_i}(X)| \le 1$, then $G_i[X]$ also contains a near-forest component; thus the algorithm continues and constructs $G_{i+1}$ from~$G_i$. We conclude that $\calH$~is reduced.  Note that $|V(G_1)| > |V(G_2)| > \dots > |V(\calH)|$, so step~(ii) is iterated at most $n$~times. By Corollary~\ref{cor:prune-fast}, each iteration takes time at most $2^\Delta\poly(n) = \poly(n)$, so our stated time bound follows.

It remains only to prove that $f(\calH) \le f(\calG)$ for all good slices $f$. Write $f(m,n) = \rho m + \sigma n$ and $\calH = \calG_s$. Let $C_1, \dots, C_r$ be the components of $G$, and for all $i \in [r]$ and $j \in [s]$, let $C_{i,j} = G_j[V(C_i) \cap V(G_j)]$ be the remaining part of~$C_i$ in~$G_j$. Thus
\begin{equation}\label{eqn:poly-reduce}
	f(\calG) - f(\calH) = \sum_{i=1}^r \big(f(C_i) - f(C_{i,s})\big).
\end{equation}
We will show that each term of this sum is non-negative. 

Since each application of $\Prune$ deletes a vertex set with at most one neighbour, it can never remove a separator; thus for all $i \in [r]$, either $C_{i,s}$ is empty or $C_{i,j}$ is a component of $G_j$ for all $j \in [s]$. If $C_{i,s}$ is empty, then since $\calG$ contains no tree components by hypothesis and $f$ is a good slice, we have 
\[
	f(C_i) - f(C_{i,s}) = \rho e(C_i) + \sigma |C_i| \ge (\rho+\sigma)|C_i| \ge 0.
\]
Suppose instead $C_{i,j}$ is a component of $G_j$ for all $j \in [s]$. Then for all $2 \le j \le s$, either $C_{i,j} = C_{i,j-1}$ or $C_{i,j}$ is formed from $C_{i,j-1}$ by deleting a connected set sending at least one edge into $V(C_{i,j})$. In either case, at least as many edges are deleted as vertices. Thus
\begin{align*}
	f(C_i) - f(C_{i,s}) 
	&= \sum_{j=2}^s \big(f(C_{i,j-1}) - f(C_{i,j})\big)\\
	&= \sum_{j=2}^s \Big(\rho\big(e(C_{i,j-1}) - e(C_{i,j})\big)  + \sigma(|C_{i,j-1}| - |C_{i,j}|\big)\Big)\\
	&\ge \sum_{j=2}^s (\rho+\sigma)\big(|C_{i,j-1}|-|C_{i,j}|\big)
	\ge 0.
\end{align*}
The result therefore follows from~\eqref{eqn:poly-reduce}.
\end{proof}

	 Having shown how to reduce graphs using the pruning operation we now define the final operation, tree removal.

	\begin{defn}\label{def:TR}
		Let $G$ be a graph with tree components $T_1,\dots,T_r$. Then we define the \emph{tree removal of $G$} by $\TR(G) = G - V(T_1) - \dots - V(T_r)$. Moreover, if $\calG = (G,\win,\wout,W)$ is a weighted graph, we define the \emph{tree removal of $\calG$} by
		\[
			\TR(\calG) = \Big(\TR(G),\win,\wout,W\cdot \prod_{i=1}^r \brute(\calT_i,\emptyset)\Big),
		\]
        where $\calT_i = (T_i, \win, \wout, 1)$.
	\end{defn}	

    The following observation is immediate from the definition and Observation~\ref{obs:unstable-easy}.

	\begin{obs}\label{obs:TR-Z}
        Let $\calG$ be a weighted graph with underlying graph~$G$ and let $\TR(\calG)$ have underlying graph~$H$.  Then,
        \begin{enumerate}[(i)]
        \item $H\subseteq G$,
        \item $\IS(\TR(\calG)) = \IS(\calG)$,
        \item if $\calG$ is $1$-balanced then $\calH$~is also $1$-balanced, and
        \item $\TR(\calG)$ can be computed in time $\poly(|V(G)|)$. \qed
        \end{enumerate}
	\end{obs}
	
	\subsection{Graph decompositions}\label{sec:decomp}

	To describe our algorithm, we use two decompositions, which we call the standard decomposition and the extended decomposition. Recall from Sections~\ref{sec:preproc} and~\ref{sec:morepreproc} that it will be important for us to keep track of the tree components formed as we proceed; for this reason, we introduce some notation for the local structure of a graph around a vertex.
	
	\begin{defn}\label{def:decomp}
		Let $G$ be a connected graph, and let $v \in V(G)$. We define the \emph{standard decomposition} of $G$ from $v$ to be the tuple $(\Gamma_v,S,X,H_1,\dots,H_k;T_1,\dots,T_\ell)$, where:
		\begin{itemize}
			\item $\Gamma_v = \Gamma_G(v)$;
			\item $H_1, \dots, H_k$ are the non-tree components of $G-v-\Gamma_v$;
			\item $T_1, \dots, T_\ell$ are the tree components of $G-v-\Gamma_v$;
			\item $S = (V(H_1) \cup \dots \cup V(H_k)) \cap \Gamma^2_G(v)$;
			\item $X = \{v\} \cup \Gamma_v \cup (V(T_1) \cup \dots \cup V(T_\ell))$.
		\end{itemize}
	\end{defn}
	
	\begin{figure}
\begin{center}
\begin{tikzpicture}[scale=1,node distance = 1.5cm]
    \tikzstyle{dot}=[fill=black, draw=black, circle, inner sep=0.1mm]
    \tikzstyle{vtx}=[fill=black, draw=black, circle, inner sep=2pt]
    \tikzstyle{vset}=[gray,rounded corners=5pt,dashed];

    %
    %
    \node[vtx] (v)  at (5,4.5) [label=90:{$v$}] {};
    \node[vtx] (g1) at (1,3.5) {};
    \node[vtx] (g2) at (2,3.5) {};
    \node[vtx] (g3) at (3,3.5) {};
    \node[vtx] (g4) at (5,3.5) {};
    \node[vtx] (g5) at (6,3.5) {};
    \node[vtx] (g6) at (7,3.5) {};
    \node[vtx] (g7) at (9,3.5) {};

    \node[dot] (dot1) at ($(g3)!0.5!(g4)$) {};
    \node[dot] at ($(dot1)+(-0.15,0)$)     {};
    \node[dot] at ($(dot1)+(0.15,0)$)      {};
    \node[dot] (dot2) at ($(g6)!0.5!(g7)$) {};
    \node[dot] at ($(dot2)+(-0.15,0)$)     {};
    \node[dot] at ($(dot2)+(0.15,0)$)      {};

    \draw (g1) .. controls ($(g1)+(310:0.5)$) and ($(g3)+(230:0.5)$) .. (g3);
    \draw (g4)--(g5);

    \draw[vset] (-0.3,4) rectangle (10.75,3);
    \node at (-0.3,3.5) [label=0:{$\Gamma_v$}] {};

    \foreach \i in {1, ..., 7}
        \draw (v)--(g\i);

    %
    %
    \node[vtx] (a1) at (0.5,2) {};
    \node[vtx] (a2) at (0.5,1) {};
    \node[vtx] (a3) at ($(a2)+(240:1)$) {};
    \node[vtx] (a4) at ($(a2)+(300:1)$) {};
    \draw (a1)--(a2)--(a3)--(a4)--(a2);
    \draw[vset] (-0.3,2.3) rectangle (1.3,-0.75);
    \node at (-0.4,-0.4) [label=0:{$H_1$}] {};

    %
    %
    \node[vtx] (b1) at (1.7,2) {};
    \node[vtx] (b2) at (2.7,2) {};
    \node[vtx] (b3) at (2.7,1) {};
    \node[vtx] (b4) at (1.7,1) {};
    \draw (b2)--(b3)--(b4)--(b1)--(b3);
    \draw (b3) -- +(260:0.7);
    \draw (b3) -- +(220:0.4);
    \draw (b4) -- +(310:0.5);
    \draw[vset] (1.45,2.3) rectangle (2.95,-0.75);
    \node at (1.35,-0.4) [label=0:{$H_2$}] {};

    \node[dot] at (3.20,0.8) {};
    \node[dot] at (3.35,0.8) {};
    \node[dot] at (3.50,0.8) {};

    %
    %
    \node[vtx] (c1) at (4,2) {};
    \node[vtx] (c2) at (5,2) {};
    \node[vtx] (c3) at (5,1) {};
    \node[vtx] (c4) at (4,1) {};
    \node[vtx] (c5) at ($(c4)+(300:1)$) {};
    \draw (c3)--(c4)--(c5)--(c3)--(c2)--(c1)--(c4);
    \draw (c3) -- +(270:0.6);
    \draw (c4) -- +(280:0.6);
    \draw (c4) -- +(330:0.4);
    \draw (c5) -- +(320:0.5);
    \draw[vset] (3.75,2.3) rectangle (5.25,-0.75);
    \node at (3.65,-0.4) [label=0:{$H_k$}] {};

    %
    %
    \draw[vset][color=red] (-1,2.5) rectangle (5.45,1.5);
    \node at (-1.05,2) [label={[red]0:{$S$}}] {};

    %
    %
    \node[vtx] (d1) at (6.5,2) {};
    \node[vtx] (d2) at (6.5,1) {};
    \draw (d1) -- (d2);
    \draw (d2) -- +(240:0.8);
    \draw ($(d2)+(240:0.5)$) -- +(300:0.7);
    \draw (d2) -- +(315:0.6);
    \draw[vset] (5.9,2.3) rectangle (7.1,-0.75);
    \node at (5.8,-0.4) [label=0:{$T_1$}] {};

    %
    %
    \node[vtx] (e1) at (7.5,2) {};
    \node[vtx] (e2) at (8.5,2) {};
    \node[vtx] (e3) at (8,1) {};
    \draw (e1)--(e3)--(e2);
    \draw (e3) -- +(240:0.8);
    \draw (e3) -- +(330:0.4);
    \draw (e3) -- +(280:1);
    \draw ($(e3)+(280:0.4)$) -- +(250:0.7);
    \draw[vset] (7.25,2.3) rectangle (8.75,-0.75);
    \node at (7.15,-0.4) [label=0:{$T_2$}] {};

    \node[dot] at (9.00,0.8) {};
    \node[dot] at (9.15,0.8) {};
    \node[dot] at (9.30,0.8) {};

    %
    %
    \node[vtx] (f1) at (10.3,2) {};
    \node[vtx] (f2) at  (9.8,1) {};
    \node[vtx] (f3) at (10.8,1) {};
    \draw (f2)--(f1)--(f3);
    \draw (f2) -- +(270:0.6);
    \draw (f2) -- +(290:0.7);
    \draw (f3) -- +(200:0.4);
    \draw (f3) -- +(250:0.6);
    \draw[vset] (9.55,2.3) rectangle (11.05,-0.75);
    \node at (9.45,-0.4) [label=0:{$T_\ell$}] {};

    %
    %
    \draw[vset][draw=blue] (-1,3.5) -- (-1,5.25) -- (11.3,5.25) -- (11.3,-1)
                  -- (5.65,-1) -- (5.65,2.75) -- (-1,2.75) -- (-1,3.5);
    \node at (-1.05,4.75) [label={[blue]0:{$X$}}] {};

    %
    %
    \draw (a1)--(g1)--(b1);
    \draw (b2)--(g2)--(c1);
    \draw (b2)--(g3)--(c1);
    \draw (c2)--(g5);
    \draw (g3)--(d1)--(g4);
    \draw (e1)--(g5)--(e2);
    \draw (g7)--(f1);

\end{tikzpicture}
\end{center}

		\caption{An example of the standard decomposition of a connected graph $G$ from a vertex~$v$.}\label{fig:decomp}
	\end{figure}
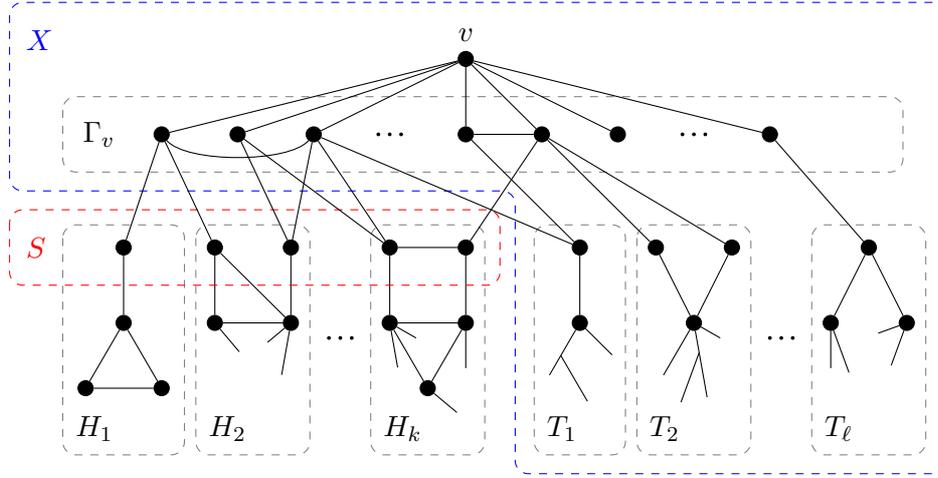
	
	An example of the standard decomposition is illustrated in Figure~\ref{fig:decomp}. Note that a (non-empty) connected graph $G$ is a near-forest if and only if there exists $v \in V(G)$ such that $S=\emptyset$ in the standard decomposition of $G$ from $v$.
	
	\begin{lemma}\label{lem:S-size}
		Let $G$ be a connected reduced graph, and let $v \in V(G)$. Then in the standard decomposition of $G$ from $v$, $|S| \ge 2$.
	\end{lemma}
	\begin{proof}
		Observe that in the standard decomposition of $G$ from $v$, $G[X]-v-\Gamma_G(v) = T_1 \cup \dots \cup T_\ell$ is a forest, so $G[X]$ itself is a near-forest. Thus by the definition of reducedness (Definition~\ref{def:reduced}), it follows that $|S| = |\Gamma_G(X)| \ge 2$.
	\end{proof}

	Our main branching operation will be the intuitive one --- we pick a vertex $v \in V(G)$ and partition sets $I \in \calI(G)$ according to whether or not $v \in I$, giving 
	\[
		\IS(\calG) = \wout(v)\,\IS(\calG-v) + \win(v)\,\wout(\Gamma_G(v))\,\IS(\calG-v-\Gamma_G(v)).
	\]
	Note that this equation holds even if $\calG-v-\Gamma_G(v)$ contains no vertices. After picking $v$, we then recurse on $\calG-v$ and $\calG-v-\Gamma_G(v)$. (Informally, we call this \emph{branching on $v$}, and we call $\calG-v$ and $\calG-v-\Gamma_G(v)$ the resulting \emph{branches}.) 

	Using the standard decomposition of $G$ from $v$, we will bound the efficiency of this branching operation in terms of $d_G(v)$, $d_G^2(v)$ and $|S|$ (see Lemma~\ref{lem:base-branch}). The bounds in Lemma~\ref{lem:base-branch} are weaker when $S$ is small. By Lemma~\ref{lem:S-size} we have $|S| \ge 2$; if $|S|=2$, say $S = \{x,y\}$, then we use an alternative strategy introduced in~\cite{DJW} and branch on~$y$ rather than $v$. In both branches, the neighbourhood of what remains of~$X$ is then either $\emptyset$ or $\{x\}$, so we can quickly remove it with $\Prune$. In~\cite{DJW}, the authors consider arbitrary graphs, and in this setting the strategy is effective enough that the $|S|=2$ case is no longer the limiting factor in the running time analysis. However, if $G$ is bipartite, then the bounds of Lemma~\ref{lem:base-branch} are stronger and this is no longer true.  What we do is construct an ``extended decomposition'' (see Definition~\ref{def:extend-decomp} below) which allows us
to branch on a different vertex~$z$, rather than on~$y$.
This gives better results in the bipartite case 
(see Appendix~\ref{app:ext-decomp-needed}).
Crucially, in the bipartite case, branching on~$z$ lets us improve what would 
otherwise be the least efficient branch when $|S|=2$, so it improves the results.	
 It will turn out that this is only an issue when $d_{G-X}(y) = 1$; in this case, we will extend $y$ into a path in $G-X$, and branch on the other endpoint $z$ of the path.

    Suppose, then, that $G$ is connected, and $|S|=2$ in the standard
    decomposition of $G$ from $v\in V(G)$.  Let $S=\{x,y\}$.  By
    Definition~\ref{def:decomp}, no component of $G-X$ in the standard
    decomposition is a tree so the component
    containing~$y$ is not an isolated vertex or path. If
    $d_{G-X}(y)=1$, $G-X$ must contain a path from~$y$ to some~$z$
    with $d_{G-X}(z)>2$.  This allows us to make the following
    definition.

	\begin{defn}\label{def:extend-decomp}
		Let $G$ be a connected graph, and let $v \in V(G)$. Suppose $|S| = 2$ in the standard decomposition of $G$ from~$v$. Then we define the \emph{extended decomposition} of $G$ from~$v$ to be the standard decomposition from~$v$, together with the tuple $(X^+,P,x,y,z,H)$, where:
		\begin{itemize}
			\item $\{x,y\}=S$, where $x<y$ (recall that $V(G) = \{1, \dots, n\}$);
			\item if $d_{G-X}(y) \ge 2$, then $z=y$; otherwise, $z$~is the (unique) closest vertex to~$y$ in $G-X$ with $d_{G-X}(z)>2$;
            \item $P$ is the (unique) $y$--$z$ path in $G-X$;
			\item $X^+ = X \cup V(P) \setminus \{z\}$;
			\item $H = G - X^+$.
		\end{itemize}
	\end{defn}

	\begin{figure}
\begin{center}
\begin{tikzpicture}[scale=1,node distance = 1.5cm]
    \tikzstyle{dot}=[fill=black, draw=black, circle, inner sep=0.1mm]
    \tikzstyle{vtx}=[fill=black, draw=black, circle, inner sep=2pt]
    \tikzstyle{vset}=[gray,rounded corners=5pt,dashed];

    %
    %
    \node[vtx] (v)  at (5,4.5) [label=90:{$v$}] {};
    \node[vtx] (g1) at (1,3.5) {};
    \node[vtx] (g2) at (2,3.5) {};
    \node[vtx] (g3) at (3,3.5) {};
    \node[vtx] (g4) at (5,3.5) {};
    \node[vtx] (g5) at (6,3.5) {};
    \node[vtx] (g6) at (7,3.5) {};
    \node[vtx] (g7) at (9,3.5) {};

    \node[dot] (dot1) at ($(g3)!0.5!(g4)$) {};
    \node[dot] at ($(dot1)+(-0.15,0)$)     {};
    \node[dot] at ($(dot1)+(0.15,0)$)      {};
    \node[dot] (dot2) at ($(g6)!0.5!(g7)$) {};
    \node[dot] at ($(dot2)+(-0.15,0)$)     {};
    \node[dot] at ($(dot2)+(0.15,0)$)      {};

    \draw (g1) .. controls ($(g1)+(310:0.5)$) and ($(g3)+(230:0.5)$) .. (g3);
    \draw (g4)--(g5);

    \draw[vset] (0.05,4) rectangle (10.75,3);
    \node at (0,3.5) [label={[gray]0:{$\Gamma_v$}}] {};

    \foreach \i in {1, ..., 7}
        \draw (v)--(g\i);

    %
    %
    \node[vtx] (x) at (1.5,1.65) [label=180:{$x$}] {};
    \node[vtx,green!50!black] (y) at (3.5,1.65) [label=180:{$y$}] {};
    \draw[vset] (0.05,2.15) rectangle (4,1.15);
    \node at (0,1.65) [label={[gray]0:{$S$}}] {};

    %
    %
    \node[vtx] (a1) at ($(x)+(-0.5,-1.5)$) {};
    \node[vtx] (a2) at ($(a1)+(240:1)$) {};
    \node[vtx] (a3) at ($(a1)+(300:1)$) {};
    \node[vtx] (a4) at ($(x)+(0.5,-1.5)$) {};
    \draw (x)--(a1)--(a2)--(a3)--(a1);
    \draw (x)--(a4)--(a3);

    \node[vtx,green!50!black] (z) at ($(y)+(0,-3)$) [label=270:{$z$}] {};
    \node[vtx,green!50!black] (p1) at ($(y)!0.33!(z)$) {};
    \node[vtx,green!50!black] (p2) at ($(y)!0.66!(z)$) {};
    \draw[green!50!black,thick] (y)--(p1)--(p2)--(z);
    \draw (a3) .. controls ($(a3)+(0.75,-1)$) and ($(z)+(-0.5,-0.25)$) .. (z);
    \draw (a4) .. controls ($(a4)+(0.25,-1)$) and ($(z)+(-0.8,0)$) .. (z);

    \node (pend) at ($(p2)!0.5!(z)$) {};     
    \node (pend1) at ($(pend)+(0, 0.1)$) {}; 
    \node (pend2) at ($(pend)+(0,-0.1)$) {}; 
    \node (bot) at ($(z)+(0,-0.5)$) {};      

    \draw[vset][color=red]
             (-1,1.5) -- (-1,2.35) -- (2.6,2.35) -- (2.6,2.35|-pend2)
                 -- (4.2,0|-pend2) -- (4.2,0|-bot) -- (-1,0|-bot) -- (-1,1.5);
    \node at (-1.1,1.85) [label={[red]0:{$H$}}] {};

    \node[green!50!black,fill=white] at ($(y)!0.5!(z) + (0.9,0)$) {$P$};
    \draw[green!50!black,thick,decorate,decoration={brace,amplitude=7.5}]
        ($(y)+(0.3,0.05)$) -- ($(z)+(0.3,-0.05)$);

    %
    %
    \node[vtx] (d1) at (6.5,2) {};
    \node[vtx] (d2) at (6.5,1) {};
    \draw (d1) -- (d2);
    \draw (d2) -- +(240:0.8);
    \draw ($(d2)+(240:0.5)$) -- +(300:0.7);
    \draw (d2) -- +(315:0.6);
    \draw[vset] (5.9,2.3) rectangle (7.1,-0.75);
    \node at (5.8,-0.4) [label={[gray]0:{$T_1$}}] {};

    %
    %
    \node[vtx] (e1) at (7.5,2) {};
    \node[vtx] (e2) at (8.5,2) {};
    \node[vtx] (e3) at (8,1) {};
    \draw (e1)--(e3)--(e2);
    \draw (e3) -- +(240:0.8);
    \draw (e3) -- +(330:0.4);
    \draw (e3) -- +(280:1);
    \draw ($(e3)+(280:0.4)$) -- +(250:0.7);
    \draw[vset] (7.25,2.3) rectangle (8.75,-0.75);
    \node at (7.15,-0.4) [label={[gray]0:{$T_2$}}] {};

    \node[dot] at (9.00,0.8) {};
    \node[dot] at (9.15,0.8) {};
    \node[dot] at (9.30,0.8) {};

    %
    %
    \node[vtx] (f1) at (10.3,2) {};
    \node[vtx] (f2) at  (9.8,1) {};
    \node[vtx] (f3) at (10.8,1) {};
    \draw (f2)--(f1)--(f3);
    \draw (f2) -- +(270:0.6);
    \draw (f2) -- +(290:0.7);
    \draw (f3) -- +(200:0.4);
    \draw (f3) -- +(250:0.6);
    \draw[vset] (9.55,2.3) rectangle (11.05,-0.75);
    \node at (9.45,-0.4) [label={[gray]0:{$T_\ell$}}] {};

    %
    %
    \draw[vset] (-0.15,3.5) -- (-0.15,5.25) -- (11.3,5.25) -- (11.3,-1)
                  -- (5.65,-1) -- (5.65,2.75) -- (-0.15,2.75) -- (-0.15,3.5);
    \node at (0,4.75) [label={[gray]0:{$X$}}] {};

    %
    %
    \draw[vset][draw=blue] (-1,3.5) -- (-1,5.45) -- (11.5,5.45) -- (11.5,-1.2)
                  -- (5.45,-1.2) -- (5.45,2.55) -- (4.2,2.55)
                  -- (4.2,2.55|-pend1) -- (2.8,2.55|-pend1) -- (2.8,2.55)
                  -- (-1,2.55) -- (-1,3.5);
    \node at (-1.15,4.95) [label={[blue]0:{$X^+$}}] {};

    %
    %
 
    \draw (g1)--(x)--(g3);
    \draw (x)--(g2)--(y);
    \draw (g5)--(y);
    \draw (g3)--(d1)--(g4);
    \draw (e1)--(g5)--(e2);
    \draw (g7)--(f1);

\end{tikzpicture}
\end{center}
\caption{An example of the extended decomposition of a connected graph~$G$ from a vertex~$v$.  Relevant elements of the graph's standard decomposition are shown in pale grey.}\label{fig:ext-decomp}
	\end{figure}
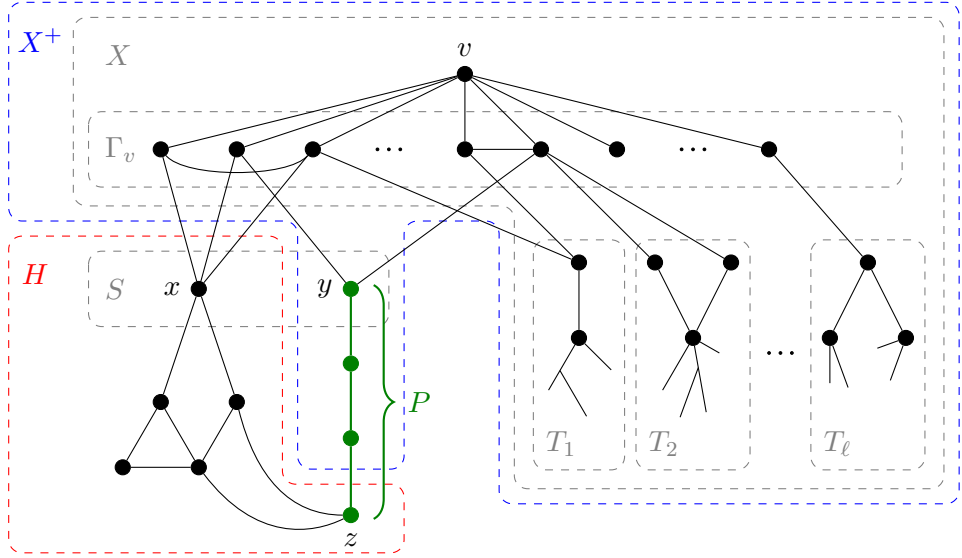
	
	When $|S|=2$ in the extended decomposition from $v$, we will branch on $z$ and then prune what remains of $X^+$ in both branches. This will be efficient enough that the worst case in the analysis will be $|S|=3$ even when $G$ is bipartite (see Lemma~\ref{lem:clever-branch}). An example of the extended decomposition is given in Figure~\ref{fig:ext-decomp} --- we justify the depiction of $P$ and $z$ in the following lemma.

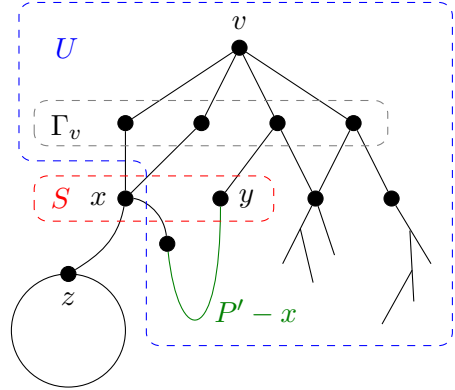
\begin{figure}
\begin{center}
\begin{tikzpicture}[scale=1,node distance = 1.5cm]
    \tikzstyle{vtx}=[fill=black, draw=black, circle, inner sep=2pt]
    \tikzstyle{svtx}=[fill=black, draw=black, circle, inner sep=1.5pt]
    \tikzstyle{vset}=[gray,rounded corners=5pt,dashed];

    %
    %
    \node[vtx] (v) at (2.75,1) [label=90:{$v$}] {};
    \node[vtx] (g1) at (1.25,0) {};
    \node[vtx] (g2) at (2.25,0) {};
    \node[vtx] (g3) at (3.25,0) {};
    \node[vtx] (g4) at (4.25,0) {};

    %
    %
    \node[vtx] (x) at (1.25,-1) [label=180:{$x$}] {};
    \node[vtx] (y) at (2.5,-1) [label=0:{$y$}] {};
    \node[vtx] (t1) at ($(g4)+(-0.5,-1)$) {};
    \node[vtx] (t2) at ($(g4)+(0.5,-1)$) {};
    \node[vtx] (z) at (0.5,-2) [label=270:{$z$}] {};
    \node[vtx] (x') at (1.8,-1.6) {};

    \draw (g1)--(v)--(g2);
    \draw (g3)--(v)--(g4);
    \draw (g1)--(x)--(g2);
    \draw (y)--(g3)--(t1)--(g4)--(t2);

    %
    %
    \draw[black!50!green] (y) .. controls (2.5,-3) and (1.95,-3) .. (x');
    \draw (x') .. controls (1.75,-1.15) and (1.45,-1).. (x);
    \node at (2.15,-2.5) [label={[black!50!green]0:{$P'-x$}}] {};

    \draw (x) .. controls (1.15,-1.55) and (0.9,-1.75) .. (z);
    \draw ($(z)+(0,-0.75)$) circle (0.75);

    %
    %
    \node (t3) at ($(t1)+(-0.5,-1)$) {};
    \node (t4) at ($(t3)+(0.5,-0.25)$) {};
    \draw (t1)--(t3);
    \draw ($(t1)!0.4!(t3)$)--(t4);
    \draw (t1)--($(t1)+(0.25,-0.75)$);

    \node (t5) at ($(t2)+(0.6,-1)$) {};
    \node (t6) at ($(t5)+(-0.3,-0.5)$) {};
    \node (t7) at ($(t6)+(-0.5,-0.3)$) {};
    \draw (t2)--(t5);

    \node[inner sep=0pt] (tm) at ($(t2)!0.4!(t5)$) {};
    \draw (tm)--(t6);
    \draw ($(tm)!0.5!(t6)$)--(t7);

    %
    %
    \draw[vset] ($(g1)+(-1.2,-0.3)$) rectangle ($(g4)+(0.45,0.3)$);
    \node at ($(g1)+(-1.25,-0.05)$) [label=0:{$\Gamma_v$}] {};

    \draw[vset][color=red] ($(x)+(-1.2,-0.3)$) rectangle ($(y)+(0.7,0.3)$);
    \node at ($(g1)+(-1.25,-1)$) [label={[red]0:{$S$}}] {};

    \draw[vset][color=blue] ($(g1)+(0,1.6)$) -- ($(g1)+(-1.4,1.6)$)
                  -- ($(g1)+(-1.4,-0.5)$)
                  -- ($(g1)+(0.275,-0.5)$) -- ($(g1)+(0.275,-2.95)$)
                  -- ($(g4)+(1.3,-2.95)$) -- ($(g4)+(1.3,1.6)$)
                  -- ($(g1)+(0,1.6)$);
    \node at ($(g1)+(-0.35,1)$) [label={[blue]180:{$U$}}] {};
\end{tikzpicture}
\end{center}
\caption{A supposed extended decomposition that leads to a contradiction in the proof of Lemma~\ref{lem:extend-decomp}.}
\label{fig:lem:extend-decomp}
\end{figure}

	\begin{lemma}\label{lem:extend-decomp}
		Let $G$ be a connected reduced graph, let $v \in V(G)$, and suppose that the standard decomposition of $G$ from~$v$ has $|S|=2$. Then in the extended decomposition of $G$ from $v$, we have $x \notin V(P)$ and $d_H(z) \ge 2$.
	\end{lemma}
\begin{proof}
    Suppose, towards a contradiction, that $x\in V(P)$ as in Figure~\ref{fig:lem:extend-decomp}. This implies that $d_{G-X}(y)=1$.  Let $P'$ be the sub-path of~$P$ from $y$ to~$x$ and let $U=(X\cup V(P'))\setminus\{x\}$.  $G[U]$ is a near-forest~-- delete $v$ and $\Gamma_G(v)$ to obtain the forest $(P'-x) \cup T_1\cup\dots\cup T_\ell$~-- but $\Gamma_G(U)=\{x\}$, contradicting the hypothesis that $G$~is reduced.

    Therefore, $z\neq x$. If $z=y$ then $X^+=X$, so $d_H(z)=d_{G-X}(y)\geq 2$.  Otherwise, $z\notin \{x,y\} = \Gamma_G(V(H))$, so $d_H(z)=d_{G-X}(z)-1\geq 2$.
\end{proof}

    \subsection{The main approximation algorithm}\label{sec:mainalgo}

	\begin{algorr}{$\iscount(\calG,\eps)$}{\label{algo:count}The input is a non-empty $1$-balanced weighted graph  
	$\calG = (G,\win,\wout,W)$ 	
	 with no tree components and a rational number $0 < \eps < 1$.}
		\item If $\Delta(G) \ge 11$:\label{rcount-1}
		\begin{enumerate}[(a)]
\item Let $v$ be the lexicographically least vertex of degree at least $11$ in~$G$ and let
\begin{align*}
\Gout= (\rGout,\cdot,\cdot,\newWout) &\leftarrow \TR\big(\calG-v\big),\\
 \Wout &\leftarrow \wout(v),\\
 \Cout & \leftarrow \mbox{If $\rGout$ is empty then $\newWout$ else $\iscount(\Gout,\eps)$, }\\
 \Gin =(\rGin,\cdot,\cdot,\newWin) &\leftarrow \TR\big(\calG-v-\Gamma_G(v)\big),\\
 \Win &\leftarrow \win(v)\cdot\wout(\Gamma_G(v)), \text{ and}\\
  \Cin & \leftarrow \mbox{If $\rGin$ is empty then $\newWin$ else $\iscount(\Gin,\eps)$. }
 \end{align*} 	 
			\item Return $\Wout\cdot  \Cout+ \Win\cdot  \Cin$.
		\end{enumerate}
		\item $\calG \leftarrow \Red(\calG)$.\label{rcount-2}
		\item If $G$ has no vertices of degree at least 6 and 2-degree at least 27, return $\basecount(\calG,\eps)$.\label{rcount-3}
		\item If $G$ has average degree at most 5, then let $v \in V(G)$ be the lexicographically least vertex of degree at least 6 and 2-degree at least 27. Otherwise, let $v \in V(G)$ be the lexicographically least vertex which maximises $d^2_G(v)$ subject to $d_G(v) \ge 2|E(G)|/|V(G)|$.\label{rcount-4}
		\item Let $G_v$ be the component of $G$ containing $v$, and consider the standard decomposition of $G_v$ from $v$. If $|S| \ge 3$, let\label{rcount-5}
\begin{align*}
\Gout= (\rGout,\cdot,\cdot,\newWout) &\leftarrow \TR\big(\calG-v\big),\\
 \Wout &\leftarrow \wout(v),\\
 \Gin =(\rGin,\cdot,\cdot,\newWin) &\leftarrow \TR\big(\calG-v-\Gamma_G(v)\big), \text{and}\\
 \Win &\leftarrow \win(v)\cdot\wout(\Gamma_G(v)).\\
  \end{align*} 		
 		\item If instead $|S| = 2$, consider the extended decomposition of $G_v$ from~$v$ and let\label{rcount-6}
\begin{align*}
\Gout= (\rGout,\cdot,\cdot,\newWout) &\leftarrow  \TR\big(\Prune(\calG-z,X^+)\big),\\
 \Wout &\leftarrow \wout(z),\\
 \Gin =(\rGin,\cdot,\cdot,\newWin) &\leftarrow   \TR\big(\Prune(\calG-z-\Gamma_G(z),X^+\setminus\Gamma_G(z))\big), \text{and}\\
 \Win &\leftarrow \win(z)\cdot\wout(\Gamma_G(z)).\\
  \end{align*} 			
		\item Let
	\begin{align*}
 \Cout & \leftarrow \mbox{If $\rGout$ is empty then $\newWout$ else $\iscount(\Gout,\eps)$ }\text{and}\\
  \Cin & \leftarrow \mbox{If $\rGin$ is empty then $\newWin$ else $\iscount(\Gin,\eps)$. }
 \end{align*}

		Return $\Wout\cdot  \Cout + \Win\cdot  \Cin$.\label{rcount-7}
	\end{algorr}

	\begin{lemma}\label{lem:iscount-works}
		Let $\calG=(G,\win,\wout,W)$ be a non-empty $1$-balanced $n$-vertex weighted graph with no tree components.  Let $0 < \eps < 1$.  Then,
        \begin{enumerate}[(i)]
        \item $\iscount(\calG,\eps)$ returns an $\eps$-approximation of $\IS(\calG)$;
        \item Within time $\poly(n,1/\epsilon)$, $\iscount(\calG,\eps)$ either terminates or computes the arguments $\Gin$ and~$\Gout$ to its recursive calls;
        \item At each recursive call $\iscount(\calH,\epsilon)$, $\calH$~is non-empty and $1$-balanced and has no tree components.
        \end{enumerate}
	\end{lemma}
	\begin{proof}
        We first prove part~(iii). The argument of a recursive call is either $\Gin$ or~$\Gout$.  These are computed from~$\calG$ by deleting vertices and calling some subset of $\Prune$, $\Red$ and~$\TR$. It is immediate from the definition that deleting vertices preserves the property of being $1$-balanced; $\Prune$, $\Red$ and $\TR$ maintain $1$-balance by Lemma~\ref{lem:MR}, Lemma~\ref{lem:poly-reduce} and Observation~\ref{obs:TR-Z}, respectively. The last operation is always a call to $\TR$, which removes any tree components by construction, so the result follows.

        We now prove (i) and~(ii).
		First, note that the input decreases in size on each recursive call, so \iscount{} must terminate.

        If $\Delta(G)\geq 11$, then steps \eqref{rcount-1}(a) and~\eqref{rcount-1}(b) are executed.  By partitioning sets $I \in \calI(G)$ according to whether or not $v \in I$, we have 
		\begin{align*}
			\IS(\calG) &= \wout(v)\,\IS(\calG-v) + \win(v)\,\wout(\Gamma_G(v))\,\IS(\calG-v-\Gamma_G(v))\\
			&= \Wout\IS(\calG-v) + \Win\IS(\calG-v-\Gamma_G(v))\,.
		\end{align*}
		By Observation~\ref{obs:TR-Z}, $\IS(\calG) = \Wout\IS(\Gout) + \Win\IS(\Gin)$. By part~(iii) of the current lemma, $\Gin$ and~$\Gout$ are $1$-balanced and have no tree components, so the recursive calls to \iscount{} are valid. Thus by the recursive correctness of \iscount, $\iscount(\calG,\eps) \in (1\pm\eps)\IS(\calG)$ as required.  $\Gin$ and~$\Gout$ are computed in time $\poly(n)$ by Observation~\ref{obs:TR-Z}.

        Otherwise, $\Delta(G)\leq 10$.  By Lemma~\ref{lem:poly-reduce}, step~\eqref{rcount-2} runs in time $2^{\Delta(G)}\,\poly(n)=\poly(n)$ and, after this step, $\calG$~is reduced and $1$-balanced and $Z(\calG)$ has not changed. Recall from Definition~\ref{defn:ourfamily} that $\ourfamily$ is the family of all graphs with minimum degree at least 2 which contain no vertices of degree at least 6 and 2-degree at least 27. Since $G$ is reduced, $\delta(G) \ge 2$. Thus if $G$~has no vertices of degree at least~$6$ and $2$-degree at least~$27$, then $G\in\ourfamily$ and step~\eqref{rcount-3} runs in time $\poly(n,1/\epsilon)$ and is correct by Theorem~\ref{thm:2-degree-FPTAS}.

        Step~\eqref{rcount-4} can clearly be computed in time $\poly(n)$.

        In step~\eqref{rcount-5}, $G_v$~is a component of a reduced graph, so it is reduced.  If $|S|\geq 3$, then correctness and polynomial running time follow exactly as for step~\eqref{rcount-1}.  Otherwise, by Lemma~\ref{lem:S-size}, $|S|\geq 2$, so we execute step~\eqref{rcount-6}.  As in step~\eqref{rcount-1}, we have
		\begin{align*}
			\IS(\calG) &= \Wout\IS(\calG-z) + \Win\IS(\calG-z-\Gamma_G(z))\,.
		\end{align*}
		Note that $\Gamma_{G-z}(X^+) = \{x\}$, and $\Gamma_{G-z-\Gamma_G(z)}(X^+ \setminus \Gamma_G(z)) \subseteq \{x\}$. Hence by Lemma~\ref{lem:MR} and Observation~\ref{obs:TR-Z}, it follows that $\IS(\calG) = \Wout\IS(\Gout) + \Win\IS(\Gin)$. Note that $\Gout$ and~$\Gin$ are $1$-balanced and have no tree components by part~(iii) of the current lemma, so the recursive calls to \iscount{} are valid. Thus, by the recursive correctness of \iscount{}, $\iscount(\calG,\eps) \in (1\pm\eps)\IS(\calG)$ as required.

        It remains to show that, in step~\eqref{rcount-6}, $\Gin$ and~$\Gout$ can be constructed in time $\poly(n)$.  Since $|S|=2$, $\Gamma_{G-z}(X^+)=\{x\}$, so $|\Gamma_{G-z}(X^+)| = 1$. Moreover, $G[X^+] - v - \Gamma_G(v) = T_1 \cup \dots \cup T_\ell \cup (P-z)$ is a forest, so $G[X^+]$ is a near-forest. Thus by Corollary~\ref{cor:prune-fast}, $\Prune(G-z,X^+)$ can be computed in time $2^{\Delta(G)}\,\poly(n)=\poly(n)$, so $\Gout$ can be computed in time $\poly(n)$ by Observation~\ref{obs:TR-Z}.  Similarly, $G[X^+\setminus \Gamma_G(z)]$ is a near-forest with $\Gamma_{G-z-\Gamma_G(z)}(X^+ \setminus \Gamma_G(z)) \subseteq \{x\}$, so $\Gin$ can also be computed in time $\poly(n)$.
    \end{proof}

	To bound the running time of \iscount, we first show that the value of any good slice must decrease by a substantial amount whenever \iscount\ makes recursive calls. We will accomplish this in Lemma~\ref{lem:branch-bounds}. (Recall from Definition~\ref{defn:good-slice} that $f\colon\R^2\to\R$ is a good slice if it is of the form $f(m,n) = \rho m+\sigma n$ for some $\rho \ge 0$ and $\rho \ge -\sigma$.) We will use the following lemma for the case where \iscount\ executes step (\ref{rcount-5}).
	
	\begin{lemma}\label{lem:base-branch}
		Let $\calG$ be a weighted graph whose underlying graph $G$ is connected and reduced. Let $v \in V(G)$, and suppose $f(m,n) = \rho m+\sigma n$ is a good slice. Let $\Gout=\TR(\calG-v)$ and $\Gin = \TR(\calG-v-\Gamma_G(v))$, and consider the standard decomposition of $G$ from $v$. Then 
		\begin{align*}
			f(\calG)-f(\Gout) &= \rho d_G(v)+\sigma,\\
            f(\calG)-f(\Gin) &\ge \rho \left\lceil\frac{d_G^2(v)+d_G(v)+|S|}{2}\right\rceil + \sigma\big(1+d_G(v)\big)\,.
        \end{align*}
        Furthermore, if $G$ is bipartite, then
        \[
			f(\calG)-f(\Gin) \ge \begin{cases}
				\,\rho d_G^2(v) + \sigma(1+d_G(v)) & \text{if }\sigma\ge 0\,,\\
				\,\rho d_G^2(v) + \sigma\left\lfloor \frac{d_G^2(v)+d_G(v)+2-|S|}{2} \right\rfloor & \text{if }\sigma<0\,.
			\end{cases}
		\]
	\end{lemma}
	\begin{proof}
		Let $\tGout = \TR(G-v)$ be the underlying graph of $\Gout$, and let $\tGin = \TR(G-v-\Gamma_G(v))$ be the underlying graph of $\Gin$. Since $G$ is reduced, whenever $Y \subseteq V(G)$ spans a tree we have $|\Gamma_G(Y)| \ge 2$; in particular, $G-v$ has no tree components, so $\tGout=\TR(G-v) = G-v$. It follows immediately that $f(G)-g(\tGout) = \rho d_G(v)+\sigma$ as required. 
		
		Observe that by the definition of the standard decomposition, we have
		\[
			\tGin = G - v - \Gamma_G(v) - V(T_1) - \dots - V(T_\ell)\,.
		\]
		Hence, writing $m_1$ for the number of edges in $G[\Gamma_G(v)]$, and writing $m_2$ for the number of edges between $\Gamma_G(v)$ and $\Gamma^2_G(v)$, we have
		\[
			f(G) - f(\tGin) = \rho\big(d_G(v) + m_1 + m_2\big) + \sigma\big(1+d_G(v)\big) + \sum_{i=1}^\ell f(T_i)\,.
		\]
		For all $i$, since $f$ is a good slice and $T_i$ is a tree, we have $f(T_i) = \rho(|V(T_i)|-1) + \sigma|V(T_i)| \ge \sigma$. Moreover, $d_G^2(v) = d_G(v) + 2m_1 + m_2$. It follows that
		\begin{equation}\label{eqn:base-branch-a}
			f(G)-f(\tGin) \ge \rho\big(d^2_G(v) - m_1\big) + \sigma\big(1+d_G(v)+\ell\big)\,.
		\end{equation}		
		$G$~is reduced, so $|\Gamma_G(V(T_i))| \ge 2$ for all $i \in [\ell]$, and each $T_i$ sends at least two edges to $\Gamma_G(v)$. Each vertex in $S$ also sends at least one edge to $\Gamma_G(v)$, so
		\begin{equation}\label{eqn:base-branch-b}
			m_2 \ge 2\ell + |S|\,.
		\end{equation}
		
		We now prove our first inequality. By~\eqref{eqn:base-branch-b},
		\[
			m_1 = \frac{d_G^2(v) - d_G(v) - m_2}{2} \le \frac{d_G^2(v) - d_G(v) - |S|}{2} - \ell\ \Rightarrow\ m_1 \le \Big\lfloor\frac{d_G^2(v) - d_G(v) - |S|}{2}\Big\rfloor - \ell\,.
		\]
		It follows from~\eqref{eqn:base-branch-a} that
		\[
			f(G)-f(\tGin) \ge \rho\bigg(\Big\lceil\frac{d^2_G(v) + d_G(v) + |S|}{2}\Big\rceil + \ell\bigg) + \sigma\big(1+d_G(v)+\ell\big)\,.
		\]
		Since $f$ is a good slice, the result now follows, as $\rho\ell + \sigma\ell \ge 0$.
		
		Finally, we prove our second inequality. Suppose that $G$ is bipartite, which implies that $m_1 = 0$. If $\sigma \ge 0$, then the result follows immediately from~\eqref{eqn:base-branch-a}, so suppose $\sigma < 0$. Then by~\eqref{eqn:base-branch-b}, we have
		\[
			\ell \le \frac{m_2-|S|}{2} = \frac{d^2_G(v)-d_G(v)-|S|}{2}\ \Rightarrow\ \ell \le \Big\lfloor \frac{d^2_G(v)-d_G(v)-|S|}{2} \Big\rfloor.
		\]
		It therefore follows from~\eqref{eqn:base-branch-a} that
		\begin{align*}
			f(G)-f(\tGin) \ge \rho d_G^2(v) + \sigma\Big\lfloor \frac{d_G^2(v) + d_G(v) + 2 - |S|}{2} \Big\rfloor,
		\end{align*}
		so the result follows.
	\end{proof}

	It is more difficult to analyse \iscount\ in the case where step (\ref{rcount-6}) is executed, so we require the following ancillary lemma before proving our bound (which appears as Lemma~\ref{lem:clever-branch}).

\begin{lemma}\label{lem:cycle-components}
    Let $H$ be a graph, let $z \in V(H)$ with $d_H(z) \ge 2$, and
    suppose $H-z$ has no tree components. Let $m_1$ be the number of
    edges in $H[\Gamma_H(z)]$, let $m_2$ be the number of edges in~$H$
    between $\Gamma_H(z)$ and $\Gamma_H^2(z)$, and let $r$ be the
    number of tree components of $H - z - \Gamma_H(z)$. Then
    $m_1 + m_2 \ge r + 2$.
\end{lemma}
\begin{proof}
    We may assume that $H$~is connected; if not, work within the
    component that contains~$z$, noting that none of the discarded
    components are trees.  Let $H^- = H-z-\Gamma_H(z)$.  Let the tree
    components of $H^-$ be $T_1, \dots, T_r$.  Each component of $H^-$
    sends at least one edge to $\Gamma_H(z)$ in~$H$, so $H^-$ has at
    most $m_2$ components and, in particular, $m_2\geq r$.  If
    $m_2\geq r+2$ then $m_1+m_2\geq r+2$, so we are done.  Two cases
    remain.

	\medskip\noindent\textbf{Case 1: $\boldsymbol{m_2 = r}$.}
    Suppose $m_2=r$.  If $m_1\geq 2$, then we are
    done. Otherwise, $H[\Gamma_H(z)]$ is a forest.  Each $T_i$
    sends exactly one edge to $\Gamma_H(z)$ in~$H$, so $H-z$ is a
    forest, which is a contradiction.

   	\medskip\noindent\textbf{Case 2: $\boldsymbol{m_2 = r+1}$.}
   	Suppose $m_2=r+1$.  If $m_1\geq 1$, we are done, so suppose
    that $m_1=0$.

    First, suppose that the only components of~$H^-$ are the trees
    $T_1, \dots, T_r$.  Some $T_i$ receives exactly two edges,
    $e_1$ and~$e_2$, from $\Gamma_H(z)$ in~$H$; the others receive
    exactly one edge each.  The graph $H-z-e_1$ is a forest: it
    consists of $\Gamma_H(z)$, which is an independent set because
    $m_1=0$, plus the trees $T_1, \dots, T_r$, plus one edge from
    each $T_j$ to $\Gamma_H(z)$.  It contains at least two
    components, because $d_H(z)\geq 2$.  Adding~$e_1$ to this
    graph gives $H-z$, which must contain at least one tree
    component, contradicting the hypothesis.

    Otherwise, $H^-$ has exactly $r+1$ components: the trees
    $T_1, \dots, T_r$ and one non-tree~$C$.  Each of these
    receives exactly one edge from $\Gamma_H(z)$ in~$H$.
    Since $d_H(z)\geq 2$, there is some 
    $y\in\Gamma_H(z)$ that is not adjacent to~$C$
    in~$H$.  The component of $H-z$ that contains~$y$ is a tree,
    contradicting the hypothesis. \qedhere
\end{proof}
	
	\begin{lemma}\label{lem:clever-branch}
		Let $\calG$ be a weighted graph whose underlying graph $G$ is connected and reduced. Let $v \in V(G)$, and suppose $f(m,n) = \rho m+\sigma n$ is a good slice. Consider the standard decomposition of $G$ from $v$, and suppose $|S|=2$. Consider the extended decomposition of $G$ from $v$, and let
		\begin{align*}
			\Gout &= \TR(\Prune(G-z,X^+))\,,\\
			\Gin &= \TR(\Prune(G-z-\Gamma_G(z),X^+ \setminus \Gamma_G(z)))\,.
		\end{align*}
		Then
		\begin{align*}
			f(\calG)-f(\Gout) &\ge \rho (1+d_G(v))+\sigma\,,\\
            f(\calG)-f(\Gin)  &\ge \rho \left\lceil\frac{d_G^2(v)+d_G(v)+4}{2}\right\rceil + \sigma\big(1+d_G(v)\big)\,.
        \end{align*}
        Furthermore, if $G$ is bipartite, then
        \[
			f(\calG)-f(\Gin) \ge \begin{cases}
				\,\rho (1+d_G^2(v)) + \sigma(1+d_G(v)) & \mbox{if $\sigma\ge 0$}\,,\\
				\,\rho d_G^2(v) + \sigma\left\lfloor \frac{d_G^2(v)+d_G(v)-2}{2} \right\rfloor & \mbox{if $\sigma<0$}\,.
			\end{cases}
		\]
	\end{lemma}
	\begin{proof}
		Let $\tGout = \TR(G-z-X^+)$ be the underlying graph of $\Gout$, and let $\tGin = \TR(G-z-\Gamma_G(z)-X^+)$ be the underlying graph of $\Gin$. We first show that $H-z$ has no tree components. Indeed, suppose for a contradiction that $C$ is a tree component of $H-z$. Since $G$ is reduced, we have $|\Gamma_G(V(C))| \ge 2$. Since $\Gamma_H(V(C)) \subseteq \{z\}$, it follows that $C$ sends an edge into $X^+$, so $x \in V(C)$. Then $C^+ := G[X^+ \cup V(C)]$ is a component of $G-z$, and $C^+ - v - \Gamma_G(v) = T_1 \cup \dots \cup T_\ell \cup (P-z) \cup C$ is a forest, so $C^+$ is a near-forest. Moreover, $\Gamma_G(V(C^+)) = \{z\}$; this contradicts the fact that $G$ is reduced. Thus $H-z$ has no tree components, so
		\[
			\tGout = \TR(H-z) = H-z\,.
		\]
		By Lemma~\ref{lem:extend-decomp}, we have $d_H(z) \ge 2$. Since $f$ is a good slice, it follows that
		\begin{equation}\label{eqn:clever-gout-1}
			f(\tGout) = f(H) - \rho d_H(z) - \sigma \le f(H) - \rho\,.
		\end{equation}
		
		Writing $m$ for the number of edges between $X^+$ and $V(H)$, we have
		\begin{align}\nonumber
			f(G)-f(H) &= \big(f(G)-f(G-v)\big) + \big(f(G-v)-f(H)\big)\\\label{eqn:clever-gout-2}
			&= \big(\rho d_G(v) + \sigma\big) + \big(f(G[X^+]-v) + \rho m\big)\,.
		\end{align}
		Let $D_1, \dots, D_R$ be the tree components of $G[X^+]-v$, and let $D_1', \dots, D_s'$ be the non-tree components. Since $f$ is a good slice, we have
		\begin{equation}\label{eqn:clever-gout-3}
			f(G[X^+]-v) = \sum_{i=1}^R f(D_i) + \sum_{i=1}^s f(D_i') \ge \sigma R\,.
		\end{equation}
		Since $G$ is reduced, $G-v$ has no tree components, so each of $D_1, \dots, D_R$ must be joined to $H$ by at least one edge. Thus $m \ge R$, so~\eqref{eqn:clever-gout-2} and~\eqref{eqn:clever-gout-3} imply
		\[
			f(G) - f(H) \ge \rho d_G(v) + \sigma + \sigma R + \rho m \ge \rho d_G(v) + \sigma\,.
		\]
		The claimed bound on $f(\Gout)$ therefore follows from~\eqref{eqn:clever-gout-1}.

		Now let $r$ be the number of tree components of $H-z-\Gamma_H(z)$, let $m_1$ be the number of edges internal to $\Gamma_H(z)$ in $H$, and let $m_2$ be the number of edges between $\Gamma_H(z)$ and $\Gamma^2_H(z)$ in $H$. Since $f$ is a good slice, we have
		\begin{align*}
			f(\tGin) &= f\big(\TR(H-z-\Gamma_H(z))\big) \le f\big(H-z-\Gamma_H(z)\big) - \sigma r\\
			&= f(H) - \rho\big(d_H(z)+m_1+m_2\big) - \sigma(1+d_H(z)+r)\,.
		\end{align*}
		Recall that we have already shown that $H-z$ has no tree components and $d_H(z) \ge 2$; thus by Lemma~\ref{lem:cycle-components} applied to $H$ and $z$, we have $m_1 + m_2 \ge r + 2$. Since $f$ is a good slice, it follows that
		\begin{equation}\label{eqn:clever-gin-1}
			f(\tGin) \le f(H) - \rho\,.
		\end{equation}
		
		Observe that by Lemma~\ref{lem:extend-decomp}, $H$ is formed from $G-X$ by removing a (possibly empty) induced path $P-z$. Since $f$ is a good slice, it follows that $f(H) \le f(G-X) = f(\TR(G-v-\Gamma_G(v)))$. Since $|S| = 2$, it follows by Lemma~\ref{lem:base-branch} that
        \[
            f(G)- f(H) \ge \rho \left\lceil\frac{d^2_G(v)+d_G(v)+2}{2}\right\rceil + \sigma\big(1+d_G(v)\big)
        \]
        in all cases and, if $G$~is bipartite, then
		\[
			f(G)- f(H) \ge \begin{cases}
				\,\rho d^2_G(v) + \sigma(1+d_G(v)) & \mbox{if $\sigma\ge 0$}\,,\\
				\,\rho d^2_G(v) + \sigma\left\lfloor \frac{d^2_G(v)+d_G(v)}{2} \right\rfloor & \mbox{if $\sigma<0$}\,.
			\end{cases}
		\]
		Thus by~\eqref{eqn:clever-gin-1} and the fact that $f$ is a good slice, we obtain the required bounds on $f(\Gin)$.
	\end{proof}
	
	To analyse the running time of \iscount, we will incorporate good slices into a piecewise-linear potential function of the following form.

	\begin{defn}\label{def:pre-pot}
		Let $s$ be a positive integer, let $-1 = k_0 < k_1 < \dots < k_s = \infty$, and let $f_1,\dots,f_s\colon\R^2\to\R$ be linear functions. Then the \emph{pre-potential} with \emph{boundary points} $k_0,\dots,k_s$ and \emph{slices} $f_1,\dots,f_s$ is the function 
$f\colon \R^2_{\geq 0}\to\R$ 
defined as follows.
If $n=0$, then $f(m,n)=0$.
Otherwise,
$f(m,n) = f_i(m,n)$ whenever $k_{i-1} < 2m/n \le k_i$. We say $f$ is a \emph{valid} pre-potential if it satisfies the following properties.
		\begin{enumerate}[(VP1)]
			\item $f_1,\dots,f_s$ are good slices. That is, $f_i(m,n) = \rho_i m + \sigma_i n$ where
			$\rho_i \geq 0$, $\sigma_i+\rho_i \geq 0$, and it is not the case that $\sigma_i=\rho_i=0$.
			\label{VP-good}
			\item $f_s(m,n) = \sigma_sn$ for some $\sigma_s > 0$.\label{VP-top}
			\item Every integer in $[6,k_{s-1}]$ is a boundary point.\label{VP-bdry}
			\item For all 
$n>0$ and $m\geq 0$,
$f(m,n) = \min_{j \in [s]}f_j(m,n)$.\label{VP-min}
		\end{enumerate}
		The \emph{potential} corresponding to $f$ is the map $f^+$ on the class of all graphs given by
		\[
			f^+(G) = \begin{cases}
				f_s(|E(G)|,|V(G)|)&\mbox{ if }\Delta(G) \ge 11,\\
				f(|E(G)|,|V(G)|)&\mbox{ otherwise.}
			\end{cases}
		\]
		If $\calG=(G,\win,\wout,W)$ is a weighted graph, we write $f^+(\calG) = f^+(G)$. We say $f^+$ is a \emph{valid potential} if it is constructed this way from a valid pre-potential $f$.
	\end{defn}

	We now set out bounds on the behaviour of a valid potential function as the algorithm \iscount\ from Section~\ref{sec:mainalgo} branches. Suppose $G$ is a reduced graph with average degree in $(k_i,k_{i+1}]$ for some $i \in \{0,\dots,s-1\}$, and suppose 
that $\Delta(G) < 11$ and 
	$G$ contains at least one vertex of degree at least 6 and 2-degree at least 27 so that $v$ is defined in step (iv) of \iscount. Then we have two cases:
	\begin{itemize}
		\item If $G$ has average degree at most $5$, then we will have $d^2_G(v) \ge 27$. 
		\item Otherwise, $v$ is a vertex which maximises $d_G^2(v)$ subject to $d_G(v) \ge 2|E(G)|/|V(G)| \in (\max\{k_i,5\},k_{i+1}]$. Since $G$ is reduced it has minimum degree at least 2, so $d^2_G(v) \ge 2d_G(v) \ge 2k_i$. Moreover, by (VP\ref{VP-bdry}), we have $k_{i+1} \le \floor{\max\{k_i,5\}}+1$; thus by Lemma~\ref{lem:D2} applied with $k=\max\{k_i,5\}$, it follows that $G$ contains a vertex with degree at least $2|E(G)|/|V(G)|$ and $2$-degree at least $D_2(k_i)$, and so $d^2_G(v) \ge D_2(k_i)$.
	\end{itemize} 
	We therefore define the following quantity.

\begin{defn}\label{def:Dtwoprime}
For all $k < 5$, let $D_2'(k) = 27$.
For all $k\geq5$, let $D_2'(k) =  \max\{2k,D_2(k)\}$.
\end{defn}	

	\begin{obs}\label{obs:v-bound}
		Suppose that $k_0, \dots, k_s$ are the boundary points of a valid pre-potential~$f$. In the execution of $\iscount(\calG,\eps)$, if $\calG$ has average degree in $(k_i,k_{i+1}]$ for some $i\in\{0, \dots, s-1\}$ and step~(iv) is executed, then $d^2_G(v) \ge D_2'(k_i)$.
	\end{obs}

	\begin{lemma}\label{lem:branch-bounds}
Let $f$ be a valid pre-potential with boundary points $k_0, \dots, k_s$ and slices $f_1,\dots,f_s$
where, for each $i\in \{1,\ldots,s\}$, $f_i = \rho_i m + \sigma_i n$. Let $0<\eps<1$, let $G_0$ be a graph, and let $\calG_0 = \TR(G_0,\mathbf{1},\mathbf{1},1)$. 
Suppose that $\calG_0$ is non-empty. Then		
		when $\iscount(\calG_0,\eps)$ is executed, the following properties hold for every recursive call $\iscount(\calG,\eps)$, where $\calG=(G,\win,\wout,W)$ and $n=|V(G)|$.
		\begin{enumerate}[(i)]
			\item $f^+(\calG) \ge 0$, with equality only if $\iscount(\calG,\eps)$ halts in time $\poly(n,1/\epsilon)$ with no further recursive calls.  
						\item If $\iscount(\calG,\eps)$ makes recursive calls to $\iscount(\Gout,\eps)$ and $\iscount(\Gin,\eps)$, then $f^+(\Gout) < f^+(\calG)$ and $f^+(\Gin) < f^+(\calG)$. Moreover, the following stronger bounds hold. If $\Delta(G) \ge 11$, then 
			\begin{align*}
				f^+(\calG) - f^+(\Gout) &\ge f_s(0,1),\\
				f^+(\calG) - f^+(\Gin) &\ge f_s(0,12).
			\end{align*}
			Otherwise, there exist $j \in [s]$ and a positive integer $x$ satisfying $\max\{6,\floor{k_{j-1}}+1\} \le x \le 10$ such that
			\begin{align*}
				f^+(\calG) - f^+(\Gout) &\ge f_j(x,\ 1),\\
                f^+(\calG) - f^+(\Gin) &\ge f_j\Big(\Big\lceil\frac{x+D_2'(k_{j-1})+3}{2}\Big\rceil,\ 1+x\Big)\,,
            \end{align*}
            and, if $G_0$ is bipartite, then
            \[
				f^+(\calG) - f^+(\Gin) \ge \begin{cases}
					f_j\big(D_2'(k_{j-1}),\ 1+x\big) & \text{if $\sigma_j \ge 0$},\\
					f_j\big(D_2'(k_{j-1}),\ \big\lfloor\frac{x+D_2'(k_{j-1})-1}{2}\big\rfloor\big) & \text{if $\sigma_j < 0$}.
				\end{cases}
            \]
		\end{enumerate} 
	\end{lemma}
	
	\begin{proof}
		\textbf{Property (i):} Since $G$ has no tree components and is non-empty , we have $|E(G)| \ge |V(G)|>0$; thus by (VP\ref{VP-good}) we have $f_i(G) \ge 0$ for all $i$, so $f^+(\calG) \ge 0$. Suppose $f^+(\calG) = 0$. By (VP\ref{VP-top}), it follows that $\Delta(G) \le 10$, so step~(\ref{rcount-2}) is executed. 
If $\Red(\calG)$ is empty then no recursive call is made, so 	the algorithm therefore halts in time $\poly(n,1/\epsilon)$.	
Assume that $\Red(\calG)$ is non-empty.		
		Let $i$ be such that $f^+(\calG) = f_i(\calG)$; then by (VP\ref{VP-min}) and Lemma~\ref{lem:poly-reduce}, we have
		\begin{equation}\label{eqn:bb-red}
			f^+(\Red(\calG)) \le f_i(\Red(\calG)) \le f_i(\calG) = f^+(\calG) = 0.
		\end{equation}
		Moreover, $\Red(\calG)$ is reduced by Lemma~\ref{lem:poly-reduce}, so it has no tree components; it follows that $f^+(\Red(\calG)) = 0$. By (VP1), it follows that $\Red(\calG)$ has the same number of edges as vertices. Again since $\Red(\calG)$ is reduced, it has no vertices of degree less than 2, so it must be 2-regular. But any such graph is a disjoint union of cycles, and cycles are near-forests, 
contradicting the fact that $\Red(\calG)$ is  reduced.	 
		
		\medskip\noindent\textbf{Property (ii):} Note that $D_2'(x) \ge 2x$ for all $x>0$, so by (VP\ref{VP-good}) and (VP\ref{VP-top}) the bounds $f^+(\Gin) < f^+(\calG)$ and $f^+(\Gout) < f^+(\calG)$ are indeed implied by the subsequently stated bounds. We split into cases depending on where $\Gin$ and $\Gout$ are defined.
		
		\medskip\noindent\textit{Case 1: Step~(\ref{rcount-1}) is executed.} In this case, $\Delta(G) \ge 11$. 
Recursive calls are made only if $\Gout$ and $\Gin$, respectively, are non-empty.		
		By (VP\ref{VP-min}),  and (VP\ref{VP-top}),
		 		\begin{align*}
			f^+(\Gout) &\le f_s(\Gout) \le f_s(G-v) = f_s(G) - \sigma_s\,,\\
			f^+(\Gin) &\le f_s(\Gin) \le f_s(G-v-\Gamma_G(v)) = f_s(G) - (1+d_G(v))\sigma_s \le f_s(G) - 12\sigma_s\,,
		\end{align*}
		as required.
		
		\medskip\noindent\textit{Case 2: Step~(\ref{rcount-5}) is executed.}\label{def:Gdash} Write $\calG' = \Red(\calG)$. 
Assume that $\calG'$ is non-empty (otherwise, there are no recursive calls).		
		$\calG_0$~is $1$-balanced by definition, so $\calG$ is $1$-balanced and has no tree components by Lemma~\ref{lem:iscount-works}(iii). Moreover, $G$ has maximum degree at most $10$ since step (v) is executed. Thus by Lemma~\ref{lem:poly-reduce}, $\Delta(\calG') \le \Delta(\calG) \leq 10$. Since $\calG$ has no tree components, as in~\eqref{eqn:bb-red} we have $f^+(\calG') \le f^+(\calG)$. 
		
		Let $j \in [s]$ be such that $k_{j-1} < 2|E(\calG')|/|V(\calG')|\le k_j$, so that $f^+(\calG') = f_j(\calG')$. Write $G_v$ for the component of $\calG'$ containing~$v$, as in $\iscount(\calG,\eps)$. Then by (VP\ref{VP-min}), we have
		\[
			f^+(\calG) - f^+(\Gin) \ge f^+(\calG') - f_j(\Gin) = f_j(\calG') - f_j(\TR(\calG'-v-\Gamma_{G_v}(v)))\,.
		\]
		By Lemma~\ref{lem:poly-reduce}, $\calG'$~is reduced and therefore has no tree components. Thus by (VP\ref{VP-good}),
		\[
			f^+(\calG) - f^+(\Gin) \ge f_j(G_v) - f_j\big(\TR(G_v-v-\Gamma_G(v))\big)\,.
		\]
		Again by (VP\ref{VP-good}), it follows by Lemma~\ref{lem:base-branch} (applied with $G = G_v$ and $f=f_j$) that
        \begin{equation}\label{eq:diffgen}
            f^+(\calG)-f^+(\Gin) \ge \rho_j \bigg\lceil\frac{\strut d_{G_v}^2(v)+d_{G_v}(v)+|S|}{2}\bigg\rceil + \sigma_j\big(1+d_{G_v}(v)\big)\,,
        \end{equation}
        and, if $G_v$~is bipartite, then
		\begin{equation}\label{eq:diffbip}
			f^+(\calG)-f^+(\Gin)\ge \begin{cases}
				\rho_j d_{G_v}^2(v) + \sigma_j(1+d_{G_v}(v)) & \text{if $\sigma_j\ge 0$},\\
				\rho_j d_{G_v}^2(v) + \sigma_j\Big\lfloor \frac{\strut d_{G_v}^2(v)+d_{G_v}(v)+2-|S|}{2} \Big\rfloor & \text{if $\sigma_j<0$}.\\
			\end{cases}
		\end{equation}
		
		$\calG'$ has average degree in $(k_{j-1},k_j]$ by the definition of~$j$, so $d_{G_v}^2(v) \ge D_2'(k_{j-1})$ by Observation~\ref{obs:v-bound}. Moreover, since $f_j$~is a good slice, the right-hand sides of \eqref{eq:diffgen} and~\eqref{eq:diffbip} are increasing functions of $d_{G_v}^2(v)$. We may therefore replace $d_{G_v}^2(v)$ by $D_2'(k_{j-1})$. Moreover, since step~(\ref{rcount-5}) is executed, we must have $|S| \ge 3$ in the standard decomposition of $G_v$ from~$v$.  
Thus~\eqref{eq:diffgen} and~\eqref{eq:diffbip} give
		\begin{equation}\label{eqn:branch-bounds-gen}
            f^+(\calG)-f^+(\Gin)\ge \rho_j\Big\lceil\frac{\strut D_2'(k_{j-1})+d_{G_v}(v)+3}{2}\Big\rceil + \sigma_j\big(1+d_{G_v}(v)\big)\,.
        \end{equation}
        If $G_0$ is bipartite, then so is~$G_v$.  Therefore, when $G_0$ is bipartite,
		\begin{equation}\label{eqn:branch-bounds-bip}
			f^+(\calG)-f^+(\Gin)\ge \begin{cases}
				\rho_j D_2'(k_{j-1}) + \sigma_j(1+d_{G_v}(v)) & \text{if $\sigma_j\ge 0$},\\
				\rho_j D_2'(k_{j-1}) + \sigma_j \Big\lfloor \frac{\strut D_2'(k_{j-1})+d_{G_v}(v)-1}{2} \Big\rfloor & \text{if $\sigma_j<0$}.
			\end{cases}
		\end{equation}
		By our choice of $v$ in step~(\ref{rcount-4}) of \iscount, we have $d_{G_v}(v) \ge \max\{6,\floor{k_{j-1}}+1\}$, and since step~(\ref{rcount-1}) was not executed, we have $d_{G_v}(v) \le 10$. Thus the lemma is satisfied on taking $x = d_{G_v}(v)$.
		
Similarly, by (VP\ref{VP-good}), (VP\ref{VP-min}) and Lemma~\ref{lem:base-branch} we have 
\[f^+(\calG) - f^+(\Gout) \ge 
f_j(G_v) -
f_j(\TR(G_v-v)) \ge \rho_jd_{G_v}(v) + \sigma_j.\]
As above, the lemma follows on taking $x =d_{G_v}(v)$.
		
\medskip\noindent\textit{Case 3: Step~(\ref{rcount-6}) is executed.} Exactly as in Case 2, write $\calG' = \Red(\calG)$ and take $j \in [s]$ such that $f^+(\calG) = f_j(\calG)$. By  (VP\ref{VP-min}),(VP\ref{VP-good}),  and Lemma~\ref{lem:clever-branch} applied to $G_v$,
        \[
            f^+(\calG) - f^+(\Gin) \ge \rho_j \bigg\lceil\frac{\strut d_{G_v}^2(v)+d_{G_v}(v)+4}{2}\bigg\rceil + \sigma_j\big(1+d_{G_v}(v)\big)\,,
        \]
        and, for bipartite~$G_v$,
		\[
			f^+(\calG) - f^+(\Gin) \ge \begin{cases}
				\rho_j (1+d_{G_v}^2(v)) + \sigma_j(1+d_{G_v}(v)) & \text{if $\sigma_j\ge 0$},\\
				\rho_j d_{G_v}^2(v) + \sigma_j\Big\lfloor \frac{\strut d_{G_v}^2(v)+d_{G_v}(v)-2}{2} \Big\rfloor & \text{if $\sigma_j<0$}.
			\end{cases}
		\]
		As before, the right-hand sides are increasing functions of  $d_{G_v}^2(v)$, since $f_j$ is a good slice.  Observation~\ref{obs:v-bound}, as in Case~2 it follows that
        \[
			f^+(\calG)-f^+(\Gin) \ge \rho_j \bigg\lceil\frac{\strut D_2'(k_{j-1})+d_{G_v}(v)+4}{2}\bigg\rceil + \sigma_j(1+d_{G_v}(v))
        \]
        and, for bipartite~$G_0$,
		\[
			f^+(\calG)-f^+(\Gin)\ge \begin{cases}
				\rho_j(1+D_2'(k_{j-1})) + \sigma_j(1+d_{G_v}(v)) & \text{if $\sigma_j\ge 0$},\\
				\rho_jD_2'(k_{j-1}) + \sigma_j \Big\lfloor \frac{\strut D_2'(k_{j-1})+d_{G_v}(v)-2}{2} \Big\rfloor & \mbox{if $\sigma_j<0$}.
			\end{cases}
		\]
		Note that these bounds are stronger than \eqref{eqn:branch-bounds-gen} and~\eqref{eqn:branch-bounds-bip}, respectively, so the result follows as in Case~2. Similarly, by (VP\ref{VP-good}), Lemma~\ref{lem:clever-branch} and Observation~\ref{obs:v-bound} we have $f^+(\calG) - f^+(\Gout) > \rho_jd_{G_v}(v) + \sigma_j$. As in Case~2, these bounds are of the required form on taking $x = d_{G_v}(v)$.
	\end{proof}

	\subsection{Analysis of recursion}\label{sec:recursion}

	We will use the theory of branching factors to bound the running time of $\iscount$; see e.g.~\cite[Chapter~2.1]{FKbook} for an overview. We briefly recall the salient points. For all integers $b \ge 2$ and all $a_1, \dots, a_b \ge 0$, we write $\tau(a_1, \dots, a_b)$ for the unique positive solution in $x$ of $\sum_{i=1}^b x^{-a_i} = 1$. 
Observe that $a_1,\ldots, a_b > 0$ and $\tau(a_1,\ldots, a_b) < c$ precisely when $c^{-a_1} + \cdots + c^{-a_b} \leq 1$.

	Suppose we have a non-negative potential function for instances of a problem, and wish to use it to bound a recursive algorithm's running time. Suppose the non-recursive part of the algorithm runs in polynomial time. Let $b, p_1,\dots,p_b \ge 1$ be such that the algorithm always recurses in one of $b$ possible ways, and that the $i$'th possible branch splits the instance into $p_i$ parts. Suppose that if the original instance has potential $x$, then the $j$'th part of the $i$'th possible branch always has potential at most $x-y_{i,j}$ for some $y_{i,j} > 0$. Finally, suppose that our potential is only zero on instances which the algorithm solves without recursing further. Under these circumstances, the following recurrence holds for the worst-case running time $T(x)$ of a potential-$x$ instance:
	\begin{equation}\label{eqn:gen-pot-rec}
	T(x) \le \begin{cases}
	\max \{\sum_{j=1}^{p_i} T(x-y_{i,j}) \colon i \in [b]\} + \poly(n) & \mbox{ if }x>0,\\
	\poly(n) & \mbox{ otherwise.}
	\end{cases}
	\end{equation}
	The solution to this recurrence is
	\begin{equation}\label{eqn:bf-use}
        \textstyle  
		T(x) = O^*(\max_i \{\tau(y_{i,1}, \dots, y_{i,p_i})\}^x).
	\end{equation}

    We now use~\eqref{eqn:bf-use} and Lemma~\ref{lem:branch-bounds} to obtain expressions for the running time of $\iscount$ in terms of valid pre-potential functions. We will then set out our choices of these functions, and prove their validity, in Section~\ref{sec:pre-pot}.

\begin{cor}\label{cor:fullalgo-runtime}
Let $G$ be an 
$n$-vertex graph which is not a forest. 
Let $f$ be a valid pre-potential with boundary points $k_0,\dots,k_s$ and good slices $f_1,\dots,f_s$, writing $f_i(m,n) = \rho_im + \sigma_in$. For all $i \in [s]$ and all $d \ge 6$, let
		\[
			T_{i,d} = \tau\bigg(f_i(d,1),\ f_i\Big(\Big\lceil\frac{d + D_2'(k_{i-1}) + 3}{2} \Big\rceil,\ 1+d\Big)\bigg),
		\]
		and 
		\[
			T = \max\Big(\big\{T_{i,d} \mid i \in [s], \max\{6,\floor{k_{i-1}}+1\} \le d \le 10\big\} \cup \big\{\tau(\sigma_s,12\sigma_s)\big\}\Big).
		\]
		If $T \le 2$, then $\iscount(\TR(G,\mathbf{1},\mathbf{1},1),\eps)$ has running time $O(2^{\sigma_s n})\,\poly(n,1/\epsilon)$.
	\end{cor}
	\begin{proof}
		Let $\calG_0 = \TR(G,\mathbf{1}, \mathbf{1}, \eps)$. 
Since $G$ is not a forest, $\calG_0$ is non-empty.	
		For all $x \ge 0$, let $R_\eps(x)$ be the worst-case running time of $\iscount(\calG,\eps)$ on any weighted graph $\calG$ with $f^+(\calG) \le x$ which  
arises 	
		as a recursive call in the evaluation of $\iscount(\calG_0,\eps)$. 
For each $i \in [s]$, let $z_i = \max\{6,\floor{k_{i-1}}+1\}$. 		
If recursive calls are made to $\iscount(\Gout,\epsilon)$
and  $\iscount(\Gin,\epsilon)$ then,
by Lemma~\ref{lem:branch-bounds}(ii), either
		\begin{align*}
			f^+(\Gout) &\le f^+(\calG) - f_s(0,1)< f^+(\calG)\,,\\
			f^+(\Gin) &\le f^+(\calG) - f_s(0,12)< f^+(\calG)\,,
		\end{align*}
		or there exist $i \in [s]$ and an integer $d$ with $z_i \le d \le 10$ such that 
		\begin{align*}
			f^+(\Gout) &\le f^+(\calG) - f_i(d,\ 1)< f^+(\calG)\,,\\
			f^+(\Gin) &\le f^+(\calG) - f_i\Big(\Big\lceil\frac{d+D_2'(k_{i-1})+3}{2}\Big\rceil,\ 1+d\Big)< f^+(\calG)\,.
		\end{align*}
(If, e.g., only 
$\iscount(\Gout,\epsilon)$ is called,  then the bounds on
$f^+(\Gout) $  still hold.)

		It follows by the above equations and Lemma~\ref{lem:branch-bounds}(i) that the conditions of~\eqref{eqn:gen-pot-rec} are satisfied by taking our potential to be $f^+$, the $y_{i,j}$'s to be the right-hand sides of the above equations, and the $p_i$'s to be 
the number of recursive calls in the $i$'th possible branch (which is at most~$2$).		
	 For brevity, let $t_{i,d,\texttt{out}} = f_i(d,1)$, and $t_{i,d,\texttt{in}} = f_i(\lceil (d+D_2'(k_{i-1})+3)/2 \rceil, 1+d)$. Thus \eqref{eqn:gen-pot-rec} implies the following recurrence for $R_\eps(x)$:
		\begin{align*}
			R_\eps(x) \le \begin{cases}
				\max\bigg(\big\{R_\eps(x - f_s(0,1)) + R_\eps(x - f_s(0,12))\big\}\ \cup&\\ \qquad \Big\{ R_\eps(x - t_{i,d,\texttt{out}}) + R_\eps(x - t_{i,d,\texttt{in}}) \colon i \in [s], z_i \le d \le 10 \Big\} \bigg) & \mbox{ if }x > 0,\\
				\poly(n,1/\epsilon) & \mbox{ otherwise.}
			\end{cases}
		\end{align*}
		It follows by~\eqref{eqn:bf-use} that $R_\eps(x) = O(T^x)\,\poly(n,1/\epsilon)$. We have $T \le 2$ by hypothesis, so $R_\eps(x) = O(2^x)\,\poly(n,1/\epsilon)$.
		
		By Observation~\ref{obs:TR-Z}, calculating $\calG_0$ from $G_0$ takes time $\poly(n)$, so our overall running time is $O(2^{f^+(G_0)})\,\poly(n,1/\epsilon)$. Moreover, by (VP\ref{VP-min}) and~(VP\ref{VP-top}), we have $f^+(\calG_0) \le f_s(\calG_0) = \sigma_sn$. The result therefore follows.
	\end{proof}

	We obtain an analogous result for bipartite graphs.  The proof is identical, except for using the bounds from the bipartite case of Lemma~\ref{lem:branch-bounds}(ii).
	
\begin{cor}\label{cor:fullalgo-runtime-bi}
Let $G$ be an $n$-vertex bipartite graph which is not a forest. 
Let $f$ be a valid pre-potential with boundary points $k_0,\dots,k_s$ and good slices $f_1,\dots,f_{s}$, writing $f_i(m,n) = \rho_im + \sigma_in$. For all $i \in [s]$ and all 
		$d \ge 6$,
		let
		\[
			T_{i,d} = \begin{cases}
				\tau\big(f_i(d,\ 1),\ f_i(D_2'(k_{i-1}),\ 1+d)\big)  & \mbox{ if $\sigma_i \ge 0$,}\\
				\tau\bigg(f_i(d,\ 1),\ f_i\Big(D_2'(k_{i-1}),\ \Big\lfloor\frac{d + D_2'(k_{i-1}) - 1}{2} \Big\rfloor \Big)\bigg)  & \mbox{ otherwise,}
			\end{cases}
		\]
		and 
		\[
			T = \max\Big(\big\{T_{i,d} \mid i \in [s], \max\{6,\floor{k_{i-1}}+1\} \le d \le 10\big\} \cup \big\{\tau(\sigma_s,12\sigma_s)\big\}\Big).
		\]
		If $T \le 2$, then $\iscount(\TR(G,\mathbf{1},\mathbf{1},1),\eps)$ has running time $O(2^{\sigma_s n})\,\poly(n,1/\epsilon)$. 	\end{cor}
	
	\section{A valid pre-potential}\label{sec:pre-pot}

	We include two valid pre-potentials, both in approximate form as Appendices~\ref{app:bi-pot} and~\ref{app:gen-pot} and in exact form as ancillary files \biprepot{} and~\genprepot{} (available on the arXiv). We use these, along with Corollaries~\ref{cor:fullalgo-runtime} and~\ref{cor:fullalgo-runtime-bi}, to bound the running time of $\iscount(\calG,\epsilon)$. To prove that they are valid, we will use the following lemma. (Note that property~(iv) is equivalent to requiring that $f$ be continuous.)
		
\begin{lemma}\label{lem:pre-pot}
Let $f$ be a pre-potential with slices $f_1,\dots,f_s$ and boundary points 
$-1 = k_0 < k_1 < \cdots <k_s=\infty$, 
writing $f_i(m,n) = \rho_i m + \sigma_i n$ for all $i \in [s]$. Then $f$ is a valid pre-potential if the following properties all hold.
\begin{enumerate}[(i)]
\item $\rho_s = 0$ and $\sigma_s > 0$.\label{PPa}
\item For all $i\in [s-1]$, $\rho_i + \sigma_i \geq 0$ \label{PPb}
\item For all $i\in [s-1]$, at least one of $\rho_i,\sigma_i$ is non-zero. \label{PPc} 
\item For all $i\in [s-1]$, $\rho_i \geq \rho_{i+1}$ and $\sigma_i \leq \sigma_{i+1}$. \label{PPd} 
\item Every integer in $[6,k_{s-1}]$ is a boundary point.\label{PPiii}
\item For all $i \in [s-1]$, $\rho_i = \rho_{i+1} + 2(\sigma_{i+1}-\sigma_i)/k_i$.\label{PPiv}
\end{enumerate}
\end{lemma}
\begin{proof}
Property (VP\ref{VP-good}) of Definition~\ref{def:pre-pot} follows from 
properties \eqref{PPa}, \eqref{PPb},  \eqref{PPc} and \eqref{PPd} of the lemma statement. 
Property (VP\ref{VP-top}) follows from 
property~\eqref{PPa}. 
Property (VP\ref{VP-bdry}) is property \eqref{PPiii}. 
To establish property~(VP\ref{VP-min}), consider
 $i \in [s]$. 
 Suppose that $n$ is a positive integer and $m$ is a non-negative integer
 satisfying $k_{i-1} < 2m/n \le k_i$. By property \eqref{PPiv}, for all $j > i$ we have
		\[
			\rho_i-\rho_j = \sum_{\ell=i}^{j-1}(\rho_{\ell}-\rho_{\ell+1}) 
= \sum_{\ell=i}^{j-1} \frac{2}{k_\ell} (\sigma_{\ell+1}-\sigma_{\ell}) 
			\le \frac{2}{k_i}\sum_{\ell=i}^{j-1}(\sigma_{\ell+1}-\sigma_{\ell}) = \frac{2}{k_i}(\sigma_j-\sigma_i).
		\]
		By property \eqref{PPd} and the choice of $n$ and $m$, it follows that  
				\begin{align*}
			f_j(m,n) - f_i(m,n) &= -(\rho_i - \rho_j)m + (\sigma_j - \sigma_i)n \ge (\sigma_j-\sigma_i)\Big(n-\frac{2}{k_i}m\Big) \ge 0.
		\end{align*}
		Thus for all $j>i$, $f_j(m,n) \ge f_i(m,n)$. Similarly, by property \eqref{PPiv}, for all $j<i$ we have
		\[
			\rho_j-\rho_i = \sum_{\ell=j}^{i-1}(\rho_{\ell}-\rho_{\ell+1}) 
	= \sum_{\ell=j}^{i-1} \frac{2}{k_\ell} (\sigma_{\ell+1}-\sigma_{\ell}) 				
			\ge \frac{2}{k_{i-1}}\sum_{\ell=j}^{i-1}(\sigma_{\ell+1}-\sigma_{\ell}) = \frac{2}{k_{i-1}}(\sigma_i-\sigma_j),
		\]
		so property \eqref{PPd} and the choice of $n$ and $m$ imply that for all $j<i$,
		\[
			f_j(m,n)-f_i(m,n)= (\rho_j - \rho_i)m - (\sigma_i - \sigma_j)n  \ge (\sigma_i-\sigma_j)\Big(\frac{2}{k_{i-1}}m-n\Big) \ge 0.
		\]
		Thus for all $j<i$, $f_j(m,n) \le f_i(m,n)$. We conclude that (VP\ref{VP-min}) holds as required.
	\end{proof}

    \section{Conclusion}
    \label{sec:conclusion}

We can now prove our main theorem.
	
	\thmmain*

\begin{proof}
Consider an input consisting of an $n$-vertex graph~$G$
and an error tolerance $\epsilon \in (0,1)$.

Let $\calG = \TR(G,\mathbf{1},\mathbf{1},1)$. 
The approximation scheme computes $\calG$ and then runs $\iscount(\calG,\eps)$. 
By Lemma~\ref{lem:iscount-works}(i), this outputs an $\eps$-approximation of $\IS(\calG)$, and by Observation~\ref{obs:TR-Z} and the definition of $\calG$ we have $\IS(G) = \IS(\calG)$, so the algorithm    
  outputs an approximation to~$Z(G)$, as required.      
  By Observation~\ref{obs:TR-Z}, we can compute $\calG$ in polynomial time, so it remains to bound the running time of $\iscount$.

		Recall from Section~\ref{sec:aad} that the Mathematica function \texttt{Dtwo[k]} in Appendix~\ref{app:mathematica-valid} outputs $D_2(k)$ given a rational argument $k \ge 2$. Thus, 
if the input file {\tt input} contains		
		  a potential function $f$ in the specified CSV format, the function \texttt{Validate[input]} given in Appendix~\ref{app:mathematica-valid} first checks whether the conditions of Lemma~\ref{lem:pre-pot} hold, and hence whether $f$ is a valid pre-potential. If $f$ is valid, then 
the run-time is $O(2^{\sigma_s n})\,\poly(n,1/\epsilon)$ by  
Corollary~\ref{cor:fullalgo-runtime} for general graphs and 
	Corollary~\ref{cor:fullalgo-runtime-bi} for bipartite graphs. For general graphs, $\sigma_s < 0.2680$, giving a running time of $O(2^{0.2680n})\,\poly(1/\epsilon)$; for bipartite graphs, $\sigma_s < 0.2372$.
As explained   at the beginning of 
	Appendix~\ref{app:mathematica-valid}, the code checks that the pre-conditions of these corollaries are satisfied and
then outputs the relevant bounds.
	 	\end{proof}

Our general running time for approximate counting is   
$O(2^{0.2680n})$ times a polynomial in~$1/\epsilon$, which is  
  better than the current best-known running time of $O(2^{0.3022n})$ for exact counting~\cite{GL}, and we have further improved on this when the input is bipartite. In the process of resolving our base case, we generalised an FPTAS of Sinclair et al.~\cite{SSSY} for 
approximating the partition function of
the hardcore model on graphs with low connective constant,
so that it also works
on graphs with arbitrary vertex weights. 
This may be useful  in other contexts.
For the recursive part of the algorithm, our approach followed the standard methods of the field, but uncovered interesting worst-case behaviour in bipartite graphs which required a more careful analysis to handle correctly.

\bibliographystyle{plain}
\bibliography{IS}	

\begin{thebibliography}{10}

\bibitem{bron}
C.~Bron and J.~Kerbosch.
\newblock Algorithm 457: Finding all cliques in an undirected graph.
\newblock {\em Comm. ACM}, 16:575--577, 1973.

\bibitem{csps}
Andrei~A. Bulatov, Martin~E. Dyer, Leslie~Ann Goldberg, Mark Jerrum, and Colin
  McQuillan.
\newblock The expressibility of functions on the boolean domain, with
  applications to counting {CSP}s.
\newblock {\em J. {ACM}}, 60(5):32:1--32:36, 2013.

\bibitem{DJW-old}
V.~Dahll\"of, P.~Jonsson, and M.~Wahlstr\"om.
\newblock Counting satisfying assignments in 2-{SAT} and 3-{SAT}.
\newblock In {\em Proc. Computing and Combinatorics, 8th Intl. Conference
  (COCOON 2002)}, pages 535--543. Springer, 2002.

\bibitem{DJW}
V.~Dahll\"of, P.~Jonsson, and M.~Wahlstr\"om.
\newblock Counting models for 2{SAT} and 3{SAT} formulae.
\newblock {\em Theoretical Computer Science}, 332(1--3):265--291, 2005.

\bibitem{Dubois}
O.~Dubois.
\newblock Counting the number of solutions for instances of satisfiability.
\newblock {\em Theoretical Computer Science}, 81:49--64, 1991.

\bibitem{relative}
Martin~E. Dyer, Leslie~Ann Goldberg, Catherine~S. Greenhill, and Mark Jerrum.
\newblock The relative complexity of approximate counting problems.
\newblock {\em Algorithmica}, 38(3):471--500, 2004.

\bibitem{FKbook}
Fedor~V. Fomin and Dieter Kratsch.
\newblock {\em Exact Exponential Algorithms}.
\newblock Springer, 2010.

\bibitem{FK}
M.~F{\"{u}}rer and S.~P. Kasiviswanathan.
\newblock Algorithms for counting 2-{SAT} solutions and colorings with
  applications.
\newblock Technical Report~33, Electronic Colloquium on Computational
  Complexity, 2005.

\bibitem{hcol}
Andreas Galanis, Leslie~Ann Goldberg, and Mark Jerrum.
\newblock Approximately counting $h$-colorings is $\#\mathrm{BIS}$-hard.
\newblock {\em {SIAM} J. Comput.}, 45(3):680--711, 2016.

\bibitem{bdpotts}
Andreas Galanis, Daniel Stefankovic, Eric Vigoda, and Linji Yang.
\newblock Ferromagnetic potts model: Refined {\#}{BIS}-hardness and related
  results.
\newblock {\em {SIAM} J. Comput.}, 45(6):2004--2065, 2016.

\bibitem{GL}
S.~Gaspers and E.~J. Lee.
\newblock Faster graph coloring in polynomial space.
\newblock In {\em Proc. Computing and Combinatorics, 23rd Intl. Conference
  (COCOON 2017)}, pages 371--383. Springer, 2017.
\newblock Full version: ArXiv CoRR abs/1607.06201.

\bibitem{Godsil}
C.~D. Godsil.
\newblock Matchings and walks in graphs.
\newblock {\em J. Graph Theory}, 5(3):285--297, 1981.

\bibitem{ferropotts}
Leslie~Ann Goldberg and Mark Jerrum.
\newblock Approximating the partition function of the ferromagnetic potts
  model.
\newblock {\em J. {ACM}}, 59(5):25:1--25:31, 2012.

\bibitem{JKP}
M.~Jenssen, P.~Keevash, and W.~Perkins.
\newblock Algorithms for \#{BIS}-hard problems on expander graphs.
\newblock In {\em Proc. 30th Annual ACM--SIAM Symposium on Discrete Algorithms
  (SODA 2019)}, pages 2235--2247. SIAM, 2019.

\bibitem{JT}
K.~Junosza-Szaniawski and M.~Tuczy{\'n}ski.
\newblock Counting independent sets via {D}ivide {M}easure and {C}onquer
  method.
\newblock CoRR abs/1503.08323, 2015.

\bibitem{LL-oneside}
J.~Liu and P.~Lu.
\newblock {FPTAS} for \#{BIS} with degree bounds on one side.
\newblock In {\em Proc. 47th Annual {ACM} {S}ymposium on Theory of {C}omputing
  (STOC 2015)}, pages 549--556. ACM, 2015.

\bibitem{provan}
J.~Scott Provan and Michael~O. Ball.
\newblock The complexity of counting cuts and of computing the probability that
  a graph is connected.
\newblock {\em {SIAM} J. Comput.}, 12(4):777--788, 1983.

\bibitem{SSSY}
A.~Sinclair, P.~Srivastava, D.~{\v{S}}tefankovi{\v{c}}, and Y.~Yin.
\newblock Spatial mixing and the connective constant: {O}ptimal bounds.
\newblock In {\em Proc. 26th Annual ACM--SIAM Symposium on Discrete Algorithms
  (SODA 2015)}, pages 1549--1563. SIAM, 2015.
\newblock Full version: ArXiv CoRR abs/1410.2595.

\bibitem{SSY}
A.~Sinclair, P.~Srivastava, and Y.~Yin.
\newblock Spatial mixing and approximation algorithms for graphs with bounded
  connective constant.
\newblock In {\em Proc. 54th Annual IEEE Symposium on Foundations of Computer
  Science (FOCS 2013)}, pages 300--309. IEEE Computer Society, 2013.
\newblock Full version: ArXiv CoRR abs/1308.1762.

\bibitem{Sly}
A.~Sly.
\newblock Computational transition at the uniqueness threshold.
\newblock In {\em Proc. 51st Annual IEEE Symposium on Foundations of Computer
  Science (FOCS 2010)}, pages 287--296. ACM, 2010.
\newblock Full version: ArXiv CoRR abs/1005.5584.

\bibitem{WahlPhD}
M.~Wahlstr{\"{o}}m.
\newblock {\em Algorithms, Measures and Upper Bounds for Satisfiability and
  Related Problems}.
\newblock PhD thesis, Link{\"{o}}pings universitet, 2007.

\bibitem{WahlPaper}
M.~Wahlstr{\"o}m.
\newblock A tighter bound for counting max-weight solutions to {2SAT}
  instances.
\newblock In {\em Proc. 3rd Intl. Conference on Parameterized and Exact
  Computation (IWPEC 2008)}, pages 202--213. Springer, 2008.

\bibitem{Weitz}
D.~Weitz.
\newblock Counting independent sets up to the tree threshold.
\newblock In {\em Proc. 38th Annual {ACM} {S}ymposium on Theory of {C}omputing
  (STOC 2006)}, pages 140--149. ACM, 2006.

\bibitem{XZZ}
M.~Xia, P.~Zhang, and W.~Zhao.
\newblock Computational complexity of counting problems on 3-regular planar
  graphs.
\newblock {\em Theoretical Computer Science}, 384:111--125, 2007.

\bibitem{Zhang}
W.~Zhang.
\newblock Number of models and satisfiability of sets of clauses.
\newblock {\em Theoretical Computer Science}, 155:277--288, 1996.

\end{thebibliography}
	
\clearpage
 
\appendix

\section{Index of Notation and Terms}\label{sec:index}

{\footnotesize
\renewcommand{\arraystretch}{1.5}
\begin{longtable}{@{}l|lll@{}}
$[n]$ & $\{1,\ldots,n\}$ & Sec.~\ref{sec:notation} & p.~\pageref{sec:notation}\\[1ex]
$\alpha(v)$, $\beta(v)$ & Used in associated average degree & Sec.~\ref{sec:aad} & p.~\pageref{sec:aad}\\
$\Gamma_G(S)$ & $\bigcup_{v \in S}\Gamma_G(v) \setminus S$ & Sec.~\ref{sec:notation} & p.~\pageref{sec:notation}\\
$\Gamma_G(v,X)$ & $\{w\in X \mid \{v,w\} \in E(G)\}$ & Sec.~\ref{sec:notation} & p.~\pageref{sec:notation}\\
$\Gamma^2_G(v)$ &   Vertices at distance   $2$ from~$v$ & Sec.~\ref{sec:notation} & p.~\pageref{sec:notation}\\
$\Gamma_v$ & $\Gamma_G(v)$ in standard decomposition & Def.~\ref{def:decomp} & p.~\pageref{def:decomp} \\
$(\Gamma_v,S,X,H_1,\dots,H_k;T_1,\dots,T_\ell)$ & Standard decomposition of $G$ w.r.t.\@ vertex~$v$& Def.~\ref{def:decomp} & p.~\pageref{def:decomp} \\
$\delta(G)$, $\Delta(G)$ & Minimum and maximum degrees & Sec.~\ref{sec:notation} & p.~\pageref{sec:notation}\\
$\kappa$ & Connective constant & Def.~\ref{def:cc-2} & p.~\pageref{def:cc-2}\\
$\lambda$ & Parameter of univariate hard-core model & Sec.~\ref{sec:gadgets} & p.~\pageref{def:hardcore}\\
$\lambda$-balanced & $\win(v)\leq\lambda\wout(v)$& Sec.~\ref{sec:notation} & p.~\pageref{def:lambal}\\
$\lambda_t$ & $\lambda(1+\lambda)^{-t}$ & Sec.~\ref{sec:gadgets} & p.~\pageref{defLamt}\\
$\rho$, $\sigma$ & Coefficients of slice of a potential function & Def.~\ref{defn:good-slice} & p.~\pageref{defn:good-slice}\\
$\tau(a_1, \dots, a_b)$ & Positive root of $\sum_{i=1}^b x^{-a_i} = 1$ & Sec.~\ref{sec:recursion} & p.~\pageref{sec:recursion} \\
$\phi$ & Weight map for $G$ & Def.~\ref{def:gadg} & p.~\pageref{def:gadg} \\[1ex]
$a_k(\bx)$ & Used in numerical associated average degree & Def.~\ref{def:rom-aad} & p.~\pageref{def:rom-aad}\\
$\aad(v)$ & Associated average degree & Sec.~\ref{sec:aad} & p.~\pageref{sec:aad}\\
$\saad_k(\bx)$ & Numerical associated average degree & Def.~\ref{def:rom-aad} & p.~\pageref{def:rom-aad}\\
$\hcalgo_{\lambda,\calF}(\calG,\eps)$ & Approximates $Z(\calG)$ when $\kappa$ is subcritical & Sec.~\ref{sec:gadgets} & p.~\pageref{def:alg1}\\
$b_k(\bx)$ & Used in numerical associated average degree & Def.~\ref{def:rom-aad} & p.~\pageref{def:rom-aad}\\
$\basecount(\calG,\epsilon)$ & FPTAS for base case &  Sec.~\ref{sec:local} & p.~\pageref{thm:2-degree-FPTAS}\\
branch & Transform into  smaller instances $\Gout$, $\Gin$ & Sec.~\ref{sec:intro} & p.~\pageref{def:branch} \\ 
$C(v)$ & Set of children of vertex~$v$ in a tree &  Sec.~\ref{sec:local} & p.~\pageref{def:treedefs}\\
$\iscount(\calG,\eps)$ & Main approximation algorithm & Sec.~\ref{algo:count} & p.~\pageref{algo:count} \\
$D(v)$ & Depth of vertex~$v$ in a tree & Sec.~\ref{sec:local} & p.~\pageref{def:treedefs}\\
$\ourfamily$ & Base case graph family handled in Sec.~\ref{sec:local} & Def.~\ref{defn:ourfamily} & p.~\pageref{defn:ourfamily}\\
$d_G(v)$ & Degree of vertex~$v$ in graph~$G$ & Sec.~\ref{sec:intro} & p.~\pageref{def:dGv} \\ 
$d_G(v,X)$ &  $|\Gamma_G(v,X)|$ & Sec.~\ref{sec:notation} & p.~\pageref{sec:notation}\\
$d^2_G(v)$, $2$-degree of~$v$&  $ \sum_{\{v,w\} \in E(G)} d_G(w)$ & Sec.~\ref{sec:notation} & p.~\pageref{sec:notation}\\
$D_2(k)$ & Function used to bound $2$-degree  & Def.~\ref{def:D2} & p.~\pageref{def:D2} \\
$D_2'(k)$ & Adjusted version of $D_2(k)$ & Def.~\ref{def:Dtwoprime} & p.~\pageref{def:Dtwoprime} \\
$\kappa$-decreasing& Used to bound connective constant  & Def.~\ref{def:kapdec} & p.~\pageref{def:kapdec}\\
$\brute(\calG,Y)$ & Exact computation when $\calG-Y$ is a forest & Sec.~\ref{sec:slices} & p.~\pageref{algo:brute} \\
$f$ & Slice of a potential function & Def.~\ref{defn:good-slice} & p.~\pageref{defn:good-slice}\\
$f_1, \dots, f_s$ & Slices of a pre-potential & Def.~\ref{def:pre-pot} & p.~\pageref{def:pre-pot} \\
$f^+$ & Potential made from pre-potential~$f$ & Def.~\ref{def:pre-pot} & p.~\pageref{def:pre-pot} \\
$\calF^+_y$ & Set of gadgets (realisations of weight maps) & Def.~\ref{def:Fplus} & p.~\pageref{def:Fplus}\\
$\calG = (G,\win,\wout,W)$ & Weighted graph & Sec.~\ref{sec:notation} & p.~\pageref{def:calG}\\
$G'$ & Realisation of a weight map & Def.~\ref{def:gadg} & p.~\pageref{def:gadg} \\
$\calG'$ & $\Red(\calG)$ & Lem.~\ref{lem:branch-bounds} & p.~\pageref{def:Gdash} \\
$G-X$, $\calG-X$ & Graph made by removing $X$ & Sec.~\ref{sec:notation} & p.~\pageref{def:GminX}\\
$\calG[X]$ & Induced weighted graph & Sec.~\ref{sec:notation} & p.~\pageref{def:calGX}\\
good slice & A slice with $\rho \geq 0$, $\rho \geq -\sigma$ and $\rho\sigma\neq 0$ & Def.~\ref{defn:good-slice} & p.~\pageref{defn:good-slice}\\
$H$ & $G-X^+$ in extended decomposition & Def.~\ref{def:extend-decomp} & p.~\pageref{def:extend-decomp} \\
$H_1, \dots, H_k$ & Non-tree components in standard decomposition & Def.~\ref{def:decomp} & p.~\pageref{def:decomp} \\
$\calI(G)$ & Independent sets of $G$ & Sec.~\ref{sec:notation} & p.~\pageref{def:calIG}\\
$k_1, \dots, k_s$ & Boundary points of a pre-potential & Def.~\ref{def:pre-pot} & p.~\pageref{def:pre-pot} \\
$L(T)$ & Set of tree $T$'s leaves & Sec.~\ref{sec:local} & p.~\pageref{def:treedefs}\\
$m(G)$ & Number of edges of $G$ & Sec.~\pageref{sec:notation} & p.~\pageref{sec:notation}\\
$n(G)$ & Number of vertices of $G$ & Sec.~\pageref{sec:notation} & p.~\pageref{sec:notation}\\
$N_G(v,\ell)$ & Number of length-$\ell$ simple paths from~$v$ in~$G$ &  Def.~\ref{def:saw} & p.~\pageref{def:saw}\\
near-forest & $G-v-\Gamma_G(v)$ a forest for some~$v$ & Def.~\ref{def:near-forest} & p.~\pageref{def:near-forest} \\
$P$ & $y$--$z$ path in extended decomposition& Def.~\ref{def:extend-decomp} & p.~\pageref{def:extend-decomp} \\
pre-potential & Function used to define a potential & Def.~\ref{def:pre-pot} & p.~\pageref{def:pre-pot} \\
$\Prune(\calG,S)$ & Weighted graph based on deleting $S$ & Def.~\ref{defn:prune} & p.~\pageref{defn:prune} \\
$\Red(\calG)$ & Produces a reduced graph from~$\calG$ & Sec.~\ref{algo:reduce} & p.~\pageref{algo:reduce} \\
reduced & $G[X]$ a near forest implies $|\Gamma_G(X)|\geq 2$ & Def.~\ref{def:reduced} & p.~\pageref{def:reduced} \\
$S$ & $\Gamma^2_G(v) \cap \bigcup_i V(H_i)$ in standard decomposition & Def.~\ref{def:decomp} & p.~\pageref{def:decomp} \\
$S'$ & $S \cap \Gamma_G(v)$ (in $\Prune$) & Def.~\ref{defn:prune} & p.~\pageref{defn:prune} \\
$S''$ & $S \setminus \Gamma_G(v)$ (in $\Prune$) & Def.~\ref{defn:prune} & p.~\pageref{defn:prune} \\
$\kappa$ subcritical wrt $\lambda$ & $\lambda < \lambda_c(\kappa)= \kappa^\kappa/{(\kappa-1)}^{\kappa+1}$ & Def.~\ref{def:subcritical} & p.~\pageref{def:subcritical}\\ 
suitable, (S1)--(S4) & Used to pin down $D_2(k)$ & Def.~\ref{def:suitable} & p.~\pageref{def:suitable} \\
$T(x)$ & Running time when instance potential is~$x$& Def.~\ref{sec:recursion} & p.~\pageref{sec:recursion} \\
$T_1, \dots, T_\ell$ & Tree components in standard decomposition & Def.~\ref{def:decomp} & p.~\pageref{def:decomp} \\
$\sawt{v}{G}$ & Self-avoiding walk tree & Def.~\ref{def:saw} & p.~\pageref{def:saw}\\
$\TR(\calG)$ & Graph made by deleting tree components & Def.~\ref{def:TR} & p.~\pageref{def:TR} \\
$\win(X)$, $\wout(X)$ &  $\prod_{x \in X}\win(x)$, $\prod_{x \in X}\wout(X)$   & Sec.~\ref{sec:notation} & p.~\pageref{def:wplusmin}\\
(VP1)--(VP4) & Properties of a valid pre-potential & Def.~\ref{def:pre-pot} & p.~\pageref{def:pre-pot} \\
$x$, $y$ & $S=\{x,y\}$ in standard/extended decomposition & Def.~\ref{def:extend-decomp} & p.~\pageref{def:extend-decomp} \\
$X$ & $\{v\}\cup\Gamma_v\cup \bigcup_i V(T_i)$ in standard decomposition & Def.~\ref{def:decomp} & p.~\pageref{def:decomp} \\
$X^+$ & $X\cup V(P)\setminus\{x\}$ in extended decomposition & Def.~\ref{def:extend-decomp} & p.~\pageref{def:extend-decomp} \\
$(X^+, P, x, y, z, H)$ & Extended decomposition of $G$ & Def.~\ref{def:extend-decomp} & p.~\pageref{def:extend-decomp} \\
$z$ & Last vertex of $P$ in extended decomposition & Def.~\ref{def:extend-decomp} & p.~\pageref{def:extend-decomp} \\
$Z(\calG)$ & Partition function & Sec.~\ref{sec:notation} &p.~\pageref{def:partition} \\ 
$Z_\lambda(G)$ & Univariate hard-core partition function & Sec.~\ref{sec:gadgets} & p.~\pageref{def:hardcore}
\end{longtable}
}

\section{Mathematica code for Section~\ref{sec:conn-const}}
\label{app:mathematica:conn-const}

The following Mathematica code is used to rigorously prove
Theorem~\ref{thm:2-degree-FPTAS}.  \texttt{FindInstance} attempts to
find integer solutions to the given equations over the given
variables.  It returns ``\texttt{\{\}}'', indicating (with certainty) that there are no
counterexamples to the claim in the theorem.

\vbox{
\small
\begin{verbatim}
psi[d_, p_] :=
    If[d == 2, 245/1000,
    If[d == 3, 456/1000,
    If[d == 4, 647/1000,
    If[d == 5, 859/1000,
    (* Else, d >= 6 *)
        If[p == 2, 1,
        If[p == 3, 941/1000,
        (*Else, p >= 4*) 889/1000
        ]]
    ]]]]
\end{verbatim}
}

\vbox{
\begin{verbatim}
FindInstance[
(* Equations *)
    0 <= n[2] <= 12  &&  0 <= n[6] <= 12  &&  0 <= n[10] <= 12  &&
    0 <= n[3] <= 12  &&  0 <= n[7] <= 12  &&  0 <= n[11] <= 12  &&
    0 <= n[4] <= 12  &&  0 <= n[8] <= 12  &&  0 <= n[12] <= 12  &&
    0 <= n[5] <= 12  &&  0 <= n[9] <= 12  &&  0 <= n[13] <= 12  &&
    2 <= p <= 13  &&  2 <= d <= 13  &&
    d == 1 + Sum[n[i], {i, 2, 13}]  &&
    (d <= 5 || p + Sum[i*n[i], {i, 2, 13}] <= 26)  &&
    Sum[n[i] psi[i, d], {i, 2, 13}] > (4141/1000) psi[d, p],
(* Variables *)
    {n[2], n[3], n[4], n[5], n[6], n[7], n[8], n[9], n[10], n[11], n[12], n[13],
     p, d}, 
(* Domain *)
    Integers]
\end{verbatim}
}

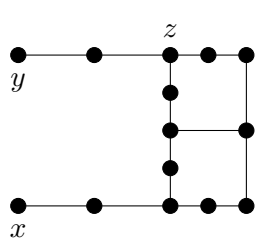
\begin{figure}[p!]
\begin{center}
\begin{tikzpicture}[scale=1,node distance = 1cm]
    \tikzstyle{vtx}=[fill=black, draw=black, circle, inner sep=2pt]
    \tikzstyle{svtx}=[fill=black, draw=black, circle, inner sep=1.5pt]
    \tikzstyle{vset}=[gray,rounded corners=5pt,dashed];

 
    \node at (-1.25,4) {$\calG$:};
    \node[vtx] (v)  at (0,2) [label=180:{$v$}] {};
    \node[vtx] (g1) at (1,0) {};
    \node[vtx] (g2) at (1,2) {};
    \node[vtx] (g3) at (1,3) {};
    \node[vtx] (g4) at (1,4) {};

    \node[vtx] (y)  at (2,2) [label=270:{$y$}] {};
    \node[vtx] (y') at (3,2) {};
    \node[vtx] (t1) at (2,3) {};
    \node[vtx] (t2) at (2,4) {};

    \draw (g1)--(v)--(g2) (g3)--(v)--(g4);
    \draw (y')--(y)--(g2)--(t1)--(g3)--(t2)--(g4);
    \draw (g1)--(2,0);
    \draw (3,2)--(4,2);

    \newcommand{\drawHstar}{
        \node[vtx] at (2,0) [label=270:{$x$}] {};
        \node[vtx] at (3,0)   {};
        \node[vtx] at (4,0)   {};
        \node[vtx] at (4,0.5) {};
        \node[vtx] at (4,1)   {};
        \node[vtx] at (4,1.5) {};
        \node[vtx] at (4,2) [label=90:{$z$}] {};
        \node[vtx] at (4.5,0) {};
        \node[vtx] at (4.5,2) {};
        \node[vtx] at (5,0)   {};
        \node[vtx] at (5,1)   {};
        \node[vtx] at (5,2)   {};

        \draw (2,0) -- (5,0) -- (5,2) -- (4,2) -- (4,0);
        \draw (4,1) -- (5,1);
    }
    \drawHstar{}

    \node at (-0.3,4) {$X$};
    \draw[vset] (-0.65,1) -- (-0.65,4.4) -- (2.4,4.4) -- (2.4,2.4)
                  -- (1.5,2.4) -- (1.5,-0.4) -- (-0.65,-0.4) -- (-0.65,1);

    \node at (3.2,4.2) {$X^+$};
    \draw[vset] (-0.85,3) -- (-0.85,4.6) -- (3.6,4.6) -- (3.6,1)
                  -- (1.7,1) -- (1.7,-0.6) -- (-0.85,-0.6) -- (-0.85,3);


    \node at (-1.65,-1.5) [label=0:{a) branching on $y$}:] {};
 
 
    \begin{scope}[shift={(-2,-5)}]
        \drawHstar{}
        \node[vtx] at (3,2) {};
        \draw (3,2) -- (4,2);
        \node at (1.5,2.8) [label=0:{$\Gout=\Prune(\calG-y,X)$:}] {};
    \end{scope}

 
    \begin{scope}[shift={(3.5,-5)}]
        \drawHstar{}
        \node at (1.5,2.8) [label=0:{$\Gin=\Prune(\calG-y-\Gamma_G(y),X\setminus\Gamma_G(y))$:}] {};
    \end{scope}


    \node at (-1.65,-6.25) [label=0:{b) branching on $z$}:] {};
 
 
    \begin{scope}[shift={(-2,-9.45)}]
        \node[vtx] at (2,0) [label=270:{$x$}] {};
        \node[vtx] at (3,0) {};
        \node[vtx] at (4,0) {};
        \node[vtx] at (4,0.5) {};
        \node[vtx] at (4,1) {};
        \node[vtx] at (4,1.5) {};
        \node[vtx] at (4.5,0) {};
        \node[vtx] at (4.5,2) {};
        \node[vtx] at (5,0) {};
        \node[vtx] at (5,1) {};
        \node[vtx] at (5,2) {};

        \draw (2,0) -- (5,0) -- (5,2) -- (4.5,2);
        \draw (4,0) -- (4,1.5);
        \draw (4,1) -- (5,1);

        \node at (1.5,2.5) [label=0:{$\Gout=\TR(\Prune(\calG-z,X^+))$:}] {};
    \end{scope}

 
    \begin{scope}[shift={(3.5,-9.45)}]
        \node[vtx] at (2,0) [label=270:{$x$}] {};
        \node[vtx] at (3,0) {};
        \node[vtx] at (4,0) {};
        \node[vtx] at (4,0.5) {};
        \node[vtx] at (4,1) {};
        \node[vtx] at (4.5,0) {};
        \node[vtx] at (5,0) {};
        \node[vtx] at (5,1) {};
        \node[vtx] at (5,2) {};

        \draw (2,0) -- (5,0) -- (5,2);
        \draw (5,1) -- (4,1) -- (4,0);

        \node at (1.5,2.5) [label=0:{$\Gin=\TR(\Prune(\calG-z-\Gamma_G(z),X^+\setminus\Gamma_G(z)))$:}] {};
    \end{scope}


    \node at (-1.65,-10.8) [label=0:{c) branching on $v$}:] {};
 
 
    \begin{scope}[shift={(-2,-15)}]
        \drawHstar{};
        \node[vtx] at (1,0)   {};
        \node[vtx] at (1,2)   {};
        \node[vtx] at (2,2) [label=270:{$y$}] {};
        \node[vtx] at (1,2.5) {};
        \node[vtx] at (1,3)   {};
        \node[vtx] at (2,2)   {};
        \node[vtx] at (2,2.5) {};
        \node[vtx] at (2,3)   {};
        \node[vtx] at (3,2)   {};

        \draw (4,2) -- (1,2) -- (2,2.5) -- (1,2.5) -- (2,3) -- (1,3);
        \draw (1,0) -- (2,0);

        \node at (1.5,3.6) [label=0:{$\Gout=\TR(\calG-v)$:}] {};
    \end{scope}

 
    \begin{scope}[shift={(3.5,-14.5)}]
        \drawHstar{};

        \node[vtx] at (2,2) [label=270:{$y$}] {};
        \node[vtx] at (3,2)   {};

        \draw (4,2) -- (2,2);

        \node at (1.5,3.1) [label=0:{$\Gin=\TR(\calG-v-\Gamma_G(v))$:}] {};
    \end{scope}

\end{tikzpicture}
\end{center}
\caption{An example showing that our branching strategy using the extended decomposition improves on other branching strategies.}\label{fig:ext-decomp-needed}
\end{figure}

\section{Better branching with the extended decomposition}
\label{app:ext-decomp-needed}
 
Consider the (bipartite) graph~$\calG$ shown in Figure~\ref{fig:ext-decomp-needed}, with underlying graph $G$. Observe that for all $\emptyset \subset X \subset V(G)$ we have  
$|\Gamma_G(X)| \ge 2$, 
and that $G$ itself contains two disjoint cycles 
whose distance is greater than~$2$, so $G$ is therefore
not a near-forest. It follows that $G$ is reduced. Let $f(m,n) = m-n$, which is a good slice.  We have $f(\calG)=25-21=4$.  Writing $f(m,n)=\rho m + \sigma n$ as in Definition~\ref{defn:good-slice}, we have $\rho = 1$ and $\sigma = -1$.  For the calculations below, note that $d_G(v)=4$ and $d^2_G(v)=10$.

The strategy of~\cite{DJW} is to branch on~$y$ and prune what remains of~$X$, as shown in Figure~\ref{fig:ext-decomp-needed}(a).  This yields instances $\Gout$ and~$\Gin$ with $f(\Gout)=14-13=1$ and $f(\Gin)=13-12=1$, so the value of~$f$ decreases by~$3$ in each branch.

Instead, we branch on~$z$ and prune what remains of $X^+$, as in Figure~\ref{fig:ext-decomp-needed}(b).  This gives instances with $f(\Gout)=11-11=0$ and $f(\Gin)=8-8=0$, so $f$~decreases by~$4$ in each branch.  Note that this example shows that Lemma~\ref{lem:clever-branch} is tight. It gives
\begin{align*}
    f(\calG)-f(\Gout)
        &\geq \rho(1+d_G(v)) + \sigma
        = 1 + 4 - 1
        = 4 \\
    f(\calG)-f(\Gin)
        &\geq \rho \lceil \tfrac12(d^2_G(v)+d_G(v)+4)\rceil +\sigma(1+d_G(v))\\
        &\qquad= \lceil \tfrac12(10+4+4)\rceil - 5
        = 4\,,
\intertext{and, using the fact that $\calG$~is bipartite (and $\sigma<0$),}
    f(\calG)-f(\Gin)
        &\geq \rho\,d^2_G(v) + \sigma\lceil\tfrac12(d^2_G(v)+d_G(v)-2)\rceil\\
        &\qquad= 10 - \lceil\tfrac12(10+4-2)\rceil
        = 4\,.
\end{align*}

Finally, we could instead branch on~$v$, as in Figure~\ref{fig:ext-decomp-needed}(c).  In this case, there is nothing to prune. We obtain $f(\Gout)=21-20=1$ and $f(\Gin)=15-14=1$, so $f$~decreases by~$3$ in each branch~--- the same decrease as branching on~$y$.  This also shows that Lemma~\ref{lem:base-branch} is tight. $S=\{x,y\}$ so the lemma gives
\begin{align*}
    f(\calG)-f(\Gout)
        &\geq \rho\,d_G(v)+\sigma
        = 4-1
        = 3 \\
    f(\calG)-f(\Gin)
        &\geq \rho \lceil\tfrac12(d^2_G(v)+d_G(v)+|S|)\rceil + \sigma(1+d_G(v))\\
        &\qquad= \lceil\tfrac12(10+4+2)\rceil - (1+4)
        = 3\,,
\intertext{and, since $\calG$ is bipartite,}
    f(\calG)-f(\Gin)
        &\leq \rho\,d^2_G(v) + \sigma\lfloor\tfrac12(d^2_G(v)+d_G(v)+2-|S|)\rceil\\
        &\qquad= 10 - \lfloor\tfrac12(10+4+2-2)\rfloor
        = 3\,.
\end{align*}

In summary, we have shown that the bounds of Lemma~\ref{lem:base-branch} are tight for the strategy of bounding on $v$, that the bounds of Lemma~\ref{lem:clever-branch} are tight for the strategy of branching on $z$.
We also showed that branching on~$z$ is better in this case because the it leads to a higher lower bound on 
$f(\calG)-f(\Gout)$.

\section{Mathematica code for Section~\ref{sec:conclusion} and auxiliary files}
\label{app:mathematica-valid}

The Mathematica code in  this section 
(also contained in the file \validatecode{} on the ArXiv)
is used in Section~\ref{sec:conclusion} to prove Theorem~\ref{thm:main} with the aid of our pre-potential functions 
(see files \biprepot{} and \genprepot, also on the ArXiv).   
 The function $\texttt{Dtwo[k]}$ computes $D_2(k)$ using the method from Section~\ref{sec:aad}.
 The function \texttt{Validate[]} takes as input the name of a file in CSV (comma-separated values) format that describes a pre-potential function which we will call~$f$.  If the first line of the file is exactly the string \texttt{Bipartite}, the potential function is interpreted as being for use with bipartite graphs; otherwise, it is for general graphs.  The second, third and fourth lines of the file are, respectively, the list of values $\rho_1, \dots, \rho_s$, the list of values $\sigma_1, \dots, \sigma_s$ and the list of values $k_1, \dots, k_{s-1}$.  To allow exact computation without numerical approximation, values are given as fractions, e.g., \texttt{12/13}.

\texttt{Validate[]} first checks that the 
input is a valid pre-potential as defined in Definition~\ref{def:pre-pot}.  
To do this, it checks that the
 conditions of Lemma~\ref{lem:pre-pot} hold and aborts if  they do not.  

\texttt{Validate[]} then verifies the conditions of Corollary~\ref{cor:fullalgo-runtime} (in the general case) and Corollary~\ref{cor:fullalgo-runtime-bi} (in the bipartite case) to verify that $f$~meets the conditions required to give a running time of $O(s^{\sigma_s n})\,\poly(n,1/\epsilon)$. Again, it aborts if any check fails.  In more detail, the following steps are taken for each $i\in[s]$.
\begin{itemize}
\item Compute $D_2'(k_{i-1})$ (stored in the variable \texttt{D2p}) according to Definition~\ref{def:Dtwoprime}.
\item For appropriate values of~$d$, as given by Corollary~\ref{cor:fullalgo-runtime} or~\ref{cor:fullalgo-runtime-bi}, compute values $\texttt{exp1}$ and $\texttt{exp2}$ such that $T_{i,d} = \tau(\texttt{exp1}, \texttt{exp2})$.  Thus, $\texttt{exp1} = f_i(d,1)$, and $\texttt{exp2}$ corresponds to some expression of the form $f_i(\cdot,\cdot)$ that depends on $d$, $i$ and whether or not the input graph is bipartite.  As observed in Section~\ref{sec:recursion}, if $2^{-\texttt{exp1}}+2^{-\texttt{exp2}} > 1$, then $T_{d,i} > 2$ and, hence, $T$ (as defined in Corollary \ref{cor:fullalgo-runtime} or~\ref{cor:fullalgo-runtime-bi} as appropriate) is greater than~$2$ and the code aborts.
\end{itemize}
In addition, the code aborts if $2^{-\sigma_s} + 2^{-12\sigma_s} > 1$. If so,  $T> 2$, so the code aborts.

If the code has not aborted as a result of any of these checks, we conclude that $T\leq 2$ and, thus the code outputs that $\iscount$ has running time $O(2^{\sigma_s n})\, \poly(n,1/\epsilon)$.

\vbox{\small
\begin{verbatim}
IsSuitable[k_, z_] := Module[ {d, K, s, q, d0, d1}, 
  K = Floor[k] + 1;
  If [z >= K^2, Return[True]];
  For [d = K, d <= z/2, d++,
   For [s = 0, s < d, s++,
    If [(z >= K s + 2 (d - s)) && ( z <= K s + (K-1)(d - s)),
     q = Floor[(z - K s)/(d - s)];
     d1 = Mod[z - K s, d-s];
     d0 = d - s - d1; 
     If [(d+d0+d1)/(1 + d0/q + d1/(q + 1)) > k, Return[True]]]]];
  Return[False]]
 \end{verbatim}}

\vbox{\small
\begin{verbatim}
Dtwo[k_] := Module[{K, z}, 
   K = Floor[k] + 1;
   For[z = 2 K, z < K^2, z++,
    If[ IsSuitable[k, z], Return[z] ]]; Return[K^2]];
 \end{verbatim}}
  
\hspace{-2truecm}  \vbox{\small
\begin{verbatim}
Validate[input_] := (
   Quiet[SetDirectory[NotebookDirectory[]]];
   file = Import[input, "CSV"];
   isBipartite = (file[[1]] == {"Bipartite"});
   rho = Map[ToExpression, file[[2]]];
   sigma = Map[ToExpression, file[[3]]];
   k = Map[ToExpression, file[[4]]];
   s = Length[rho];  
   (* check that the input is a valid pre-potential *)
   For [ i = 1, i <= s - 2, i++, If[ k[[i]] >= k[[i + 1]], Print["k's not ascending"]; Abort[]]];
   If[rho[[s]] != 0 || sigma[[s]] <= 0, Print["Fails Lemma 48(i)"]; Abort[]];
   For [i = 1, i <= s - 1, i++,
      If [rho[[i]] + sigma[[i]] < 0, Print["Fails Lemma 48(ii)"]; Abort[]];
      If [rho[[i]] == 0 && sigma[[i]] == 0, Print["Fails Lemma 48(iii)"];Abort[]];
      If [ rho[[i]] < rho[[i + 1]] || sigma[[i]] > sigma[[i + 1]], 
         Print["Fails Lemma 48(iv)"]; Abort[]];
    ];
   For [j = 6, j <= k[[s - 1]], j ++, If [! MemberQ[k, j], Print["Fails Lemma 48(v)"]; Abort[]]];
   (* Establish the runtime guarantee via Corollary 46 or 47. T_{i,d} = tau(exp1,exp2) *)
   For [i = 1, i <= s, i ++,
       If [i == 1 || k[[i - 1]] < 5 , D2p = 27, D2p = Max[2 k[[i - 1]], Dtwo[k[[i - 1]] ]]];
       For [d = If [i == 1, 6, Max[Floor[k[[i - 1]]] + 1, 6]], d <= 10, d ++,
          exp1 = (d rho[[i]] + sigma[[i]]); (* exp1 = f_i(d,1) *)
          If[ isBipartite,
             If[sigma[[i]] >= 0,
                exp2 = D2p * rho[[i]] + (1 + d)  sigma[[i]],(* exp2=f_i(D2p,1+d) 1st part Cor 27*)
                exp2 = D2p * rho[[i]] + Floor[(d + D2p  - 1) / 2]  sigma[[i]]],(*2nd part Cor 27*)
             exp2 = Ceiling[(d + D2p + 3) / 2] * rho[[i]] + (1 + d)  sigma[[i]] ]; (* Cor 26 *)
         If[2^(-exp1) + 2^(-exp2) > 1,
            Print["Can't apply Cor 46/47 because T_{i,d}>2 for i =", i, " and d= ", d]; Abort[]]; 
         If[2^(-sigma[[s]]) + 2^(-12 sigma[[s]]) > 1, 
            Print["Can't apply Cor 46/47 because tau(sigma_s,12 sigma_s)>2"]; Abort[]]]];
   Print["Potential function is valid, running time O(2^(",
      Ceiling[sigma[[s]], 0.0001],"n))*poly(1/eps), 
      which is O(",Ceiling[2^sigma[[s]], 0.0001], "^n)*poly(1/eps))."]);
  \end{verbatim}
  }

 \hspace{-2truecm} \vbox{\small
 \begin{verbatim}
Validate["bi-potential.csv"] ;

Validate["potential.csv"]  
\end{verbatim}
}

\hspace{-2truecm}\vbox{
This code will output:
{\small
\begin{verbatim}
Potential function is valid, running time O(2^0.2372n)*poly(1/eps), 
    which is O(1.1788^n)*poly(1/eps)).
Potential function is valid, running time O(2^0.2680n)*poly(1/eps), 
    which is O(1.2042^n)*poly(1/eps)).
\end{verbatim}
}}

\section{A pre-potential function for bipartite graphs}\label{app:bi-pot}

The following table contains an approximation of our pre-potential function for bipartite graphs, rounded away from zero to five decimal places. We give an exact version in the ancillary file \biprepot{}, available on the arXiv.

\rowcolors{1}{white}{lightgray}
\begin{longtable}{|l|l|l|}
	\hline
	Boundary points $k_0,\dots,k_s$ & Edge weights $\rho_1,\dots,\rho_s$ & Vertex weights $\sigma_1,\dots,\sigma_s$\\
	\hline
    $k_0 = -1$ & & \\
	$k_{1} = 4$ & $\rho_{1} = 0.13168$ & $\sigma_{1} = -0.13168$\\
	$k_{2} = 4.18794$ & $\rho_{2} = 0.13168$ & $\sigma_{2} = -0.13168$\\
	$k_{3} = 4.23794$ & $\rho_{3} = 0.13168$ & $\sigma_{3} = -0.13168$\\
	$k_{4} = 4.28572$ & $\rho_{4} = 0.1241$ & $\sigma_{4} = -0.11562$\\
	$k_{5} = 4.33572$ & $\rho_{5} = 0.11588$ & $\sigma_{5} = -0.098$\\
	$k_{6} = 4.38572$ & $\rho_{6} = 0.10882$ & $\sigma_{6} = -0.0827$\\
	$k_{7} = 4.43572$ & $\rho_{7} = 0.10176$ & $\sigma_{7} = -0.06721$\\
	$k_{8} = 4.44445$ & $\rho_{8} = 0.0957$ & $\sigma_{8} = -0.05377$\\
	$k_{9} = 4.49445$ & $\rho_{9} = 0.0957$ & $\sigma_{9} = -0.05377$\\
	$k_{10} = 4.55$ & $\rho_{10} = 0.08753$ & $\sigma_{10} = -0.03542$\\
	$k_{11} = 4.57143$ & $\rho_{11} = 0.08317$ & $\sigma_{11} = -0.0255$\\
	$k_{12} = 4.62143$ & $\rho_{12} = 0.08294$ & $\sigma_{12} = -0.02496$\\
	$k_{13} = 4.66667$ & $\rho_{13} = 0.07306$ & $\sigma_{13} = -0.00213$\\
	$k_{14} = 4.71667$ & $\rho_{14} = 0.07306$ & $\sigma_{14} = -0.00213$\\
	$k_{15} = 5.10639$ & $\rho_{15} = 0.07214$ & $\sigma_{15} = 0.00004$\\
	$k_{16} = 5.2174$ & $\rho_{16} = 0.06856$ & $\sigma_{16} = 0.00919$\\
	$k_{17} = 5.33334$ & $\rho_{17} = 0.06528$ & $\sigma_{17} = 0.01774$\\
	$k_{18} = 5.45455$ & $\rho_{18} = 0.06226$ & $\sigma_{18} = 0.02578$\\
	$k_{19} = 5.5$ & $\rho_{19} = 0.05947$ & $\sigma_{19} = 0.03339$\\
	$k_{20} = 5.55556$ & $\rho_{20} = 0.05693$ & $\sigma_{20} = 0.0404$\\
	$k_{21} = 5.625$ & $\rho_{21} = 0.05458$ & $\sigma_{21} = 0.0469$\\
	$k_{22} = 5.71429$ & $\rho_{22} = 0.05242$ & $\sigma_{22} = 0.053$\\
	$k_{23} = 5.83334$ & $\rho_{23} = 0.0504$ & $\sigma_{23} = 0.05877$\\
	$k_{24} = 6$ & $\rho_{24} = 0.0485$ & $\sigma_{24} = 0.06431$\\
	$k_{25} = 6.08696$ & $\rho_{25} = 0.03667$ & $\sigma_{25} = 0.09978$\\
	$k_{26} = 6.17648$ & $\rho_{26} = 0.0354$ & $\sigma_{26} = 0.10365$\\
	$k_{27} = 6.26866$ & $\rho_{27} = 0.03421$ & $\sigma_{27} = 0.10734$\\
	$k_{28} = 6.36364$ & $\rho_{28} = 0.03308$ & $\sigma_{28} = 0.11089$\\
	$k_{29} = 6.46154$ & $\rho_{29} = 0.032$ & $\sigma_{29} = 0.11429$\\
	$k_{30} = 6.5$ & $\rho_{30} = 0.03099$ & $\sigma_{30} = 0.11757$\\
	$k_{31} = 6.54546$ & $\rho_{31} = 0.03003$ & $\sigma_{31} = 0.12068$\\
	$k_{32} = 6.6$ & $\rho_{32} = 0.02913$ & $\sigma_{32} = 0.12362$\\
	$k_{33} = 6.66667$ & $\rho_{33} = 0.02828$ & $\sigma_{33} = 0.12643$\\
	$k_{34} = 6.76057$ & $\rho_{34} = 0.02747$ & $\sigma_{34} = 0.12913$\\
	$k_{35} = 6.85715$ & $\rho_{35} = 0.0267$ & $\sigma_{35} = 0.13174$\\
	$k_{36} = 7$ & $\rho_{36} = 0.02596$ & $\sigma_{36} = 0.13428$\\
	$k_{37} = 7.07369$ & $\rho_{37} = 0.01826$ & $\sigma_{37} = 0.16123$\\
	$k_{38} = 7.14894$ & $\rho_{38} = 0.01779$ & $\sigma_{38} = 0.16288$\\
	$k_{39} = 7.22581$ & $\rho_{39} = 0.01735$ & $\sigma_{39} = 0.16448$\\
	$k_{40} = 7.30435$ & $\rho_{40} = 0.01692$ & $\sigma_{40} = 0.16604$\\
	$k_{41} = 7.38462$ & $\rho_{41} = 0.0165$ & $\sigma_{41} = 0.16755$\\
	$k_{42} = 7.46667$ & $\rho_{42} = 0.0161$ & $\sigma_{42} = 0.16901$\\
	$k_{43} = 7.5$ & $\rho_{43} = 0.01572$ & $\sigma_{43} = 0.17044$\\
	$k_{44} = 7.53847$ & $\rho_{44} = 0.01536$ & $\sigma_{44} = 0.17181$\\
	$k_{45} = 7.58334$ & $\rho_{45} = 0.01501$ & $\sigma_{45} = 0.17314$\\
	$k_{46} = 7.63637$ & $\rho_{46} = 0.01467$ & $\sigma_{46} = 0.17441$\\
	$k_{47} = 7.71429$ & $\rho_{47} = 0.01435$ & $\sigma_{47} = 0.17565$\\
	$k_{48} = 7.79382$ & $\rho_{48} = 0.01403$ & $\sigma_{48} = 0.17685$\\
	$k_{49} = 7.875$ & $\rho_{49} = 0.01373$ & $\sigma_{49} = 0.17803$\\
	$k_{50} = 8$ & $\rho_{50} = 0.01344$ & $\sigma_{50} = 0.17918$\\
	$k_{51} = 8.064$ & $\rho_{51} = 0.00799$ & $\sigma_{51} = 0.20096$\\
	$k_{52} = 8.12904$ & $\rho_{52} = 0.00784$ & $\sigma_{52} = 0.20159$\\
	$k_{53} = 8.19513$ & $\rho_{53} = 0.00768$ & $\sigma_{53} = 0.20222$\\
	$k_{54} = 8.2623$ & $\rho_{54} = 0.00754$ & $\sigma_{54} = 0.20282$\\
	$k_{55} = 8.33058$ & $\rho_{55} = 0.00739$ & $\sigma_{55} = 0.20342$\\
	$k_{56} = 8.4$ & $\rho_{56} = 0.00725$ & $\sigma_{56} = 0.204$\\
	$k_{57} = 8.47059$ & $\rho_{57} = 0.00712$ & $\sigma_{57} = 0.20457$\\
	$k_{58} = 8.5$ & $\rho_{58} = 0.00698$ & $\sigma_{58} = 0.20513$\\
	$k_{59} = 8.53334$ & $\rho_{59} = 0.00686$ & $\sigma_{59} = 0.20567$\\
	$k_{60} = 8.57143$ & $\rho_{60} = 0.00673$ & $\sigma_{60} = 0.2062$\\
	$k_{61} = 8.61539$ & $\rho_{61} = 0.00661$ & $\sigma_{61} = 0.20671$\\
	$k_{62} = 8.68218$ & $\rho_{62} = 0.0065$ & $\sigma_{62} = 0.20721$\\
	$k_{63} = 8.75$ & $\rho_{63} = 0.00639$ & $\sigma_{63} = 0.20769$\\
	$k_{64} = 8.8189$ & $\rho_{64} = 0.00628$ & $\sigma_{64} = 0.20817$\\
	$k_{65} = 8.88889$ & $\rho_{65} = 0.00617$ & $\sigma_{65} = 0.20865$\\
	$k_{66} = 9$ & $\rho_{66} = 0.00607$ & $\sigma_{66} = 0.20911$\\
	$k_{67} = 9.05661$ & $\rho_{67} = 0.00199$ & $\sigma_{67} = 0.22746$\\
	$k_{68} = 9.11393$ & $\rho_{68} = 0.00195$ & $\sigma_{68} = 0.22761$\\
	$k_{69} = 9.17198$ & $\rho_{69} = 0.00192$ & $\sigma_{69} = 0.22774$\\
	$k_{70} = 9.23077$ & $\rho_{70} = 0.00189$ & $\sigma_{70} = 0.22788$\\
	$k_{71} = 9.29033$ & $\rho_{71} = 0.00187$ & $\sigma_{71} = 0.22801$\\
	$k_{72} = 9.35065$ & $\rho_{72} = 0.00184$ & $\sigma_{72} = 0.22815$\\
	$k_{73} = 9.41177$ & $\rho_{73} = 0.00181$ & $\sigma_{73} = 0.22828$\\
	$k_{74} = 9.47369$ & $\rho_{74} = 0.00178$ & $\sigma_{74} = 0.2284$\\
	$k_{75} = 9.5$ & $\rho_{75} = 0.00176$ & $\sigma_{75} = 0.22853$\\
	$k_{76} = 9.52942$ & $\rho_{76} = 0.00173$ & $\sigma_{76} = 0.22865$\\
	$k_{77} = 9.5625$ & $\rho_{77} = 0.00171$ & $\sigma_{77} = 0.22877$\\
	$k_{78} = 9.6$ & $\rho_{78} = 0.00168$ & $\sigma_{78} = 0.22889$\\
	$k_{79} = 9.65854$ & $\rho_{79} = 0.00166$ & $\sigma_{79} = 0.229$\\
	$k_{80} = 9.7178$ & $\rho_{80} = 0.00163$ & $\sigma_{80} = 0.22911$\\
	$k_{81} = 9.77778$ & $\rho_{81} = 0.00161$ & $\sigma_{81} = 0.22923$\\
	$k_{82} = 9.83851$ & $\rho_{82} = 0.00159$ & $\sigma_{82} = 0.22933$\\
	$k_{83} = 9.9$ & $\rho_{83} = 0.00157$ & $\sigma_{83} = 0.22944$\\
	$k_{84} = 10$ & $\rho_{84} = 0.00155$ & $\sigma_{84} = 0.22955$\\
	$k_{85} = \infty$ & $\rho_{85} = 0$ & $\sigma_{85} = 0.23725$\\
	\hline
\end{longtable}
\pagebreak 

\section{A potential function for general graphs}\label{app:gen-pot}

The following table contains an approximation of our pre-potential function for bipartite graphs, rounded away from zero to five decimal places. We give an exact version in the ancillary file \genprepot{}, available on the arXiv.

\begin{tabular}{|l|l|l|}
	\hline
	Boundary points $k_0,\dots,k_s$ & Edge weights $\rho_1,\dots,\rho_s$ & Vertex weights $\sigma_1,\dots,\sigma_s$\\
	\hline
    $k_0 = -1$ & & \\
	$k_{1} = 5.10639$ & $\rho_{1} = 0.13168$ & $\sigma_{1} = -0.13168$\\
	$k_{2} = 5.33334$ & $\rho_{2} = 0.11402$ & $\sigma_{2} = -0.08658$\\
	$k_{3} = 5.5$ & $\rho_{3} = 0.10004$ & $\sigma_{3} = -0.04931$\\
	$k_{4} = 5.625$ & $\rho_{4} = 0.08896$ & $\sigma_{4} = -0.01883$\\
	$k_{5} = 5.83334$ & $\rho_{5} = 0.08006$ & $\sigma_{5} = 0.0062$\\
	$k_{6} = 6$ & $\rho_{6} = 0.07243$ & $\sigma_{6} = 0.02845$\\
	$k_{7} = 6.08696$ & $\rho_{7} = 0.04523$ & $\sigma_{7} = 0.11007$\\
	$k_{8} = 6.26866$ & $\rho_{8} = 0.04157$ & $\sigma_{8} = 0.12119$\\
	$k_{9} = 6.46154$ & $\rho_{9} = 0.03833$ & $\sigma_{9} = 0.13135$\\
	$k_{10} = 6.54546$ & $\rho_{10} = 0.03542$ & $\sigma_{10} = 0.14076$\\
	$k_{11} = 6.66667$ & $\rho_{11} = 0.0329$ & $\sigma_{11} = 0.14901$\\
	$k_{12} = 6.85715$ & $\rho_{12} = 0.03066$ & $\sigma_{12} = 0.15648$\\
	$k_{13} = 7$ & $\rho_{13} = 0.0286$ & $\sigma_{13} = 0.16354$\\
	$k_{14} = 7.07369$ & $\rho_{14} = 0.0105$ & $\sigma_{14} = 0.22689$\\
	$k_{15} = 7.14894$ & $\rho_{15} = 0.00994$ & $\sigma_{15} = 0.22886$\\
	$k_{16} = 7.22581$ & $\rho_{16} = 0.00994$ & $\sigma_{16} = 0.22886$\\
	$k_{17} = 7.38462$ & $\rho_{17} = 0.00942$ & $\sigma_{17} = 0.23075$\\
	$k_{18} = 7.5$ & $\rho_{18} = 0.00892$ & $\sigma_{18} = 0.23258$\\
	$k_{19} = 7.58334$ & $\rho_{19} = 0.00846$ & $\sigma_{19} = 0.2343$\\
	$k_{20} = 7.63637$ & $\rho_{20} = 0.00804$ & $\sigma_{20} = 0.2359$\\
	$k_{21} = 7.71429$ & $\rho_{21} = 0.00804$ & $\sigma_{21} = 0.2359$\\
	$k_{22} = 7.875$ & $\rho_{22} = 0.00764$ & $\sigma_{22} = 0.23743$\\
	$k_{23} = 8$ & $\rho_{23} = 0.00727$ & $\sigma_{23} = 0.23891$\\
	$k_{24} = \infty$ & $\rho_{24} = 0$ & $\sigma_{24} = 0.26796$\\
	\hline
\end{tabular}

\end{document}